\documentclass[superscriptaddress,aps,pra,nofootinbib,notitlepage,longbibliography,10pt]{revtex4-1}
\usepackage[utf8]{inputenc}
\usepackage{mathrsfs}

\usepackage{amssymb,amsmath,amsthm,amsfonts,dsfont}

\newcommand{\ket}[1]{\left|#1\right\rangle}
\newcommand{\bra}[1]{\left\langle#1\right|}

\usepackage[colorlinks]{hyperref}
\hypersetup{
	pdfstartview={FitH},
	pdfnewwindow=true,
	colorlinks=true,
	linkcolor=blue,
	citecolor=blue,
	filecolor=blue,
	urlcolor=blue}
\usepackage[all]{hypcap}
\usepackage{cleveref}
\usepackage{url}
\usepackage{graphicx}
\usepackage{float}
\usepackage[ruled]{algorithm2e}
\usepackage{footnote}
\usepackage{textcomp}
\usepackage{gensymb}
\usepackage[margin=1in]{geometry}
\usepackage{comment}
\usepackage[normalem]{ulem}

\usepackage{tikz}
\usetikzlibrary{quantikz}

\usepackage{xcolor}
\usepackage{pgfplots}

\usepackage{mathtools}
\mathtoolsset{centercolon}
\usepackage{cleveref}

\DeclareMathOperator{\spec}{spec}

\newcommand{\ketbra}[2]{\ket{#1}\!\bra{#2}}

\newcommand{\floor}[1]{\lfloor{#1}\rfloor}

\newtheorem{theorem}{Theorem}
\newtheorem{definition}{Definition}
\newtheorem{lemma}{Lemma}
\newtheorem{proposition}{Proposition}

\newtheorem{corollary}{Corollary}

\newcommand{\eq}[1]{(\ref{eq:#1})}
\newcommand{\thm}[1]{\hyperref[thm:#1]{Theorem~\ref*{thm:#1}}}
\newcommand{\defn}[1]{\hyperref[defn:#1]{Definition~\ref*{defn:#1}}}
\newcommand{\lem}[1]{\hyperref[lem:#1]{Lemma~\ref*{lem:#1}}}
\newcommand{\prop}[1]{\hyperref[prop:#1]{Proposition~\ref*{prop:#1}}}
\newcommand{\fig}[1]{\hyperref[fig:#1]{Figure~\ref*{fig:#1}}}
\newcommand{\tab}[1]{\hyperref[tab:#1]{Table~\ref*{tab:#1}}}
\renewcommand{\sec}[1]{\hyperref[sec:#1]{Section~\ref*{sec:#1}}}
\newcommand{\subsec}[1]{\hyperref[subsec:#1]{Section~\ref*{subsec:#1}}}
\newcommand{\append}[1]{\hyperref[append:#1]{Appendix~\ref*{append:#1}}}
\newcommand{\cor}[1]{\hyperref[cor:#1]{Corollary~\ref*{cor:#1}}}
\newcommand{\obs}[1]{\hyperref[obs:#1]{Observation~\ref*{obs:#1}}}

\newcommand{\TBG}{\widetilde{B^2_G}}

\newcommand{\abser}{\alpha}

\newcommand{\nn}{\nonumber \\}

\newcommand{\edr}{Erd\H os-R\'enyi }

\newcommand{\Google}{\affiliation{
Google Quantum AI, Venice, CA 90291, United States}}

\newcommand{\MQ}{\affiliation{
School of Mathematical and Physical Sciences,
Macquarie University, Sydney, NSW 2109, Australia}}

\newcommand{\Toronto}{\affiliation{
Department of Computer Science, University of Toronto, ON M5S 2E4, Canada}}

\newcommand{\PNNL}{\affiliation{
Pacific Northwest National Laboratory, Richland, WA 99354, United States}}

\newcommand{\Leiden}{\affiliation{%
applied Quantum algorithms (aQa), Leiden University, 2300 RA Leiden, The Netherlands}}

\newcommand{\Caltech}{\affiliation{
Department of Computing and Mathematical Sciences, Caltech, Pasadena, CA 91125, United States}}

\allowdisplaybreaks

\begin{document}

\title{Analyzing Prospects for Quantum Advantage in Topological Data Analysis}

	\author{Dominic W.~Berry}
	\email{corresponding author: dominic.berry@mq.edu.au}
	\MQ

	\author{Yuan Su}
	\Google

	\author{Casper Gyurik}
    \Leiden
	
	\author{Robbie King}
	\Google
	\Caltech
	
    \author{Joao Basso}
    \Google
    
    \author{Alexander Del Toro Barba}
    \Google
	
	\author{Abhishek Rajput}
	\Toronto
	
	\author{Nathan Wiebe}
	\Toronto
	\PNNL
		
	\author{Vedran Dunjko}
	\Leiden

	\author{Ryan Babbush}
	\email{corresponding author: ryanbabbush@gmail.com}
	\Google

\date{\today}

	\begin{abstract}
    Lloyd \emph{et al.}~\cite{Lloyd2016} were first to demonstrate the promise of quantum algorithms for computing Betti numbers, a way to characterize topological features of data sets. Here, we propose, analyze, and optimize an improved quantum algorithm for topological data analysis (TDA) with reduced scaling, including a method for preparing Dicke states based on inequality testing, a more efficient amplitude estimation algorithm using Kaiser windows, and an optimal implementation of eigenvalue projectors based on Chebyshev polynomials. We compile our approach to a fault-tolerant gate set and estimate constant factors in the Toffoli complexity. Our analysis reveals that super-quadratic quantum speedups are only possible for this problem when targeting a multiplicative error approximation and the Betti number grows asymptotically. Further, we propose a dequantization of the quantum TDA algorithm that shows that having exponentially large dimension and Betti number are necessary, but insufficient conditions, for super-polynomial advantage. We then introduce and analyze specific problem examples which have parameters in the regime where super-polynomial advantages may be achieved, and argue that quantum circuits with tens of billions of Toffoli gates can solve seemingly classically intractable instances.   
	\end{abstract}
	
	\maketitle

\section{Introduction}
\label{sec:intro}

An important outstanding challenge in quantum computing is to find quantum algorithms that provide a significant speedup for practical problems.
One area of great interest is quantum machine learning \cite{Biamonte2017}.
Early proposals included, for example, principal component analysis \cite{Lloyd2014}, and were often based on quantum solution of linear equations \cite{HHL}.
However, it has proven possible to dequantize many of these proposals, indicating that there is at most a polynomial speedup \cite{tang:dequantization1,tang:dequantization2}.
Analysis of the cost taking into account error-correction overhead indicates that more than a quadratic speedup would be needed to provide a useful quantum advantage within quantum error-correction \cite{SandersPRQ20,QuadSpeedup}.

An algorithm for topological data analysis proposed by Lloyd \emph{et al.}~\cite{Lloyd2016} turned out not to be directly ``dequantizable'' using the same techniques, raising the question of whether a greater speedup was possible.
A simple analysis by Gunn \emph{et al.}~\cite{Gunn} contradicted some of the scaling results originally reported by Lloyd \emph{et al.}, and indicated that under certain assumptions there would still only be a quadratic speedup for these algorithms (our analysis agrees with that of Gunn \emph{et al.}).
Here we give a far more careful analysis of the complexity, and examine applications which provide better-than-quadratic speedups.

An important goal in data analysis is to extract features of a data set and use them to cluster or classify the data.
This data set would be represented as a set of points in some metric space, such as $\mathbb{R}^n$ with the Euclidean distance function.
One approach for the analysis is to convert the point cloud into a graph where the vertices are the given data points and the edges are determined by whether or not pairs of points lie within a chosen distance $\epsilon$. This approach can capture features such as connectivity but ignores potential higher dimensional features, especially if the data points are sampled from some underlying high-dimensional manifold.
Topological data analysis (TDA) attempts to extract such higher dimensional global topological features of an underlying data set by applying techniques from the field of algebraic topology, in particular what is known as simplicial homology.

A simplex is a point, line segment, triangle, or higher-dimensional equivalent, and a simplicial complex is a collection of simplexes.
One can form a simplicial complex from the data set with respect to a distance scale $\epsilon$, by adding points that are within distance $2\epsilon$ to simplices.
The Betti number $\beta_k$ is the number of $k$-dimensional holes of the complex.
One can determine the Betti number for a chosen range of $\epsilon$.
Betti numbers which persist over an appreciable range of the values of $\epsilon$ are indicative of intrinsic topological features of the data set, as opposed to artifacts that appear at a particular scale and disappear shortly thereafter. The study of such features is referred to as persistent homology.

The classical complexities of algorithms for estimating Betti numbers are typically exponential in $k$.
That means the computation can be intractable even for a moderate amount of data.
That is an important feature for the promise of quantum algorithms, because even fully error-corrected quantum computers with millions of physical qubits are expected to be very limited in data storage.
The most promising applications of quantum computers are therefore those involving a limited amount of classical data that needs to be fed into the quantum algorithm as part of the problem specification.

Recent work on quantum TDA algorithms introduced more efficient fermionic representations of the Dirac operator \cite{Cade} and employed the quantum singular value transformation to implement the kernel projector \cite{Ubaru,hayakawa2021quantum}.
Some of these techniques have led to significant asymptotic improvements over the original approach, but it is unclear whether they are useful for reducing the fault-tolerant implementation cost for solving problems in practice. Indeed, to the best of our knowledge, no study has been done on the fault-tolerant implementation of quantum TDA algorithms for solving \emph{any} instance of problems of practical interest.

In this work, we give a new algorithm for estimating Betti numbers on a quantum computer.
{We significantly reduce the cost of fault-tolerant implementation as compared to prior work, as well as estimating the constant factors that are needed to give realistic estimates of gate counts.} 
Specifically, we 
develop a new method to prepare the initial Dicke states,
introduce improved amplitude estimation using Kaiser windows,
directly construct the quantum walk operator from block encoding,
optimally project onto the kernel of the boundary map,  
then use the overlap estimation to estimate the kernel dimension of the block-encoded operator, leading to a quadratic improvement in precision over classical sampling.
We also provide the concrete constant factors in the complexity of our algorithm and estimate its fault-tolerant cost, going beyond the asymptotic analyses of all existing work on quantum topological data analysis. 
Finally, we show that it is possible to construct specific data sets {which have parameters in a regime where} quantum TDA would appear to have a significant speedup. 
In particular, we give examples of a very specific family of problem instances {exhibiting the required parameters} for the quantum TDA to have a superpolynomial speedup {over the naive general classical algorithm}, and a more general family of instances that {exhibits the required parameters} for a quartic speedup.
Here, we are comparing to well studied classical approaches with complexity approximately linear in the possible number of cliques.

We provide a more detailed explanation of the technical background needed to understand Betti numbers in \sec{introbackground}.
We then provide the improved algorithm and the analysis of its complexity in \sec{algor}.
We use this result to analyse the regimes where large quantum speedups may be expected in \sec{regimes_qspeedup}.
In particular, we consider cases where the Betti number would be large (implying a large quantum speedup) in \subsec{K-graphs}, and
novel competing classical algorithms in \subsec{classical}.
We then conclude in \sec{conc}.

\section{Technical background}
\label{sec:introbackground}
Here we give the more detailed background that is needed to understand the standard approaches for this problem and our contribution.
For the technical definitions of the simplicial complex and Betti number, see \append{background}.

\subsection{Overview of the TDA algorithm and its implementation}
\label{sec:alg_review}
In order to analyse the Betti numbers, the points and lines between the points are represented by a graph $G$.
Then a simplex is represented by a clique in the graph (groups of vertices that are all connected by edges).
The $n$ vertices of the graph $G$ are represented by $n$ qubits.
That is, $\ket{1}\ket{0} \cdots \ket{0}$ would represent the first vertex, and $\ket{0}\ket{0} \cdots \ket{1}\ket{0}$ would represent vertex $n-1$.
Note that this is a distinct representation from that often used to analyse sparse Hamiltonians, where each computational basis state represents a distinct vertex (so $n$ qubits would represent $2^n$ vertices).
In the representation here, a computational basis state with more than one $\ket{1}$ would represent a clique of the graph (a set of vertices with edges connecting every pair of vertices).
For example, $\ket{1}\ket{1}\ket{1}\ket{0} \cdots \ket{0}$ would represent a clique of the first three vertices.

The entire Hilbert space can then be subdivided into subspaces of different Hamming weights.
Using $\mathcal{H}_k$ to denote the space spanned by computational basis states with Hamming weight $k$, we have
\begin{equation}
    \left(\mathbb{C}^{2}\right)^{\otimes n}=\bigoplus_{k=0}^{n}\mathcal{H}_k,
\end{equation}
where $\dim(\mathcal{H}_k)=\binom{n}{k}$.
This space includes all states of the various Hamming weights.
One can also restrict to only states which represent vertices of cliques of the graph $G$.
Denoting by $\mathcal{H}_k^G$ the space spanned by basis states of all $k$-cliques of $G$, we have $\mathcal{H}_k^G\subseteq \mathcal{H}_k$.
We also use $\mathrm{Cl}_k(G)$ to denote the set of bit strings which correspond to $k$-cliques of $G$.

We define boundary maps $\partial_k:\mathcal{H}_{k+1}\rightarrow\mathcal{H}_{k}$ by their actions on the basis states $\ket{x}\in \mathcal{H}_{k+1}$ as
\begin{equation}
\label{eq:bndry}
    \partial_k\ket{x}:=\sum_{i=0}^{k}(-1)^i\ket{x\backslash(i)},
\end{equation}
where $x\backslash(i)$ means the $i$th $1$ in the bit string $x$ is set to $0$. We also define $\partial_k^G:\mathcal{H}_{k+1}^G\rightarrow\mathcal{H}_{k}^G$ as the restriction of $\partial_k$ to $\mathcal{H}_{k+1}^G$. 
{That is, it gives zero for any $x$ not representing a $k+1$-clique of $G$.}
By definition, we have that both $\mathrm{im}(\partial_{k+1}^G)$ and $\ker(\partial_k^G)$ are subspaces of $\mathcal{H}_k^G$. But in fact, we have $\mathrm{im}(\partial_{k+1}^G)\subseteq\ker(\partial_k^G)\subseteq\mathcal{H}_k^G$, which can be seen from (with $\ket{x}\in \mathcal{H}_k^G$)
\begin{align}
    \partial_k^G\partial_{k+1}^G\ket{x}
    &=\sum_{i=0}^{k+1}(-1)^i\partial_k^G\ket{x\backslash(i)}\nonumber\\
    &=\sum_{i=0}^{k+1}(-1)^i\sum_{j=0}^{i-1}(-1)^j\ket{x\backslash(j,i)}
    +\sum_{i=0}^{k+1}(-1)^i\sum_{j=i+1}^{k}(-1)^{j-1}\ket{x\backslash(i,j)}\nonumber\\
    &=\sum_{i=0}^{k+1}(-1)^i\sum_{j=0}^{i-1}(-1)^j\ket{x\backslash(j,i)}
    +\sum_{j=0}^{k+1}(-1)^{j-1}\sum_{i=0}^{j-1}(-1)^{i}\ket{x\backslash(i,j)} \nonumber\\
    &=0.
\end{align}
Since $\mathrm{im}(\partial_{k+1}^G)$ is a subspace of $\ker(\partial_k^G)$, one can define
the quotient space
\begin{equation}
    H_{k}(G):=\ker(\partial_{k}^G)/\mathrm{im}(\partial_{k+1}^G).
\end{equation}
This space is called the $k$th homology group, and its dimension
\begin{equation}
    \beta_{k}^G:=\dim(H_{k}(G))=\dim(\ker(\partial_{k}^G))-\dim(\mathrm{im}(\partial_{k+1}^G))
\end{equation}
is the $k$th Betti number.
In practice, Betti numbers $\beta_k^G$ can be used to extract features of the shape of the data modeled by the graph $G$, and their estimation is the main problem in the topological data analysis we will consider here.
In this work we will be estimating $\beta_{k-1}^G$ for the $k-1$th Betti number, so we can simplify our discussion by considering Hamming weight $k$.

To describe our quantum algorithm and its circuit implementation for estimating Betti numbers, we will introduce the Dirac operator $B_G$. Specifically, for any graph $G$ and a fixed value of $k$, we define
\begin{equation}\label{eq:bdef}
B_G :=
\begin{bmatrix}
    0 & \partial_{k-1}^G & 0\\
    \partial_{k-1}^{G\dagger} & 0 & \partial_{k}^G\\
    0 & \partial_{k}^{G\dagger} & 0
\end{bmatrix},
\end{equation}
where the blocks indicate the subspaces $\mathcal{H}_{k-1}^G$, $\mathcal{H}_{k}^G$ and $\mathcal{H}_{k+1}^G$.
Since $\partial_{k-1}^G\partial_{k}^G$ gives zero, squaring $B_G$ yields
\begin{equation}
B_G^2=
\begin{bmatrix}
    \partial_{k-1}^G \partial_{k-1}^{G\dagger} & 0 & 0 \\
    0 & \partial_{k-1}^{G\dagger} \partial_{k-1}^{G} + \partial_{k}^{G}\partial_{k}^{G\dagger} & 0 \\
    0 & 0 & \partial_{k}^{G\dagger}\partial_{k}^{G}
\end{bmatrix}.
\end{equation}

It can be seen here that the middle part corresponds to the combinatorial Laplacian \cite{Eckmann1944}
\begin{equation}
    \Delta_{k-1}^G = \partial_{k-1}^{G\dagger} \partial_{k-1}^{G} + \partial_{k}^{G}\partial_{k}^{G\dagger}.
\end{equation}
It is known that \cite{Eckmann1944,friedman:betti,Casper}
\begin{equation}
    \dim(\ker(\Delta_{k-1}^G)) = \beta_{k-1}^G,
\end{equation}
which provides a convenient way of computing Betti numbers.
Since $B_G$ is Hermitian, the kernel of $B_G$ and $B_G^2$ is identical.
Therefore, to estimate the Betti number corresponding to a particular graph and a fixed value of $k$, it suffices to construct the Dirac operator and compute the dimension of its kernel on the subspace $\mathcal{H}_{k}^G$.

It can be difficult in general to perform topological data analysis on a classical computer due to the high-dimensional nature of the problem, with the dimension increasing exponentially in $k$.
However, the Dirac operator could be efficiently simulated on a quantum computer, in the sense of solving the Schr\"odinger equation with the Dirac operator as the Hamiltonian.
That indicates exponential speedups are possible, though there are a number of other stages needed for the quantum algorithm.
Previous work has provided several approaches for estimating Betti numbers on quantum computers.
The stages of these approaches include preparation of a uniformly mixed state, construction of the projector onto the kernel subspace, and estimation of the overlap.

The original approach of Lloyd et al.\ \cite{Lloyd2016} applied amplitude amplification and estimation to prepare the desired initial state in $\mathcal{H}_k^G$, starting from a uniform mixture of all Hamming weight $k$ basis states in $\mathcal{H}_k$.
Their approach actually produces a superposition over all values of $k$, so the success probability of obtaining a specific $k$ can be quite low; this issue was addressed by later work such as \cite{Gunn,Casper}.
To construct the projector onto the kernel, they implement Hamiltonian simulation and perform quantum phase estimation on the resulting operator. The Betti number is then estimated as the frequency of zero eigenvalues in the measurement.
That is, the algorithm can be summarised as below.
 \begin{enumerate}
     \item For $i=1,\ldots,m$, repeat:
     \begin{enumerate}
         \item Prepare the mixed state
         \begin{equation}
             \rho_k^G=\frac{1}{\dim(\mathcal{H}_k^G)}\sum_{x\in\mathrm{Cl}_k(G)}\ketbra{x}{x}.
         \end{equation}
         \item Apply quantum phase estimation to the unitary $e^{iB_G t}$.
         \item Measure the eigenvalue register to obtain an approximation $\widetilde{\lambda}_i$.
     \end{enumerate}
     \item Output the frequency of zero eigenvalues:
     \begin{equation}
         \frac{\#\{i,\widetilde{\lambda}_i=0\}}{m}.
     \end{equation}
 \end{enumerate}

In this work, we give a new algorithm for estimating Betti numbers on a quantum computer.
We provide a number of improvements which {significantly reduce the cost of fault-tolerant implementation.} Specifically, we do the following.
\begin{itemize}
    \item Develop new methods to prepare a mixture of fixed Hamming-weight states with garbage information that have significantly lower fault-tolerant cost.
    \item Introduce improved amplitude estimation using Kaiser windows to estimate the number of steps of amplitude estimation needed.
    \item Directly construct the quantum walk operator from block encoding without an additional step of quantum simulation.
    \item Project onto the kernel of the boundary map by implementing a Chebyshev polynomial to optimally filter the zero eigenvalues. This is more efficient than previous approaches that implement the phase estimation or rectangular window functions for filtering.    
    \item Use the overlap estimation to estimate the kernel dimension of the block-encoded operator, leading to a quadratic improvement in precision over the classical sampling approach used by previous work.
\end{itemize}
We also provide the concrete constant factors in the complexity of our algorithm and estimate its fault-tolerant cost for solving example problems, going beyond the asymptotic analyses of all prior work on quantum topological data analysis.

Ultimately, the performance of our algorithm (as well as other algorithms from previous work) will depend on several important problem parameters.
First, the desired state on which we perform the kernel projector is a uniform mixture of all the $|\mathrm{Cl}_k(G)|$ basis states in $\mathcal{H}_k^G$, whereas we start with a uniform mixture of all $\binom{n}{k}$ basis states in $\mathcal{H}_k$.
Their ratio $\binom{n}{k}/|\mathrm{Cl}_k(G)|$ will determine the number of amplification steps required in the state preparation. 
There is potential to improve the efficiency of preparation of the cliques via an improved clique-finding algorithm.

Second, we need to implement a spectral projector that distinguishes the zero eigenvalue from the remaining nonzero eigenvalues of the Dirac operator, and the cost of implementing such a projector will depend on the spectral gap of the Dirac operator.
Third, the output of the quantum TDA algorithm will not be the actual Betti number but instead a normalized version of the quantity $\beta_{k-1}^G/|\mathrm{Cl}_k(G)|$.
In order to estimate the Betti number to some additive precision, we need to increase the complexity by a factor that depends on $|\mathrm{Cl}_k(G)|$, with the result that the complexity would roughly scale as $\sqrt{\binom{n}{k}}$.

An alternative scenario is that a fixed relative error is required; that is, the ratio of the uncertainty in the Betti number to the Betti number.
Then the complexity would roughly scale as $\sqrt{\binom{n}{k}/\beta_{k-1}^G}$, as we show in \sec{algor}.
This means that significant speedups can be provided in cases where the Betti number $\beta_{k-1}^G$ is large, and we provide examples of such graphs in \sec{regimes_qspeedup}.

{
Our overall complexity may be summarised as in the following lemma.
\begin{lemma}[Total complexity]\label{lem:complex}
    The complexity of estimating to relative error $r$ the Betti number $\beta^G_{k-1}$ of graph $G$ with $n$ vertices may be approximated as, for two different methods
    \begin{align}\label{eq:method1}
 &  \frac{\ln(1/\delta_2)}{r_2}\sqrt{\frac{|{\rm Cl}_k(G)|}{\beta^G_{k-1}}} \left[ \frac {\pi}{2} \sqrt{\frac{\binom{n}{k}}{|{\rm Cl}_k(G)|}}  (6|E| + n\log^2 n) + \frac{n}{\lambda_{\min}} \ln\left(\frac{4|{\rm Cl}_k(G)|}{r_3\beta^G_{k-1}}\right) ( 6|E| + 5n ) \right] \, , \\
 \label{eq:method2}
  &    \frac{\ln(1/\delta_2)}{r_2}\sqrt{\frac{|{\rm Cl}_k(G)|}{\beta^G_{k-1}}} \left[ \frac {\pi}{2} \sqrt{\frac{n^k/k!}{|{\rm Cl}_k(G)|}}  (6|E| + 2kn) + \frac{n}{\lambda_{\min}} \ln\left(\frac{4|{\rm Cl}_k(G)|}{r_3\beta^G_{k-1}}\right) ( 6|E| + 5n ) \right] \, ,
\end{align}
with probability of failure $\delta=\delta_1+\delta_2$, 
where $r=r_1+r_2+r_3$, $|{\rm Cl}_k(G)|$ is the number of $k$-cliques, and it is assumed we are given a classical database of edges of the graph.
In the case where we are instead given a database of missing edges, then $6|E|$ is replaced with $4|E^C|$ in the above expressions.
\end{lemma}

See Section \ref{sec:total} for the explanation of this total complexity.
}

\subsection{Complexity classes of TDA}
\label{sec:complexity}

Linear-algebraic applications of quantum computing have led to numerous suggestions of how various types of machine learning subroutines could be implemented on a quantum computer with superpolynomial speed-ups over their classical counterparts. 
Many of these methods were in the end shown to only suffice for at most polynomial speed-ups, due to the randomized ``dequantizations'' of Tang and others~\cite{tang:dequantization1, tang:dequantization2, chia:dequantization_overview}. 
The algorithm of Lloyd et al.~\cite{Lloyd2016}, however, turned out not to be directly ``dequantizable'' using similar techniques, raising the question of whether more robust complexity-theoretic quantum-classical separations can be proven.
The current landscape on this topic is somewhat involved.

In general, we have a number of discrepancies between the computational problems in ordinary TDA applications and the computational problems for which we have certain complexity-theoretic insights.
In ordinary TDA applications one is typically concerned with the computation of the exact count of zero eigenvalues of combinatorial Laplacians. 
By the result of~\cite{marcos}
-- which shows that deciding if a combinatorial Laplacian has a trivial or non-trivial kernel (i.e., Betti number zero or non-zero) is $\mathsf{QMA}_1$-hard -- this problem is likely beyond what is efficient even for quantum computers in the worst case.
This observation goes in line with classical bodies of work showing that exact computations of Betti numbers is $\mathsf{NP}$-hard~\cite{adamaszek:complexity_betti}, and that it can even be $\mathsf{PSPACE}$-hard for more involved topological spaces (i.e., so-called \textit{algebraic-varieties})~\cite{scheiblechner:complexity_betti}.

From the perspective of the types of problems quantum algorithms may be efficient for, one could attempt a few relaxation of the problem. 
First, it may be fruitful to relax the TDA problem with respect to the quantity estimated.
Specifically, instead of the number of exactly zero eigenvalues, one could relax it and count the number of ``small'' eigenvalues (i.e., below a threshold).
This relaxation may be convenient from a quantum algorithmic perspective, but it also still useful from a data analysis perspective, since Cheeger's inequality demonstrates that the magnitudes of the small non-zero eigenvalues of the graph Laplacian characterises the connectedness of the graph~\cite{mohar:discrete_cheeger}, and similar results hold for combinatorial Laplacians~\cite{gundert:cheeger}.
In folklore it is conjectured that for difficult cases, the magnitude of the smallest non-zero eigenvalue of combinatorial Laplacians very often scales inverse polynomially~\cite{friedman:betti}, in which case the number of ``small'' eigenvalues coincides with the number of zero eigenvalues if the threshold is chosen appropriately. 
While the problem of counting small eigenvalues is more suitable to be solved on a quantum computer, it could turn out to still be $\mathsf{QMA}_1$-hard if the TDA matrices have a sufficiently large spectral gap.
Specifically, if the TDA matrices are sufficiently gapped, then one could count the number of zero eigenvalues (which is $\mathsf{QMA}_1$-hard~\cite{marcos}) by counting the number of eigenvalues below the spectral gap (i.e., the number of ``small'' eigenvalues). 

A related (yet different) problem for which complexity-theoretical results are known is that of estimating \textit{normalized} Betti numbers to within additive inverse polynomial precision. 
That is, the number of zero eigenvalues divided by the total number of eigenvalues, which here would be $\beta_{k-1}^G/|\mathrm{Cl}_k(G)|$ (if the TDA matrix is sufficiently gapped).
This quantity is natural from a quantum computational complexity perspective (though not from an applications perspective), since a quantum algorithm naturally estimates probabilities (so normalized quantities in this case), and since additive errors allow for a direct relationship to definitions of complexity classes like $\mathsf{DQC1}$.

Specifically, in~\cite{Casper} it was shown that the generalization of this problem, namely estimating the ratio when allowing a range of small eigenvalues, rather than strictly zero eigenvalues, for \textit{arbitrary} Hermitian operators (i.e., the so-called \emph{low-lying spectral density}) is $\mathsf{DQC1}$-hard.
This result was build upon in~\cite{Cade}, where it was shown that the problem remains $\mathsf{DQC1}$-hard when restricting the input to combinatorial Laplacians of general chain complexes.
It is unknown whether the hardness persists when further restricting to combinatorial Laplacians of clique-complexes, and the closest result to this is the $\mathsf{QMA}_1$-hardness result of~\cite{marcos} for the problem of exact counting. 
This normalized quantity is not typically studied, and indeed there are concerns that Betti numbers may fail to be large enough to be detectable when normalized (see also \append{proof_betti_density}).

As discussed above, estimates of the normalized Betti number with additive error are more natural from a quantum computational complexity perspective.
However, from the perspective of applications, we typically work with (unnormalized) Betti numbers (and perhaps their estimates). 
For this case, the rescaling from normalized Betti numbers to Betti numbers causes an in general exponential blow up of additive errors, and leads to algorithms which always have exponential run-times (for constant error). 
At the same time, in many applications, we only require small additive errors when the quantities in question are themselves small. 
For these reasons here we focus on estimation to within a given relative error; that is, the error in the Betti number divided by the Betti number. 
That is immune to rescaling and can lead to efficient algorithms in the cases when the Betti numbers are large.

Note that the problem of estimating the low-lying spectral density up to a certain relative error is also $\mathsf{DQC1}$-hard.
The reason is that the relative error must always be at least as large as the error in the normalised quantity, and estimating the normalised quantity to additive precision $\epsilon$ is $\mathsf{DQC1}$-hard for $\epsilon = 1/\mathrm{poly}(n)$~\cite{Casper}.
It is unknown whether the hardness result holds for Betti numbers, because they are found by restricting to combinatorial Laplacians, rather than arbitrary Hermitian operators.
Generally, the larger the Betti number the more efficient the quantum algorithm will be, which in certain cases results in a polynomial quantum runtime.
Examples of cases where the Betti numbers are large are discussed in more detail in \sec{regimes_qspeedup}.

\section{Optimization and analysis of quantum topological data analysis}
\label{sec:algor}

{
In this section we describe our algorithm in detail.
Subsections \ref{subsec:sorting}, \ref{sec:clique}, and \ref{sec:amplifying} are for preparing a state similar to $\rho_k^G$ in prior work.
That is a combination of states $\ket{x}$ that correspond to $k$-cliques of the graph $G$.
The general principle is to first prepare a Dicke state, which is explained in Subsection \ref{subsec:sorting}.
That is an equal superposition of all states with $k$ ones, of which only a subset will be $k$-cliques.
Therefore, Subsection \ref{sec:clique} then describes how to efficiently detect the states out of those that are $k$-cliques.

In order to provide a further speedup we then use amplitude amplification, as described in Subsection \ref{sec:amplifying}.
The key difficulty there is that the number of steps of amplitude amplification depends on the amplitude.
We therefore use amplitude estimation separately from the amplitude amplification.
This provides a significant advantage over fixed-point amplitude amplification \cite{YoderPRL14}, which incurs a logarithmic overhead, because we only need to perform the estimation once but perform the amplification many times.
Moreover, our amplitude estimation technique using Kaiser windows is improved over standard amplitude amplification, and can be used in far more general applications.

Then in Subsection \ref{sec:block_encode} we describe how to block encode the operator $B_G$ in order to provide a quantum walk operator that has eigenvalues related to those of $B_G$, as in \cite{Low2019hamiltonian,BerryNPJ18}.
In particular, we need to find eigenvalues $\pm 1$ of this walk operator, which correspond to eigenvalue 0 for $B_G$.
This provides a significant advantage over prior work that was based on simulating a Hamiltonian time evolution under $B_G$, because we avoid the overheads inherent in simulating Hamiltonian evolution.

Next, in Subsection \ref{sec:poly_approx} we show how to use an optimal filter to find eigenstates of $B_G$ with eigenvalue 0 (or $\pm 1$ for the walk operator).
This method improves over prior work which used phase estimation, which has an overhead due to it providing more information (an estimate rather than just distinguishing between zero and nonzero eigenvalues). 

Finally, in Subsection \ref{sec:total} we use the amplitude estimation again to estimate the proportion of zero eigenvalues, then provide the overall complexity.
The use of amplitude estimation here provides a square-root speedup over work based on classical sampling.
}

\subsection{Generating Dicke states with garbage}
\label{subsec:sorting}
In this section, we consider preparing an $n$-qubit uniform superposition of Hamming weight $k$ basis states (which is allowed to be entangled with garbage states). Such a state is known in previous literature as the Dicke state. 
In \cite{Dicke} it was shown how to prepare a Dicke state with ${\mathcal O}(nk)$ gates, although these gates included rotations, so there would be a logarithmic factor in the complexity when counting non-Clifford gates.

Because the preparation here allows an entangled state to be prepared between the superposition for the Dicke state and ancilla states, it is possible to prepare the state more efficiently.
One approach is to apply a quantum sort to $n$ registers, then use it to apply an inverse sort to the $n$ qubits with the first $k$ set in the state $\ket{1}$.
This is similar to the approach used for symmetrising states for chemistry in \cite{BerryNPJ18}.
Another approach is to use inequality testing to obtain $k$ successes.
Both those approaches give a factor of $\log n$ in the number of qubits required, which is costly when $n$ is large.
Throughout we use ``$\log$'' for base 2, and ``$\ln$'' for natural logs.

{We provide two schemes here.}
In {our first scheme} we prepare $n$ registers with approximately $\log n$ qubits in equal superposition.
This is similar to the first step in \cite{BerryNPJ18} where a sort was used.
Here we instead find a threshold such that $k$ registers are less than or equal to this threshold.
{This approach is explained in detail in \append{Dicke1}.

Our second scheme is based on preparing separate superposition states for the positions of each of the individual ones.
The target state is obtained by adding all those ones into the target state, but has a higher amplitude for failure arising from ones in the same locations.
We provide the details in \append{Dicke2}.
The complexities of these two schemes are as in the following lemma.

\begin{lemma}[Dicke preparation]\label{lem:dicke}
    The Dicke state with $k$ ones in $n$ qubits may be prepared with probability of success 
    \begin{equation}
    \frac 1{(cn)^n} \binom{n}{k} \sum_{\ell=1}^{cn} [\ell^k-(\ell-1)^k] (cn-\ell)^{n-k} \approx 1-\frac 1{2c} \, .
\end{equation}
using 
    \begin{equation}\label{eq:dickeprep}
    (n_{\rm seed}+1)\left[\frac n2  (n_{\rm seed}+2)+ \lceil \log n\rceil\right] 
\end{equation}
Toffolis, where $n_{\rm seed}$ is a number of seed qubits
\begin{equation}
    n_{\rm seed}:=\lceil\log cn\rceil \, ,
\end{equation}
for some constant $c$.
Alternatively it may be prepared with probability of success
\begin{equation}
    \frac{k!}{n^k} \binom{n}{k}
\end{equation}
with Toffoli complexity
\begin{equation}
    (k+2)n + k(4\lceil \log n\rceil-1) + \lceil \log k\rceil \, ,
\end{equation}
or 
\begin{equation}
    (k+2)n - 2k + \lceil \log k\rceil \, ,
\end{equation}
for $n$ a power of 2.
For this preparation, the state may be entangled with an ancilla system.
\end{lemma}

Although the first scheme has better asymptotic complexity of $\widetilde{\mathcal{O}}(n)$, we find that for realistic parameters its complexity is considerably larger.
The lower probability of success of the second scheme results in a larger factor in the complexity, so the approach that is optimal will depend on the parameters.
}

{
In comparison, prior work in Refs.~\cite{Gunn,Casper} used a procedure based on an integer enumeration of all basis states for the Dicke state.
Lloyd \emph{et al.}~\cite{Lloyd2016} used a method with a superposition over values of $k$ that is not directly comparable.
The complexity in Ref.~\cite{Gunn} does not appear correct (the complexity in Ref.~\cite{Casper} just cites that result).
The method it uses is to first compute a Pascal's triangle of binomial coefficients up to $\binom{n}{k}$ with complexity $\widetilde{\mathcal{O}}(n^2 k)$.

To convert a natural number $l$ to a Hamming-weight $k$ string, it then starts by finding the largest value of $x$ such that $\binom{x}{k}<l$.
It is said that the value of $x$ can be found using $\widetilde{\mathcal{O}}(k)$ gates via a binary search using the Pascal's triangle as a lookup table.
The complexity is given as the number of stpes in the binary search, which is not correct.
The reason is that $l$ is given in quantum superposition, so $x$ needs to be searched for in superposition, and finding the appropriate entry in the lookup table (to perform the inequality test for the binary search) has complexity of the size of the lookup table.
The value of $k$ is fixed so not the entire lookup table is needed, but there are $n$ entries needed, and each has size $\mathcal{O}(k\log n)$.
This has a complexity of $\mathcal{O}(kn\log n)$, which then needs to be performed $\mathcal{O}(\log n)$ times in the binary search, so would give a complexity $\mathcal{O}(kn\log^2 n)$.
That complexity needs to be multiplied by $k$ steps of the algorithm to give overall complexity $\mathcal{O}(k^2 n\log^2 n)$ for the conversion.

The complexity of converting in the opposite direction, from a Hamming-weight $k$ string to a natural number, is given correctly in Ref.~\cite{Gunn} as $\widetilde{\mathcal{O}}(nk)$.
The complexity of the Pascal's triangle is somewhat less than that given in Ref.~\cite{Gunn}.
It can be calculated classically and entered into a quantum registers with $\mathcal{O}(n k \log n)$ Clifford gates, so zero Toffoli complexity.
The complexity of preparing an equal superposition over natural numbers is $\mathcal{O}(k\log n)$.
That is omitted in Ref.~\cite{Gunn}, which is reasonable because it is smaller than the other complexities.

So in Ref.~\cite{Gunn}, the leading order complexity of the Dicke state preparation is $\mathcal{O}(k^2 n\log^2 n)$, which is a factor of $k^2$ larger than our complexity of $\mathcal{O}(n\log^2 n)$ for our first approach, and a factor of $\widetilde{\mathcal{O}}(k)$ larger than the complexity of Dicke state preparation in Ref.~\cite{Dicke}.
In the example we give in Section \ref{subsec:K-graphs}, $k=16$ so the factor of $k^2$ is 256, and our first approach has about two orders of magnitude improvement in the complexity of this step as compared to Ref.~\cite{Gunn}.
For that example there is about another factor of 4 improvement by using our second approach.
}

\subsection{Detecting the cliques}
\label{sec:clique}

In the previous section, we have discussed the preparation of the $n$-qubit Dicke state with Hamming weight $k$ (and an additional garbage register)
\begin{equation}
    \frac{1}{\sqrt{\binom{n}{k}}}\sum_{|x|=k}\ket{x}\frac{1}{\sqrt{k!(n-k)!}}\sum_{\sigma(0\cdots01\cdots1)=x}\ket{\sigma}.
\end{equation}
Here, the first register holds all $n$-qubit strings with Hamming weight $k$, representing subsets of $k$ vertices in an $n$-vertex graph $G$. We now describe a quantum circuit that detects whether a given string $x$ represents a $k$-clique in the underlying graph, with the promise that $x$ has Hamming weight $k$. Specifically, our goal is to implement the mapping
\begin{equation}
    \ket{x}\ket{0}\ket{0}\mapsto\ket{x}\ket{x\in\mathrm{Cl}_{k}(G)}\ket{\mathrm{garb}_x}.
\end{equation}
Here, the second register has value $1$ if $x$ represents a $k$-clique in $G$ and $0$ otherwise. The third register contains some garbage information $\mathrm{garb}_x$ that can depend on $x$ and need not be uncomputed.

Our implementation of the clique detection is related to the approach of \cite{Metwalli}. Specifically, we introduce a register of $\floor{\log\binom{k}{2}}+1\le 2\log k$ qubits to represent integers $0,\ldots,\binom{k}{2}$. This register will be used to count the number of edges in the subgraph induced by the $k$ vertices denoted by $x$.
For the graph, we assume that it is given by a classical database, so we need to run through this classical database, rather than assuming any oracular access to the graph.
Let us assume that we have a listing of all edges in the graph.
That is, for each edge, we have a listing of the two nodes.
In order to implement this classical data, for each edge in the list we use a Toffoli with the qubits representing those two nodes as controls, and an ancilla as target.
In the case where both qubits are in the state 1, the ancilla qubit will be flipped.

The complexity is then given by a number of Toffolis equal to the number of edges, which we denote $|E|$.
We aim to sum all the bits output by these Toffolis.
Provided that we are restricted to Hamming weight $k$, if $x$ represents a $k$-clique then every pair of ones in $x$ will result in a 1, so the sum will yield $\binom{k}{2}$.
Summing bits in the obvious way would yield a complexity scaling as $2|E|\log k$ Toffolis, because each addition requires multiple Toffolis.
An improved method is given in \cite{Kivlichan2020improvedfault}, where it would take no more than $|E|$ Toffolis, but the same number of ancillas would be required, which would typically be a prohibitively large cost.
An alternative way of summing bits is given in \cite{SandersPRQ20}, where multiple groups of bits are summed, and their sums are summed.
The overall complexity is no more than $2|E|$ Toffolis, and only a logarithmic number of ancillas is used. The costs of the three main parts of the algorithm are as follows.
\begin{enumerate}
    \item There is cost $|E|$ Toffolis for checking the edges of the graph.
    The resulting qubits can be erased with measurements and phase corrections, with zero Toffoli cost.
    \item The complexity of the efficient bit sum approach from \cite{SandersPRQ20} is $2|E|$.
    \item There is complexity no more than $2\log k$ Toffolis to check that the output register is equal to $\binom{k}{2}$.
\end{enumerate}
Therefore, the total cost of clique detection is no more than $3|E|+2\log k$ Toffolis.
In many cases we will need to reflect on the result of this test.
In that case the $2\log k$ cost is not doubled, because we can replace the equality test with a controlled phase.
Therefore the cost in that case is $6|E|+2\log k$.
If we were to retain the qubits resulting from the edge checking and use the sum from \cite{Kivlichan2020improvedfault}, the cost would be $2|E|+2\log k$, though with a large ancilla cost.

Later when we consider the block encoding of the Hamiltonian we will need to allow a wider range of Hamming weights, $k-1$, $k$ and $k+1$,
in the case where we are block encoding the Hamiltonian projected onto this subspace.
First we can sum the ones in the string $x$, which has Toffoli complexity $n$.
We can check if the sum is equal to $k-1$ with $\lceil \log n\rceil$ Toffolis, then check if it is $k$ or $k+1$ with further Toffolis with the unary iteration procedure.
The number of Toffolis needed depends on the value of $k$, and some values will require about another $\lceil \log n\rceil$ Toffolis.
For each we can use CNOTs to output a success flag on an ancilla qubit.

We can also output the value of $\binom{k-1}{2}$, $\binom{k}{2}$, or $\binom{k+1}{2}$ in another register.
In this case we would need to apply an equality test between the result in our sum register and the result in this register, which again has a Toffoli complexity no larger than $2\log k$.
There are also $n$ Toffolis needed to sum the ones in $x$ and no more than $3\lceil \log n\rceil$ Toffolis to check the number of ones.

Much of the complexity of the algorithm is due to the use of amplitude amplification to find the cliques.
There has been much work on quantum algorithms for clique finding, but these algorithms are typically posed in terms of calls to an oracle for the graph, with a possibly large complexity for additional gates.
What that means is that the complexity in terms of oracle calls is no more than $\mathcal{O}(n^2)$ to find all the edges, and then there can be a very large amount of postprocessing to find the cliques.

{
When there are $n$ vertices there cannot be any more than $n(n-1)/2$ edges, but in practice we would only use the above algorithm for $|E|$ less than half this.
The reason is that for larger numbers of edges, it is more efficient to use a database of \emph{missing} edges.
If we use such a database, we can then iterate through all pairs of nodes in the database and use a Toffoli controlled on the corresponding qubits.
If we find any cases where this gives one, it means that there is a missing edge and the state does not represent a clique.
We therefore need to perform an OR on all the resulting qubits.

Similarly to the case above using the list of edges, one could perform an approach using addition and obtain a complexity with $|E|$ replaced by $|E^C|$, where $E^C$ is the set of missing edges.
But, since we only need to find a single one out of all the results, we can instead use an approach for a multiply-controlled Toffoli with a limited number of ancillas.
If one is willing to use about $\sqrt{|E^C|}$ ancillae, then the cost is $2|E^C|$, with the same cost for erasure.

In particular, consider grouping the list of missing edges into sets of approximately $\sqrt{|E^C|}$.
For each we can perform the Toffolis with the qubits representing the corresponding nodes as controls, then perform a multiply-controlled Toffoli with those qubits as controls (with the appropriate bit flips to give an OR).
The multiply-controlled Toffoli has Toffoli cost one less than the number of qubits in the group.
So, for example if $\sqrt{|E^C|}$ is an integer, then it is cost $\sqrt{|E^C|}-1$.
The qubits giving the results of the Toffolis for the missing edges can be erased using Clifford gates similarly as before.

We do this for each of the approximately $\sqrt{|E^C|}$ groups.
Then we perform a multiply-controlled Toffoli on the results for all these groups.
Because there are about $\sqrt{|E^C|}$ this is again the Toffoli and ancilla cost.
For example, if $\sqrt{|E^C|}$ is an integer, then there would be Toffoli cost $(\sqrt{|E^C|}-1)\sqrt{|E^C|}$ for the $\sqrt{|E^C|}$ groups, followed by $\sqrt{|E^C|}-1$ for the multiply-controlled Toffoli on the results, for cost $|E^C|-1$.
That is combined with the $|E^C|$ cost of the Toffolis for the individual missing edges, for total cost $2|E^C|-1$.
Given the improved efficiency of this approach, it would be preferable to using the list of edges for $|E|$ above about $n^2/5$.
We can therefore summarise the complexity of the clique checking as follows.

\begin{lemma}[Clique checking]\label{lem:cliques}
    It can be checked that an $n$-qubit state corresponds to a clique of graph $G$ using $3|E|+2\log k$ Toffolis given a classical database of edges $E$, or $2|E^C|$ given a classical database of missing edges $E^C$.
    The costs of reflecting on the result of the clique check are $6|E|+2\log k$ or $4|E^C|$ Toffolis.
\end{lemma}

The complexity in Ref.~\cite{Lloyd2016} is given as $\mathcal{O}(k^2)$ in terms of calls to an oracle for the distances between the points.
Similar complexities are given in Refs.~\cite{Gunn,Casper} but are not explained, so they are likely using the same assumption as Ref.~\cite{Lloyd2016}.
That complexity is not directly comparable to the result here, because we are instead assuming that we are given an explicit listing of edges (or missing edges).

In order to provide a comparison between that approach and ours, we can consider a slight modification of that where a database of locations of points is given.
In that case, one can use a quantum sort on the Dicke state to also sort the locations of the nodes to the first $k$ data locations.
That sort has complexity $\mathcal{O}(n\log n \log(1/\epsilon))$ given that the locations are given to accuracy $\epsilon$ (relative to the range of positions).
Then the distances can be checked with complexity $\mathcal{O}(k^2\log^2(1/\epsilon))$.
The factor of the square of the log here is from the complexity of performing squares for determining the (squared) distance.

We consider an example of a graph in Section \ref{subsec:K-graphs} with $n=256$, $k=16$ and $|E|=30720$.
In that example, $|E^C|=1920$, so it is more efficient to use the list of missing edges in our approach.
In comparison, if one were to attempt to use the database of locations, then just the sort would have higher complexity than $2|E^C|$.
If edges were determined from positions given in a three-dimensional space, then an accuracy of the components of the locations of only 4 bits would result in complexity larger than our costing by an order of magnitude.
Although it is difficult to compare our approach to Refs.~\cite{Lloyd2016,Gunn,Casper} due to the different model, we can expect an actual implementation of that type of approach to have at least an order of magnitude larger complexity.}

\subsection{Amplifying the initial state}
\label{sec:amplifying}

We aim to amplify the initial state so that we have the state with an equal superposition over cliques.
The strategy is to initially estimate the amplitude {just once}, then apply the appropriate number of steps of amplitude amplification when we are preparing the state to estimate the size of the null eigenspace.
It is possible to show that the complexity of estimating the amplitude is as given in the following Lemma.

\begin{lemma}[Quantum amplitude estimation]
\label{lem:qae}
Let $U$ be a unitary and let $0<a<1$ be such that
\begin{equation}
    U\ket{0,0}=a\ket{\psi_0,0}+\sqrt{1-a^2}\ket{\psi_1,1}.
\end{equation}
There exists a quantum algorithm which estimates $a$ to within error $\epsilon$ with probability of error less than $\delta$, using
\begin{equation}
    N = \frac{\pi}{\epsilon}\sqrt{1+\alpha^2} = \frac{1}{2\epsilon} \ln(1/\delta) + \mathcal{O}(\epsilon^{-1} \ln \ln (1/\delta))
\end{equation}
calls to $U$ or $U^\dagger$.
\end{lemma}

The proof for this Lemma is given in \append{proofest}.
To see the value of $\epsilon$ needed, note that probability of success will be reduced to approximately $\sin^2((1\pm\epsilon/a)\pi/2)$ if we incorrectly choose the number of iterates in the amplitude amplification due to imprecision in estimating the amplitude.
That translates to a probability of failure of the amplitude amplification of approximately $(\epsilon\pi/2a)^2$.
For our application, the amplitude is approximately
\begin{equation}
    \sqrt{1-\frac 1{2c}} \sqrt{\frac{|{\rm Cl}_k(G)|}{\binom{n}{k}}},
\end{equation}
where the first factor comes from failure of the Dicke state preparation, and the second from the clique checking.
For simplicity, in the following expressions for complexity we will omit the factor of $\sqrt{1-1/{2c}}$ which is close to 1.
The amplitude estimation is needed because it typically will be unknown how many cliques there are $|{\rm Cl}_k(G)|$.
Inaccuracy in the amplitude estimation translates to a probability for failure of the amplitude amplification due to using an incorrect number of steps.

In practice the ``failure'' of the amplitude amplification is not a major problem, because it can be combined into an uncertainty in estimation of the Betti number.
That is, in the next step instead of estimating the Betti number relative to $|{\rm Cl}_k(G)|$, we will be estimating it relative to a value that may be increased by about a factor of $1/[1-(\epsilon\pi/2a)^2]$ (using the approximation of the $\sin$ function).
If we want $(\epsilon\pi/2a)^2$ no more than a relative error $r$, then we should choose
\begin{equation}
    \epsilon \le \frac{2\sqrt{r}}{\pi}\sqrt{\frac{|{\rm Cl}_k(G)|}{\binom{n}{k}}}.
\end{equation}
That means that the cost would be
\begin{equation}\label{eq:ampest}
    \frac{\ln(1/\delta)}{\sqrt{r}}\frac{\pi}{4}\sqrt{\frac{\binom{n}{k}}{|{\rm Cl}_k(G)|}}
\end{equation}
steps.
In comparison, the number of steps of the amplitude amplification is approximately
\begin{equation}
    \frac{\pi}{4}\sqrt{\frac{\binom{n}{k}}{|{\rm Cl}_k(G)|}}.
\end{equation}
That is, the amplitude estimation is more costly by a factor of $\ln(1/\delta)/\sqrt r$.

This cost of the Dicke state preparation from {Lemma \ref{lem:dicke}} will be doubled in amplitude estimation and amplification when we account for the need to unprepare the Dicke state.
{We also need to reflect on the clique check, with complexity $6|E|+2\log k$ or $4|E^C|$ as described in Lemma \ref{lem:cliques}.
The total complexity for each step of the amplitude estimation and amplification is the total of these two complexities.}

{
This approach of separating the estimation and amplification provides a significant improvement over the obvious approach of using fixed-point amplitude amplification \cite{YoderPRL14} to provide amplification with an unknown overlap.
That requires a logarithmic factor in the complexity similar to amplitude estimation.
In contrast, here we only have that logarithmic factor in the cost once in the initial amplitude estimation, then in the remainder of the algorithm we eliminate the logarithmic factor by just performing amplitude amplification with the initially estimated amplitude.

It is somewhat ambiguous to compare our approach to that in Refs.~\cite{Lloyd2016,Gunn,Casper}.
Reference \cite{Lloyd2016} just invokes the ``multi-solution version of Grover’s algorithm'', which is not sufficiently specific to give a complexity because there are multiple approaches.
Reference \cite{Gunn} cites the version of Grover's algorithm from \cite{Brassard1998}, and Ref.~\cite{Casper} just mentions Grover's algorithm and uses the complexity from \cite{Gunn}.
The problem with citing \cite{Brassard1998} is that it is not sufficient to specify exactly which approach is intended.

One approach for searching with an unknown number of solutions given in that work is to just use the approach of \cite{Boyer}, which would give a factor of $9/\pi$ in the complexity, but that approach would not be compatible with a later amplitude estimation used in Ref.~\cite{Gunn}.
That is because the approach of \cite{Boyer} relies on a sequence of measurements to obtain success of the search.
The measurements would prevent the later amplitude estimation (for the number of zero eigenvalues) being used.

Reference \cite{Brassard1998} also mentions the approach of performing amplitude estimation, followed by Grover's algorithm for a known number of solutions.
Our proposal here is to divide between using amplitude estimation once, followed by amplitude amplification based on the estimation many times within the rest of the algorithm.
That gives a significant improvement over using both at every step
(which would be the obvious interpretation of just citing \cite{Brassard1998}).

Moreover, we provide a significant improvement in the efficiency of amplitude estimation over that in \cite{Brassard1998}. See Theorem 6 of that work for their result in terms of the error in the squared amplitude.
Translating that to the error in the amplitude, the number of steps needed is approximately $\pi/\epsilon$ to obtain $1-\delta=8/\pi^2$.
Repetitions would be needed to obtain a desired $\delta$ which would typically be smaller.
If, for example, $\delta=1/20$ and the number of repetitions is 5, then our approach gives about an order of magnitude improvement.
}

\subsection{Block-encoding the sparse Hamiltonian}
\label{sec:block_encode}
Having constructed the sparse oracles in the previous section, we now implement a quantum circuit that block-encodes the sparse Hamiltonian.
{Block encoding is a generalisation of a linear combination of unitaries, where an operator $B$ is given by $\bra{0}U\ket{0}=B/\lambda$ for a unitary operator $U$ acting on an ancilla system as well, and $\ket{0}$ on that ancilla system.
Together with a reflection on the ancilla system, it can then be used to construct what was dubbed a ``qubitised'' or ``qubiterate'' operator.
These principles were introduced in Ref.~\cite{Low2019hamiltonian}.}
We use a similar principle as in \cite{Ubaru,Cade}, except here we are implementing the Dirac operator $B_G$ rather than the combinatorial Laplacian.
In \cite{Cade} it is shown that the Dirac operator for \emph{all} Hamming weights and unrestricted by the cliques can be written as  
\begin{equation}
\label{decomposition}
    B = \sum_{j=1}^{n} (a_j + a_j^\dagger),
\end{equation}
where $a_j$ and $a_j^\dagger$ are fermionic annihilation and creation operators on qubit $j$. Using the usual Jordan-Wigner representation that gives the Hamiltonian
\begin{equation}
    \sum_{j=1}^{n} Z_1\otimes Z_{j-1}\otimes X_j ,
\end{equation}
where the subscripts indicate the qubits that these operators act on (starting the numbering from 1).
This is the core of the implementation of the complete Hamiltonian, and can easily be implemented by first preparing an equal superposition state over $n$ basis states, then applying the controlled string of Pauli operators as in Figure 9 of \cite{BabbushPRX18}.

To understand the reason that the Pauli string encodes the matrix, note that $\partial_k$ will remove a one from some location in the bit string $x$ of Hamming weight $k+1$ and apply a sign according to the number of ones prior to that location.
That can be achieved by applying an $X$ in that location, and applying $Z$ gates on all qubits prior to that location.
We need a superposition of applying the $X$ in all locations where there are ones.
Moreover, we also want to apply $\partial_{k+1}^\dagger$ to a bit string of Hamming weight $k+1$.
This involves flipping a zero to a one (which can be done with an $X$ gate) and applying a sign according to the number of ones prior to that position.
This can again be done using a string of $Z$ gates.
Now we want a superposition of performing $X$ gates at all locations where there are ones, \emph{and} $X$ gates where there are zeros, which can be implemented by the above sum of Pauli strings.

Here we aim to block encode the matrix
\begin{equation}
B_G =
\begin{bmatrix}
    0 & \partial_{k-1}^G & 0\\
    \partial_{k-1}^{G\dagger} & 0 & \partial_{k}^G \\
    0 & \partial_{k}^{G\dagger} & 0
    \end{bmatrix}.
\end{equation}
The difference of this from the unrestricted case $B$ in \cite{Ubaru} is that it only acts on states with Hamming weight $k-1,k,k+1$, and gives zero otherwise.
Similarly, it only gives states with Hamming weight in this range.
Moreover, $B_G$ is restricted to the clique subspace.
That means it must give zero if the input state is not a clique, and must also not give any output states that are not cliques.

Next we provide a general method of constructing a qubiterate operator in cases where tests on the system state are required.
The block encoding with the tests can be described as
\begin{equation}
    \left( \ket{0}\bra{0} \otimes P \right) V \left( \ket{0}\bra{0} \otimes P \right) = \ket{0}\bra{0} \otimes B_G/\lambda ,
\end{equation}
where $P$ is a projection on the system that tests the Hamming weight and cliques.
We are adopting notation similar to Eq.~(3) in \cite{BerryNPJ18}, but replacing the identity with $P$ to indicate that a projection is needed on the target system.
We will assume $V$ is Hermitian; if it is not we can construct a Hermitian $V$ by block encoding it as $V\mapsto V\otimes \ketbra{1}{0} + V^\dagger \otimes \ketbra{0}{1}$~\cite{Harrow2009}.
Similarly, we are writing $B_G$ for the operator we aim to block encode, but this reasoning applies for a more general Hamiltonian $H$.

If $\ket{k}$ is an eigenstate of $B_G$ with energy $E_k$ and satisfying $P\ket{k}=\ket{k}$, then by definition we must have
\begin{equation}
    V \ket{0} \ket{k} = \frac{E_k}{\lambda} \ket{0} \ket{k}  + i \sqrt{1-\left| \frac{E_k}{\lambda}\right|^2}\ket{0k^\perp},
\end{equation}
where $\ket{0k^\perp}$ is defined as a state such that
\begin{equation}
    \left( \ket{0}\bra{0} \otimes P \right)\ket{0k^\perp} = 0.
\end{equation}
Then we can define the qubiterate as
\begin{equation}
    W := R V,
\end{equation}
with
\begin{equation}
    R := i\left( 2\ket{0}\bra{0} \otimes P - I \right).
\end{equation}
This is similar to that in \cite{BerryNPJ18}, except we have included the projection $P$ in the reflection operation.
That is, we are applying the tests as part of the reflection, instead of applying them in the operation $V$.

Then we obtain
\begin{equation}
    W \ket{0} \ket{k} = i \frac{E_k}{\lambda} \ket{0} \ket{k}  + \sqrt{1-\left| \frac{E_k}{\lambda}\right|^2}\ket{0k^\perp}.
\end{equation}
It is also found that
\begin{equation}\label{eq:Wperp}
W \ket{\chi k^\perp} = i \frac {E_k}\lambda\ket{\chi k^\perp} + \sqrt{1-\left|\frac{E_k}\lambda\right|^2} \ket{\chi}\ket{k}.    
\end{equation}
Here we have corrected a minor error from \cite{BerryNPJ18} where there was an $i$ appearing on the second term.
See \append{proj} for the derivation.
Then it is easy to see that
\begin{equation}
    \frac 1{\sqrt 2} \left( \ket{0} \ket{k} \pm \ket{0k^\perp} \right)
\end{equation}
are eigenstates of $W$ with eigenvalues $\pm e^{\pm i \arcsin(E_k/\lambda)}$.
This is the usual relation for the eigenvalues of the qubitised operators, showing that this approach for constructing the walk operator works.

For our implementation here, $V$ is just the controlled string of Pauli operators together with preparation of an equal superposition state.
The reflection on the target system expressed by the projector can be implemented by computing an ancilla qubit flagging that the projection is satisfied (we have the appropriate Hamming weight range and cliques), reflecting on that qubit and the control qubits, then uncomputing the test.
In some cases this can give a significant reduction in complexity over performing the test before and after $V$.
If the ancilla qubits used to compute the tests are retained, then they can be erased with Clifford gates and measurements.
For the application here that would be too costly in terms of ancilla qubits, so we incur the Toffoli cost of the test again in erasing the ancillas.

For the complexity of the implementation we have the following costs.
\begin{enumerate}
    \item Preparing an equal superposition state over $n$ basis states, which can be performed with complexity $4\lceil \log n\rceil+1$ Toffolis \cite{SandersPRQ20}, or just with Hadamards if $n$ is a power of 2.
    This cost is incurred twice.
    \item The controlled string of Pauli operators can be applied with Toffoli complexity $n-1$ using the method in Figure 9 of \cite{BabbushPRX18}.
    \item The Hamming weight can be computed with no more than $n$ Toffolis and $n$ ancilla qubits \cite{Kivlichan2020improvedfault}.
    In that case we would not need to double the complexity for the reflection because the sum can be uncomputed with Cliffords.
    We could use $2n$ Toffolis with a logarithmic number of qubits \cite{SandersPRQ20}, but in that case we would need to double the complexity for the uncomputation cost.
    \item The complexity of outputting qubits with $\binom{k-1}{2}$, $\binom{k}{2}$, or $\binom{k+1}{2}$ is no more than $3\lceil \log n\rceil$.
    At the same time we can use the QROM to output a qubit which flags if the Hamming weight is outside the range.
    These qubits can be erased with Cliffords by retaining a logarithmic number of ancilla qubits.
    \item As described above the cost of the reflection on the clique test is no more than $6|E|+2\log k$ {given a database of edges, or $4|E^C|$ given a database of missing edges, as in Lemma \ref{lem:cliques}.}
    \item Note that there is a reflection on the result of two tests, but this would correspond to a controlled-$Z$ which is a Clifford gate.
\end{enumerate}
{The overall complexity is therefore as in the following lemma.
\begin{lemma}[Block encoding complexity]\label{lem:block}
    The Toffoli complexity of block encoding $B_G/\lambda$ with the operator $B_G$ as defined in Eq.~\eqref{eq:bdef} is
    \begin{equation}
        6|E| + 2\log k + 5n + 11\log n + \mathcal{O}(1),
    \end{equation}
    when given a database of edges $E$, or
\begin{equation}
    4|E^C| + 5n + 11\log n + \mathcal{O}(1) .
\end{equation}
when given a database of missing edges $E^C$.
The value of $\lambda$ for this block encoding is approximately $n$.
\end{lemma}

These complexities come from adding the Toffoli complexities in the list above.
The value of $\lambda$ is obtained by noting that we use a linear combination of $n$ Pauli strings.}
This value will be increased very slightly because of imperfect preparation of an equal superposition state in the method of \cite{SandersPRQ20}.
That increase is normally less than one part in 1000, so will be ignored here.

{
In comparison, the approaches in Refs.~\cite{Gunn,Casper} are not very specific about the approach.
They give a factor of $n^2$ for an $n$-sparse operator, coming from the general procedure for decomposing an unstructured $n$-sparse operator into $1$-sparse operators from \cite{BerryFOCS15}.
The implementation of the operator also requires checking cliques, which is not addressed in Refs.~\cite{Gunn,Casper}.

Ignoring those issues of how the operator is applied, the major difference between the proposal here is that we use a block encoding to construct a walk operator instead of simulating evolution under the Hamiltonian.
References \cite{Gunn,Casper} invoke the results in Refs.~\cite{BerryFOCS15,QSP} for the Hamiltonian evolution for unit time.
In practice that would need to be adjusted to a time $1/\lambda$ in the simulation in order to prevent wraparound of the eigenvalues (which would cause nonzero eigenvalues to be measured as zero).
For these short times the complexity of the Hamiltonian evolution is multiplied by a logarithmic factor.

If we use the estimate of the complexity from Ref.~\cite{QSP} given in Ref.~\cite{BabbushPRA19}, then we may expect the complexity to be larger than the complexity for the block encoding we have given here by a factor of 6 for the example in Section \ref{subsec:K-graphs}.
There will be more significant factors in other examples with smaller gaps than the example in Section \ref{subsec:K-graphs}.
}

\subsection{Projection-based overlap estimation}
\label{sec:poly_approx}

In order to estimate the number of zero eigenvalues of the Hamiltonian, we project onto the zero eigenspace, then perform amplitude estimation.
The projection can be approximated using a Chebyshev polynomial approach.
First, recall that in the qubitisation the zero eigenvalue of the Hamiltonian is mapped to eigenvalues $\pm 1$ of the qubitised operator.
For the filter function on the phase $\phi$ of the eigenvalues of the walk operator, one can take
\begin{equation}
    \tilde w(\phi) = \epsilon T_{\ell}\left( \beta \cos\left( \phi \right) \right)
\end{equation}
for $\phi$ taking discrete values $\pi k/\ell$ for $k$ from $-\ell$ to $\ell$, and
where $\beta=\cosh(\tfrac 1\ell \cosh^{-1}(1/\epsilon))$.
Taking the discrete Fourier transform of these values gives the window $w_j$ such that
\begin{equation}\label{eq:tildew}
\tilde w(\phi)=  \sum_{j=-\ell}^\ell w_j e^{ij\phi}.
\end{equation}
Note that $\tilde w(\phi)$ is a function of $\cos\phi$, so $w_j=w_{-j}$.
Moreover, we have values of $j$ separated by 2.
If $\ell$ is even, then we have even powers of $\cos\phi$, and therefore only even $j$.
This means that it can be regarded as a polynomial in $e^{2i\phi}$.
We can select between the qubitised walk step and its inverse by controlling on the reflection, so implementing a linear combination of unitaries may be performed with cost $\ell$.

The peak for $\tilde w(\phi)$ will be at $0$ and $\pi$, which is what is needed because the qubitised operator produces duplicate eigenvalues at phases of $0$ and $\pi$.
The width of the operator can be found by noting that the peak is for the argument of the Chebyshev polynomial equal to $\beta$, and the width is where the argument is 1, so $\beta\cos(\phi)=1$.
This gives us
\begin{equation}\label{eq:chebwid}
    \cosh(\tfrac 1\ell \cosh^{-1}(1/\epsilon)) \cos(\phi) = 1 .
\end{equation}
The gap in the Hamiltonian is $\lambda_{\min}$, which translates to a gap in the qubitised operator of $\arcsin(\lambda_{\min}/\lambda)$.
Because the width of the peak should be equal to the gap, we can replace $\phi$ with $\arcsin(\lambda_{\min}/\lambda)$, and solving for $\ell$ gives
\begin{equation}\label{eq:filtering}
    \ell = \frac{\cosh^{-1}(1/\epsilon)}{\cosh^{-1}(1/\sqrt{1-(\lambda_{\min}/\lambda)^2})} \le \frac{\lambda}{\lambda_{\min}} \ln(2/\epsilon).
\end{equation}

{The complexity of the filter on the walk operator may therefore be given as in the following lemma.
\begin{lemma}[Eigenvalue filtering]\label{lem:filter}
    The complexity of filtering out nonzero eigenvalues of $B_G$ by a factor of $\epsilon$ is
    \begin{equation}
    \frac{n}{\lambda_{\min}} \ln(2/\epsilon) 
\end{equation}
calls to the block encoding of $B_G$, given that the gap from eigenvalue 0 is at least $\lambda_{\min}$.
\end{lemma}

This lemma is obtained by using $\lambda=n$ for the block encoding of $B_G$ in Eq.~\eqref{eq:filtering}.}
To determine the appropriate value of $\epsilon$ to take, note that $\epsilon$ tells us the multiplying factor for amplitudes for states with eigenvalues outside the gap.
The state starts with equal weighting on all eigenvalues, so ideally we should have the amplitude after filtering $\sqrt{\beta^G_{k-1}/|{\rm Cl}_k(G)|}$.
If the state amplitudes outside the gap are multiplied by $\epsilon$, then the error in the squared amplitude can be at most $\epsilon^2$.
This corresponds to an error in $\beta^G_{k-1}$ of $\epsilon^2|{\rm Cl}_k(G)|$, or a relative error of $\epsilon^2|{\rm Cl}_k(G)|/\beta^G_{k-1}$.

{In comparison, the approaches in Refs.~\cite{Lloyd2016,Gunn,Casper} are based on phase estimation.
They just give the scaling without specifying the method, which is needed to know the constant factors.
The best algorithm for phase estimation would correspond to the method we have given in Appendix \ref{append:proofest}.
That would have asymptotic complexity similar to the filtering approach given here.
To compare the complexities, we first need to note that the accuracy of the phase estimation should be half the gap.
This is because if the phase estimate has error half the gap, if the eigenvalue is zero then it could give an estimate of $\lambda_{\min}/2$, which could also correspond to an eigenvalue of $\lambda_{\min}$.
When this factor of 2 is accounted for the phase estimation approach has similar complexity to the filter.
The phase estimation has a somewhat larger complexity in the non-asymptotic regime, though by only about 60\%.

That is, there is a moderate improvement over phase estimation even if one were to use the optimal phase estimation introduced in Appendix \ref{append:proofest}.
Because Refs.~\cite{Lloyd2016,Gunn,Casper} did not use that method of phase estimation we would provide a larger improvement over those works, though the size of the improvement is ambiguous because they do not specify the method of phase estimation.
}

\subsection{Total complexity of algorithm}
\label{sec:total}

The Toffoli costs of the algorithm are as follows.
{In the following we will present the complexity when using the database of edges, then explain the modification for a database of missing edges.}
\begin{enumerate}
    \item The preparation of the Dicke state has a leading order complexity
\begin{equation}
    n\log^2 n + \mathcal{O}(n\log n)
\end{equation}    
    {Toffolis, or approximately $2kn$ for the two schemes presented in \append{Dicke}.
    Here we are including a factor of $2$ for inversion.}
    \item The cost of checking cliques is given by $6|E|+2\log k$ where we take into account the need to uncompute the result.
    \item The cost of amplitude estimation is a number of iterations of steps 1 and 2 given as
    \begin{equation}\label{eq:ampest2}
    \frac{\ln(1/\delta)}{\sqrt{r}}\frac{\pi}{4}\sqrt{\frac{\binom{n}{k}}{|{\rm Cl}_k(G)|}} .
\end{equation}
    \item The cost of amplitude amplification of the cliques is given by approximately $(\pi/4)\sqrt{{\binom{n}{k}}/{|{\rm Cl}_k(G)|}}$ of iterations of steps 1 and 2.
    \item The walk step for the qubitisation needs $6|E| + 5n + 11\log n + 2\log k + \mathcal{O}(1)$ Toffolis.
    \item For the filtering there are $\frac{n}{\lambda_{\min}} \ln(2/\epsilon)$ calls to the block encoding with costs in item 5 above.
    \item Lastly, we need to perform amplitude estimation on the entire procedure, using $\approx \log(1/\delta)/2\epsilon$ calls to the amplitude amplification in 4 and filtering in 6. 
\end{enumerate}
{In comparison Ref.~\cite{Gunn} invoked the amplitude estimation scheme of Ref.~\cite{Brassard1998}, which we improve over by about an order of magnitude.
Reference \cite{Casper} just uses classical sampling, which is quadratically more costly.}
To give the leading-order complexities, the combined cost of steps 1 and 2 is
\begin{equation}
    6|E| + n\log^2 n + \mathcal{O}(n\log n).
\end{equation}
To distinguish the $\epsilon$, $\delta$ and $r$ (relative error) needed in different steps we will use subscripts.
The cost of amplitude estimation is then approximately
\begin{equation}\label{eq:amescost}
    \frac{\ln(1/\delta_1)}{\sqrt{r_1}}\frac{\pi}{4}\sqrt{\frac{\binom{n}{k}}{|{\rm Cl}_k(G)|}} (6|E| + n\log^2 n) .
\end{equation}
This is expected to be a trivial cost in the overall algorithm, because the amplitude amplification is performed many more times.

For the remainder of the algorithm, we have an initial cost of 
\begin{equation}
    \frac {\pi}{4} \sqrt{\frac{\binom{n}{k}}{|{\rm Cl}_k(G)|}}  (6|E| + n\log^2 n)
\end{equation}
for the amplitude amplification for the initial state.
Then there is a cost for the block encoding of
\begin{equation}
    6|E| + 5n + \mathcal{O}(\log n),
\end{equation}
for each step.
Multiplying by the number of steps needed for filtering, there is a cost
\begin{equation}
    \frac{n}{\lambda_{\min}} \ln(2/\epsilon_3) [ 6|E| + 5n + \mathcal{O}(\log n)].
\end{equation}
To determine the appropriate value of $\epsilon_3$ to take, note that we are measuring a kernel of size $\beta^G_{k-1}$ as compared to an overall dimension of $|{\rm Cl}_k(G)|\gg \beta^G_{k-1}$.
The relative accuracy in the estimation of $\beta^G_{k-1}$ will therefore be about $\epsilon_3^2|{\rm Cl}_k(G)|/\beta^G_{k-1}$ as explained above at the end of \sec{poly_approx}.
If we aim for relative accuracy $r_3$, then we have complexity
\begin{equation}
    \frac{n}{2\lambda_{\min}} \ln\left(\frac{4|{\rm Cl}_k(G)|}{r_3\beta^G_{k-1}}\right) ( 6|E| + 5n + \mathcal{O}(\log n)).
\end{equation}

Lastly, the amplitude estimation on the entire procedure needs a number of repetitions
\begin{equation}
    \frac{\ln(1/\delta_2)}{2\epsilon_2}.
\end{equation}
But, this amplitude estimation is on a number of steps corresponding to a reflection requiring both the forward and reverse calculations.
That introduces a further factor of $2$, so we should use
\begin{equation}
    \frac{\ln(1/\delta_2)}{\epsilon_2}.
\end{equation}
Next, $\epsilon_2$ corresponds to an accuracy of estimating a ratio $\sqrt{\beta^G_{k-1}/|{\rm Cl}_k(G)|}$.
If we want a relative accuracy $r_2$, then the error propagation formula gives
\begin{equation}
    r_2 = \frac{\Delta \beta^G_{k-1}}{\beta^G_{k-1}} = \frac{\epsilon_2}{\beta^G_{k-1}} \left( \frac{d }{d\beta^G_{k-1}}\sqrt{\frac{\beta^G_{k-1}}{|{\rm Cl}_k(G)|}}\right)^{-1} = 2 \epsilon_2 \sqrt{\frac{|{\rm Cl}_k(G)|}{\beta^G_{k-1}}} ,
\end{equation}
where $\Delta \beta^G_{k-1}$ is uncertainty in $\beta^G_{k-1}$.
In terms of $r_2$, the number of repetitions becomes
\begin{equation}
    2\frac{\ln(1/\delta_2)}{r_2}\sqrt{\frac{|{\rm Cl}_k(G)|}{\beta^G_{k-1}}}.
\end{equation}
Applying this to the complexity required for each step, we get a complexity of approximately
\begin{equation}\label{eq:leadcomp}
   \frac{\ln(1/\delta_2)}{r_2}\sqrt{\frac{|{\rm Cl}_k(G)|}{\beta^G_{k-1}}} \left[ \frac {\pi}{2} \sqrt{\frac{\binom{n}{k}}{|{\rm Cl}_k(G)|}}  (6|E| + n\log^2 n) + \frac{n}{\lambda_{\min}} \ln\left(\frac{4|{\rm Cl}_k(G)|}{r_3\beta^G_{k-1}}\right) ( 6|E| + 5n ) \right].
\end{equation}
{This is the expression given as Eq.~\eqref{eq:method1} in Lemma~\ref{lem:complex}.
If we use the second Dicke preparation scheme, then $n\log^2 n$ would be replaced with $2kn$, and $\binom{n}{k}$ replaced with $n^k/k!$, which gives the expression in Eq.~\eqref{eq:method2} of Lemma~\ref{lem:complex}.
For this complexity the factor of $6|E|$ at the beginning is for the complexity from the database of edges; for the database of missing edges this factor would be replaced with $4|E^C|$.}

Comparing this to the amplitude estimation cost in \eq{amescost} the primary difference is the factor of
\begin{equation}
    \sqrt{\frac{|{\rm Cl}_k(G)|}{\beta^G_{k-1}}}
\end{equation}
here.
There is another difference in that the amplitude estimation cost has the factor $1/\sqrt{r_1}$ rather than $1/r_2$, so the scaling in terms of relative error is improved.
In cases where the number of cliques $|{\rm Cl}_k(G)|$ is much larger than the Betti number $\beta^G_{k-1}$, then the amplitude estimation cost is trivial.

We will have a total probability of failure $\delta=\delta_1+\delta_2$ due to the two amplitude estimations, and a total relative error $r_1+r_2+r_3$.
In order to reduce the complexity, we can use the fact that the cost of the initial amplitude estimation is much smaller, so we can take $\delta_1$ and $r_1$ smaller without much impact on the overall complexity.
The $r_3$ appears inside a logarithm, so can be taken to be smaller than $r_3$.

To give the scaling of the complexity in a simpler way, we can simply ignore the amplitude estimation complexity, and replace $\delta_2$ with $\delta$, and replace both $r_2$ and $r_3$ with $r$ (since $r_3$ can be taken as for example $r/20$ without much impact on the overall complexity).
We will also omit terms of complexity $kn$ or $n$ as compared to $|E|$.
That then gives
\begin{equation}\label{eq:toffoli_requirement}
  T(G, k, r, \delta) := 6|E|\frac{\ln(1/\delta)}{r}\sqrt{\frac{|{\rm Cl}_k(G)|}{\beta^G_{k-1}}} \left[ \frac {\pi}{2} \sqrt{\frac{\binom{n}{k}}{|{\rm Cl}_k(G)|}}   + \frac{n}{\lambda_{\min}} \ln\left(\frac{4|{\rm Cl}_k(G)|}{r\beta^G_{k-1}}\right) \right],
\end{equation}
where $T(G, k, r, \delta)$ gives the required number of Toffoli gates to estimate, with precision parameters $r, \delta$, the $(k-1)$-th order Betti number of a graph $G$ with $n$ nodes, $|E|$ edges and a Laplacian with gap $\lambda_{\min}$.

This expression for the complexity is in terms of the relative accuracy $r$.
Alternatively, if we aimed for a given absolute accuracy $\abser$, then $\abser=r\beta^G_{k-1}$, and so the expression for the complexity becomes
\begin{equation}\label{eq:toffoli_abser}
  T(G, k, r, \delta) := 6|E|\frac{\ln(1/\delta)}{\abser}\sqrt{{|{\rm Cl}_k(G)|}{\beta^G_{k-1}}} \left[ \frac {\pi}{2} \sqrt{\frac{\binom{n}{k}}{|{\rm Cl}_k(G)|}}   + \frac{n}{\lambda_{\min}} \ln\left(\frac{4|{\rm Cl}_k(G)|}{\abser}\right) \right].
\end{equation}

We now discuss the complexity of just the first term in the square brackets, which corresponds to the state preparation cost rather than the filtering cost.
That cost will be dominant if the gap is large, though it must be emphasised that the gap will be small in many cases.
This first term for the cost gives
\begin{equation}\label{eq:simplecost}
  T(G, k, r, \delta) = 3\pi |E|\frac{\ln(1/\delta)}{r}\sqrt{\frac{\binom{n}{k}}{\beta^G_{k-1}}} .
\end{equation}
If we are aiming for a given absolute accuracy $\abser$ in $\beta^G_{k-1}$, then the complexity would be
\begin{equation}
  T(G, k, r, \delta) = 3\pi |E|\frac{\ln(1/\delta)}{\abser}\sqrt{{\binom{n}{k}}{\beta^G_{k-1}}} .
\end{equation}
The complexity is now larger for large Betti number $\beta^G_{k-1}$.
The reason for this is that the amplitude estimation is estimating the \emph{square root} of $\beta^G_{k-1}$.
The square root has a small derivative for large values of $\beta^G_{k-1}$, making it more difficult to estimate the Betti number with small absolute error.
Again, note that the last three expressions above are only for the state preparation cost, without the filtering cost.

To compare to the complexity of classical approaches, an exact diagonalisation approach would tend to scale as 
${\binom{n}{k}}^2$, whereas approximate schemes scale as $\binom{n}{k}$.
Thus the quantum algorithm would give approximately a square-root speedup over these classical algorithms if $\beta^G_{k-1}$ is on the order of a constant and one is targeting a fixed relative error estimate.
On the other hand, for graphs with large $\beta^G_{k-1}$, a speedup that is greater than a square root can be obtained for fixed relative error estimates.

\section{Regimes for quantum speed-up}
\label{sec:regimes_qspeedup}

In this section, we ask if there exist regimes where our quantum algorithm offers a significant speedup over the best classical algorithms. The aim is to compute to relative error the $(k-1)^{\text{th}}$ Betti number of the clique complex of a graph $G$. Say $G$ has $n$ nodes, $|E|$ edges, $r$ is the desired multiplicative error, and $\lambda_{\min}$ is the spectral gap of the combinatorial Laplacian $\Delta_{k-1}^G = \partial_{k-1}^{G\dagger} \partial_{k-1}^{G} + \partial_{k}^{G}\partial_{k}^{G\dagger}$. To simplify the arguments, we will represent the quantum complexity of this problem as
\begin{equation}
T_q = \widetilde{\mathcal{O}}\left(\frac{n \left| E \right|}{r \, \lambda_{\min}} \sqrt{\frac{1}{\beta_{k-1}} \binom{n}{k}}\right) .
\end{equation}
Comparing to Eq.~\eq{toffoli_requirement}, this will asymptotically upper bound both terms up to log factors.

For a rough estimate of the cost of computing the Betti number classically, one could use $|{\rm Cl}_k(G)|$ (i.e., the number of $k$-cliques) or $\binom{n}{k}$.
The reason is that classical algorithms typically start by constructing a list of $k$-cliques, and afterwards compute the nullity of the combinatorial Laplacian or boundary operator.
The cost of this second step (i.e., estimating the nullity of the combinatorial Laplacian or boundary operator), scales at best linearly in size of the matrix $|{\rm Cl}_k(G)|$~\cite{ubaru:num_rank}.
On the other hand, the first step (i.e., listing all $k$-cliques) can be done using a brute force search at cost $\binom{n}{k}$.
There are more efficient algorithms for listing cliques, though the complexity tends to be dependent on the properties of the graph.
However, $|{\rm Cl}_k(G)|$ always lower bounds the cost of listing all the $k$-cliques.
Therefore, $|{\rm Cl}_k(G)|$ and $\binom{n}{k}$ can be considered to be lower and upper bounds on the scaling of the classical complexity, respectively.
In conclusion, the best classical algorithms for this problem have scaling lower bounded by
\begin{equation}
T_c = \Omega\left(|{\rm Cl}_k(G)|\right) .
\end{equation}
Recall ${\rm Cl}_k(G)$ is the set of cliques of size $k$, which form the $(k-1)$-simplices of the simplicial complex.
Classical algorithms have extra factors in the complexity, such as $1/r^2$ dependence on the required precision, that introduce orders of magnitude over this lower bound.
Another category of classical algorithms to compare to are those tailored for the specific regime where quantum algorithms are most efficient. 
Notable examples include the algorithm developed in Section~\ref{subsec:classical}, and the algorithm of Apers et al.~\cite{Apers22}.
We defer their comparison to Section~\ref{subsec:classical}, where we will highlight regimes where the examples introduced in Section~\ref{subsec:K-graphs} continue to exhibit a superpolynomial speedup.

The quantum algorithm will offer a speedup on instances where $\beta_{k-1}$ is large, and where $\lambda_{\min}$ is not too small.
We can remove the dependence on $\lambda_{\min}$ if we instead focus on computing an approximate Betti number, in the following sense.

\begin{definition}
The $\delta$-approximate $k^{\text{th}}$ Betti number is $B^\delta_k = \dim{\{v \in \mathcal{H}^G_{k-1} : \frac{v^\dag \Delta_k v}{v^\dag v} \leq \delta\}}$. Note $B^0_k = \beta_k$, and in general $B^\delta_k \geq \beta_k$.
\end{definition}

The same quantum algorithm computes $B^\delta_{k-1}$ to relative error with cost
\begin{equation}
T_q = \widetilde{\mathcal{O}}\left(\frac{n \left| E \right|}{r \, \delta} \sqrt{\frac{1}{B^\delta_{k-1}} \binom{n}{k}}\right) .
\end{equation}

\subsection{A family of graphs with large Betti numbers and large spectral gaps} \label{subsec:K-graphs}

In this section, we will construct a family of graphs with all the necessary parameters to enable a large quantum speedup. 
Our objective here is to demonstrate the existence of instances that fulfill all the prerequisites for the quantum algorithm to achieve a superpolynomial quantum speedup.

Let $K(m,k)$ be the $k$-partite complete graph, where each partition contains $m$ vertices. That is, $K(m,k)$ consists of $k$ clusters, each with $m$ vertices; there are no edges within clusters, but all edges between clusters are included. Note $K(m,1)$ is a collection of $m$ points with no edges.
$K(m, k)$ gives a useful example of a clique complex with a high Betti number \cite{adamaszek:extremal_betti}. It also has a Laplacian with a large spectral gap.

\begin{figure}[tbh]
\centering
\label{fig:kpartite}
  \includegraphics[width=0.3\linewidth]{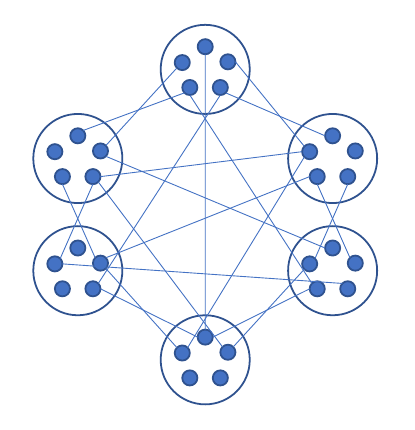}
  \caption{The graph $K(5, 6).$}
  \label{partite_complete}
\end{figure}

\begin{proposition} \label{prop:Betti_number}
The $(k-1)^{\text{th}}$ Betti number of (the clique complex of) $K(m,k)$ is
\begin{equation}
\beta_{k-1} = (m-1)^k .
\end{equation}
\end{proposition}

\begin{proposition} \label{prop:spectral_gap}
The combinatorial Laplacian $\Delta_{k-1} = \partial_{k-1}^{\dagger} \partial_{k-1} + \partial_{k}\partial_{k}^{\dagger}$ of (the clique complex of) $K(m,k)$ has spectral gap
\begin{equation}
\lambda_{\min} = m .
\end{equation}
\end{proposition}

We prove these in \append{proof_betti_density} using techniques from simplicial homology. A further useful fact is that
\begin{equation} \label{eq:clique_k}
|{\rm Cl}_k(K(m,k))| = m^k .
\end{equation}
Standard classical approaches need to at least store a vector of this length, so we can give a classical complexity as
\begin{equation}
T_c \sim e^{k \ln m} .
\end{equation}
As a first approximation for the quantum cost, we use the formula in \eq{simplecost} and consider just the square root factor and $|E|$.
{In fact, for this example the large number of edges means that it is better to use the list of missing edges to give complexity proportional to $|E^C|$.}
Bearing in mind that $n=mk$, Stirling's approximation gives
\begin{align}
\binom{n}{k} &\simeq {\frac{1}{\sqrt{2\pi}} \sqrt{\frac{mk}{(n-k)k}} \frac{(mk)^n}{(n-k)^{n-k} k^k}}
 \nn
&= \frac{1}{\sqrt{2\pi}} \sqrt{\frac{m}{n - \frac{n}{m}}} \left(\frac{m}{(m-1)^{1 - \frac{1}{m}}}\right)^n \nn
&\leq \left(\frac{m}{(m-1)^{1 - \frac{1}{m}}}\right)^n \, ,
\end{align}
{where in the first line we have omitted the exponentials in Stirling's approximation because they cancel.}
Proposition \ref{prop:Betti_number} gives $\beta_{k-1} = (m-1)^{n/m}$, giving a quantum complexity scaling as
\begin{equation}
T_q \sim {|E^C| \frac{n}{m} \left(\frac{m}{m-1}\right)^{n/2} \leq n^2 e^{k/2}.}
\end{equation}

Therefore, for constant $m$, there is a polynomial speedup by a $2\ln m$ root {(ignoring the $n^2$ factor).}
Alternatively, taking $k$ constant, the above formulae give
\begin{align}
    T_c &= \mathcal{O}(n^k), \\
    T_q &= \mathcal{O}(n^2) .
\end{align}
Then there is a polynomial speedup by a $k/2$ root.
To obtain a superpolynomial speedup, $m$ can be taken to increase close to linear in $n$, but $k$ can be taken to also increase with $n$.
Close to the best result is obtained for $k = c\ln^2n$ with some constant $c$.
Then the logs of the complexities are approximately
\begin{align}
    \ln T_c &\sim c \ln^3 n, \\
    \ln T_q &\sim 2\ln n + (c/2)\ln^2 n.
\end{align}
That implies a speedup by a $2\ln n$ root, which is superpolynomial.

This is still not an exponential speedup, but as far as the graph is concerned this is the best speedup that could be obtained from this type of approach.
This is because, with $k$ constant, the quantum complexity ignoring the {$|E^C|$} factor is $\mathcal{O}(1)$.
The Betti number is already scaling the same as $\binom{n}{k}$, but the overhead from {$|E^C|$} means that the speedup is not exponential.

Next we provide numerical results for the Toffoli complexity as a function of $n$ and $k$.
For these calculations we have made a number of adjustments to our simplified expressions in order to provide more accurate results.
In particular we compute the integral of the Kaiser window rather than using the asymptotic expression, as well as including the Dicke state preparation cost and the initial amplitude estimation cost for the number of steps needed for the state preparation.
{We are also using the second Dicke preparation scheme from \append{Dicke} which provides a smaller complexity for this example.}

The results are as given in Fig.~\ref{Toffoli_counts} as a function of $n$ for a range of values of $k$.
It can be seen that the cost of the quantum algorithm for a given $k$ scales approximately as $n^2$, with the cost scaling primarily coming from the number of edges in the graph.
The classical cost given as the number of k-cliques or $\binom{n}{k}$ has similar scaling, which is considerably worse than for the quantum algorithm, and is much worse for larger values of $k$, as expected from the analysis above.

For the example of $n=256,k=16$ the quantum cost is approximately {$6.8$ billion} Toffolis, which is comparable to gate counts for classically intractable instances of quantum chemistry \cite{Lee2020}.
In contrast, the number of cliques is about $2\times 10^{19}$, and $\binom{n}{k}\approx 10^{25}$.
These numbers are sufficiently large that it should be classically intractable {for any method that scales as $|{\rm Cl}_k(G)|$.}
For example, just storing the vector would be beyond the storage capacity of supercomputers.
{Potentially, more advanced classical algorithms that do not need to store the vector could be tractable.}

{To compare to the scheme as presented in Refs.~\cite{Gunn,Casper}, We improve by about two orders of magnitude for this example by using a more efficient Dicke state preparation scheme.
We have a further order of magnitude improvement in complexity by using optimal quantum amplitude estimation in the final step.
That gives at least three orders of magnitude improvement, which is the difference between a quantum computer running for a day versus years.
The total improvement is unclear because some parts of the algorithm were not specified in Refs.~\cite{Gunn,Casper}.

We obtain about another order of magnitude improvement by separately performing an amplitude estimation to avoid needing to repeatedly perform it in the initial state preparation.
The question of how this would be performed was not addressed in Refs.~\cite{Gunn,Casper}.
There is a more modest improvement in using the optimal filter as compared to optimal phase estimation.
But, the optimal phase estimation is a procedure introduced here, and there would be larger improvement over less efficient phase estimation.
The type of phase estimation was not addressed in prior work.
We also provide an improved clique checking procedure, but the model of the graph considered in Refs.~\cite{Gunn,Casper} is different, making a direct comparison of complexities impossible.
}

\begin{figure}[t!]
\centering
  \includegraphics[width=0.5\linewidth]{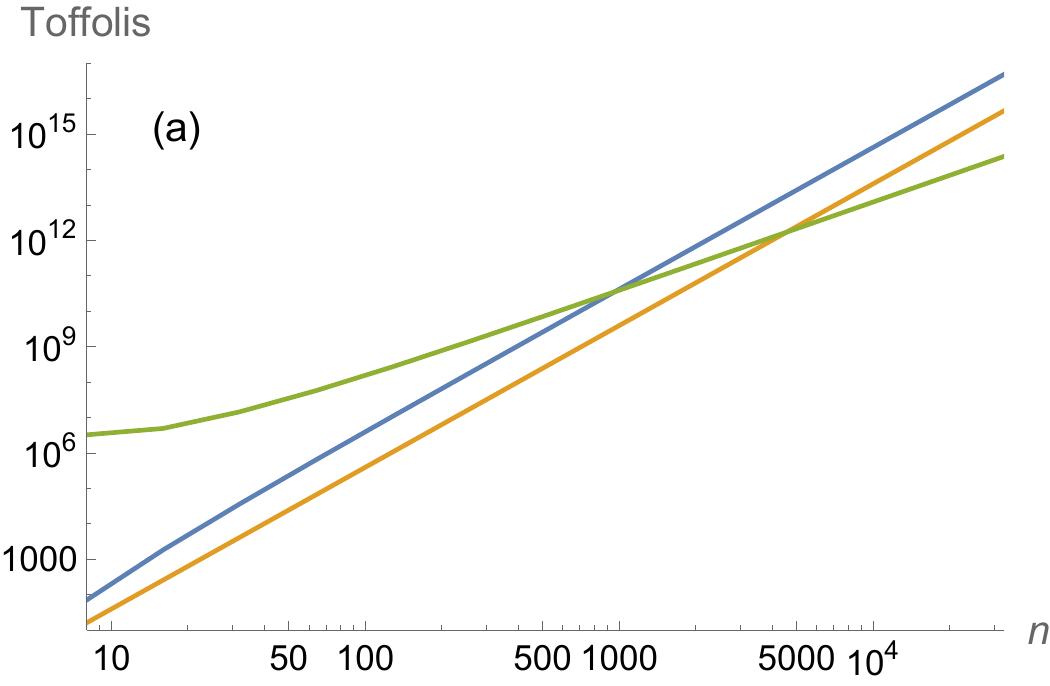}\includegraphics[width=0.5\linewidth]{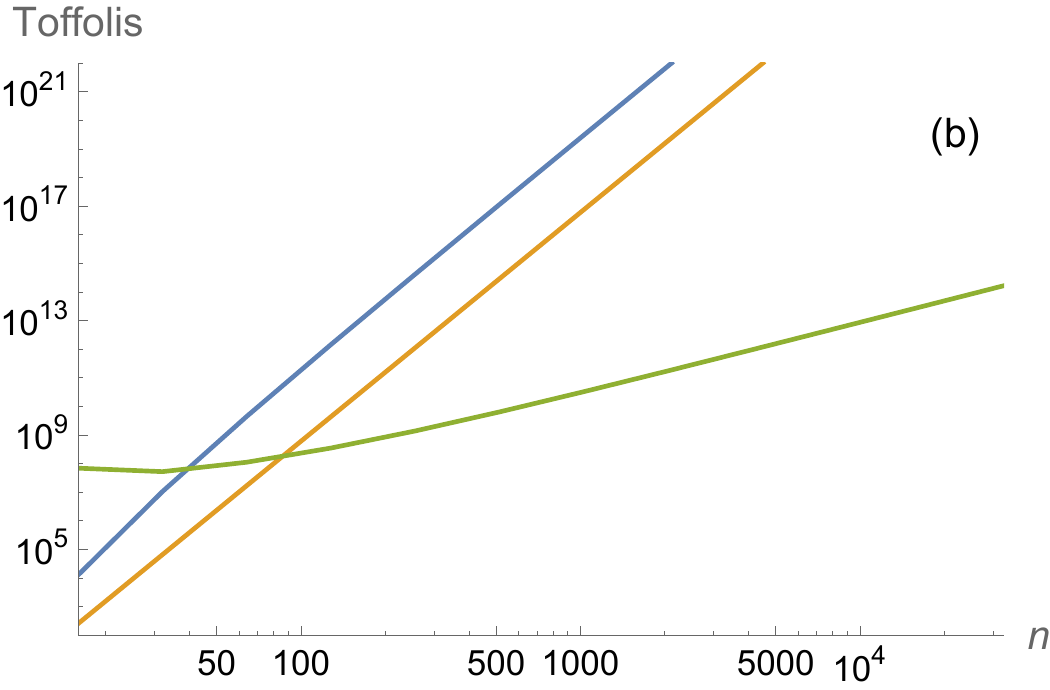}
  \includegraphics[width=0.5\linewidth]{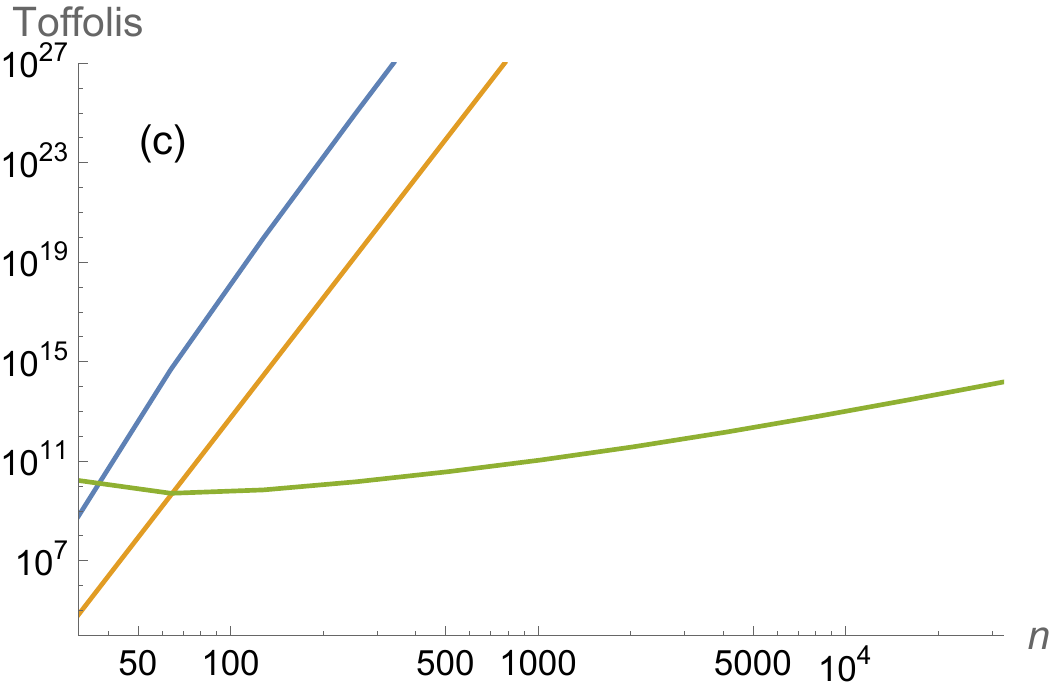}\includegraphics[width=0.5\linewidth]{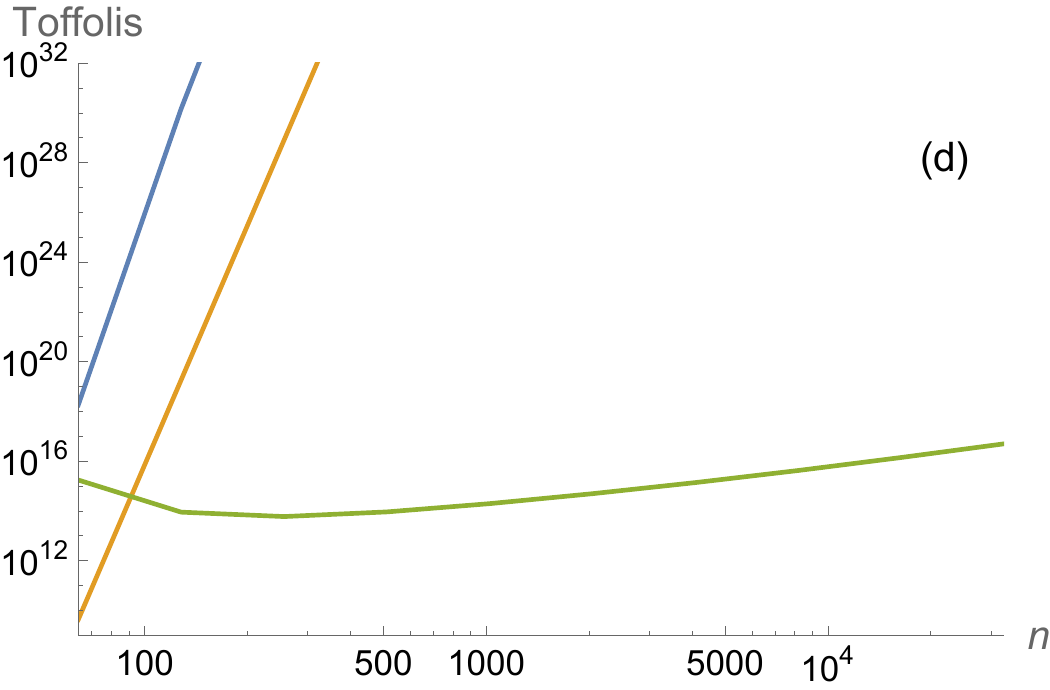}
  \caption{The Toffoli counts for the quantum algorithm for the Betti number of $K(m,k)$ as a function of $n$ for $k=4$ (a), $k=8$ (b), $k=16$ (c), and $k=32$ (d).
  The lines are green for the quantum complexity, blue for $\binom{n}{k}$, and orange for the number of cliques $m^k$.
  The values of $\delta$ and $r$ are held constant at $1/20$.
  For the relative precision required $r$, the value of $r_2$ (filtering error) is taken to be $r/20$, and $r_3$ (the amplitude estimation error) is taken to be $r\times 0.95$.}
  \label{Toffoli_counts}
\end{figure}

\subsection{\edr graphs}

The family of graphs in \subsec{K-graphs} is specifically constructed to have high Betti number and large spectral gap. One might wonder what speedups are \emph{generically} possible. To shed light on this question, we examine the \edr family of random graphs.

The \edr random graph $G(n,p)$ has $n$ vertices, and each of the $\binom{n}{2}$ edges is present i.i.d.\ with probability $p$. In \cite{Kahle}, the following theorem is established.

\begin{theorem}
Let $p = n^\alpha$. If $-1/k < \alpha < -1/(k+1)$, then
\begin{equation}
\frac{\beta_k}{\binom{n}{k+1} p^{\binom{k+1}{2}}} \rightarrow 1 \quad \text{almost surely}
\end{equation}
On the other hand, if $\alpha < -1/k$ or $\alpha > -1/(2k+1)$, then $\beta_k \rightarrow 0$ almost surely.
\end{theorem}

Taking $p = n^{-1/(k+\frac{1}{2})}$ gives $\beta_k \sim \binom{n}{k+1} n^{-k/2}$ almost surely. Ignoring the factor $n|E|/\lambda_{\min}$, our quantum algorithm can compute the $k^{\text{th}}$ Betti number in time scaling as $T_q \sim n^{k/4+2}$ for constant $k$.
For large $k$, where the $+2$ coming from {$|E^C|$} is negligible, this is approximately a quartic speedup.

\subsection{Rips complexes} \label{sec:Rips_complex_speedup}

One of the main applications of topological data analysis is to Rips complexes induced by finite-dimensional data in $\mathbb{R}^d$. This is another shortcoming of the graph family from \subsec{K-graphs} -- they are defined as abstract graphs, rather than being induced from finite-dimensional data. But are such speedups possible for Rips complexes? Unfortunately, there are results which prevent these large speedups.

It is shown in \cite{goff2009extremal} that, for any fixed $k$ and $d$
\begin{equation}
\max_{S \subset \mathbb{R}^d : |S|=n} \frac{\beta_k(\mathcal{R}_\epsilon (S))}{n^k}  \rightarrow 0 \ \text{as} \ n \rightarrow \infty
\end{equation}

In \cite{kahle2011random}, the author studies a setting where $n$ data points are drawn from a fixed underlying probability measure on $\mathbb{R}^d$. This is arguably the setting of interest in topological data analysis. They show that the Betti numbers of the derived Rips complexes have three `phases' depending on the scale $\epsilon$. (Recall that we include an edge if two points are within distance $\epsilon$.) For small $\epsilon = o(n^{-1/d})$, called the subcritical phase, the Betti numbers vanish asymptotically. Intuitively the complex is highly disconnected, since we are below the percolation threshold. There is a critical phase $\epsilon \sim n^{-1/d}$, where the Betti numbers will scale linearly $\beta_k \sim n$. Then for large $\epsilon = \omega(n^{-1/d})$, in the supercritical phase, the Betti number grows sublinearly $\beta_k = o(n)$. Thus in all regimes, the Betti number grows at most linearly in the number of points. This is of course far from the $n^k$ scaling needed for superpolynomial speedup.

However, it is possible to construct a Rips complex with large Betti number and large spectral gap, even in $\mathbb{R}^2$ \cite{goff2009extremal}. We describe such a Rips complex here.

Construct $S \subset \mathbb{R}^2$ as follows. Let $m = n/2k$, $\theta = \pi/k$, and $\delta = n^{-4}$. For $i=1,\dots,k$, let $x^+_i = (1/2, i\delta)$ and $x^-_i = (-1/2, i\delta)$. Let $S_0 = \{x^+_1,\dots,x^+_m,x^-_1,\dots,x^-_m\}$. For $j=1,\dots,k-1$, construct $S_j$ by rotating $S_0$ about the origin by an angle $j \theta$. Then finally $S = S_0 \cup \dots \cup S_{k-1}$. We will take the Rips complex $\mathcal{R}_1(S)$ with $\epsilon = 1$. $\{x^+_1,\dots,x^+_m\}$ and $\{x^-_1,\dots,x^-_m\}$ become $m$-simplices. There is an edge $(x^+_i,x^-_i)$ for every $i$, but no edges $(x^+_i,x^-_j)$ for $i\neq j$. Due to the small value of $\delta$, each $S_i$ is completely connected to every other $S_j$.

\begin{figure}[H]
\centering
  \includegraphics[width=2in]{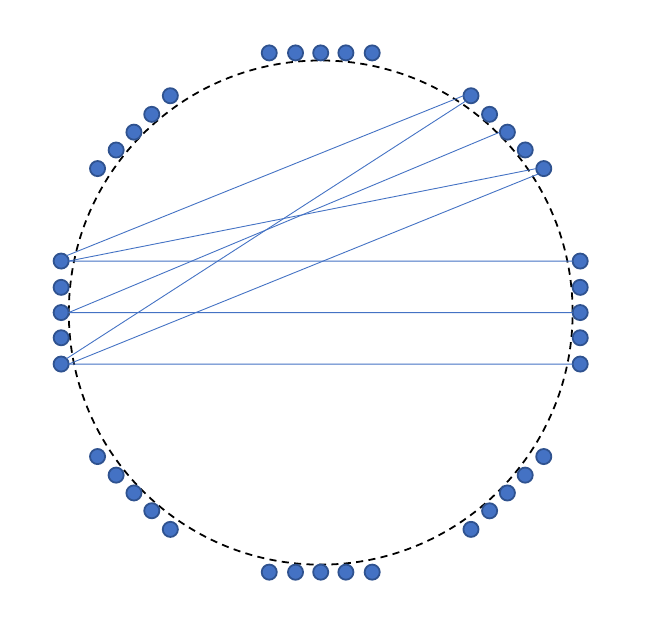}
  \caption{The Rips complex $\mathcal{R}_1(S)$.}
  \label{rips_complex}
\end{figure}

\begin{proposition} \label{prop:Betti_number_Rips}
The $(2k-1)^{\text{th}}$ Betti number of $\mathcal{R}_1(S)$ is
\begin{equation}
\beta_{2k-1}(S) = (m-1)^k = \left(\frac{n}{2k} - 1\right)^k .
\end{equation}
\end{proposition}

\begin{proposition} \label{prop:spectral_gap_Rips}
The combinatorial Laplacian $\Delta_k = \partial_k^{\dagger} \partial_k + \partial_{k+1}\partial_{k+1}^{\dagger}$ of $\mathcal{R}_1(S)$ has constant spectral gap $\lambda_{\min}$.
\end{proposition}

We prove these in \append{proof_betti_density} using techniques from simplicial homology. Our quantum algorithm can compute the $k^{\text{th}}$ Betti number in time scaling as $T_q \sim n^{3 + k/4}$ for constant $k$.

\subsection{Randomized classical algorithms for Betti number estimation}
\label{subsec:classical}

While the previous discussion shows there are cases where quantum algorithms can provide {superpolynomial} advantages with respect to deterministic classical algorithms for TDA, the question of how randomized classical algorithms perform in this setting is comparably understudied. There are works studying generalizations of random walk operators corresponding to higher order Laplacians for simplicial complexes \cite{cohen:walk, mukherjee:walk, parzanchevski:walk}. These are specific to the combinatorial Laplacian context and while they could lead to more efficient classical approaches, no such result is known at the present.

Here we show, perhaps surprisingly, that there exists a randomized classical algorithm which can compute normalized Betti numbers in the clique dense case using a polynomial number of operations under appropriate assumptions. This shows that the sufficient conditions needed for quantum algorithms to provide an advantage are more subtle than anticipated and that simply having a high-dimensional vector space does not necessarily guarantee a super-polynomial speedup.

The main idea of our algorithm is to use imaginary-time evolution to create a projector onto the kernel of $\Delta_{k-1}^G ={\partial_{k-1}^{G\dagger}} \partial_{k-1}^G + \partial_{k}^G {\partial_{k}^{G\dagger}}$. More specifically, we focus on the Dirac operator $B_G$ and look at simulating its imaginary  time dynamics of its square using path integral Monte-Carlo. We further simplify this approach by taking $\tilde{B}_G$ to be an analogous operator to $B_G$ except we now use an energy penalty to penalize any configuration that is not a clique or of the correct parity. In particular, 
\begin{equation}
    \TBG = B_G^2 + \gamma_{\min} (1-P),
\end{equation}
where $P$ as before is the projector onto the states of appropriate Hamming weight and configurations that correspond to a clique, and $\gamma_{\min}$ is an upper bound on the spectral gap of $B_G^2$ which coincides with the second smallest eigenvalue of the combinatorial Laplacian. Further, it is easy to see that if a vector is in the kernel of $\TBG$ it is also in the kernel of $B_G$.
 Following the same reasoning as before, as $B_G$ is Hermitian so is $\tilde{B}_G$ and thus it has a complete set of orthonormal eigenvectors.
This implies that any unit vector $\ket{\psi}$ which is supported on $\mathcal{H}_k^G$ can be decomposed as
\begin{equation}
    \ket{\psi} = \cos(\theta) \ket{\psi_g} + \sin(\theta) \ket{\psi_b},
\end{equation}
where $\ket{\psi_g}$ is the projection of $\ket{\psi}$ onto the kernel of $\TBG$ and $\ket{\psi_b}$ is its orthogonal complement.
Then
\begin{align}
    e^{-\TBG t} \ket{\psi} = \cos(\theta) \ket{\psi_g} + e^{-\TBG t} \sin(\theta) \ket{\psi_b} .
\end{align}
Let $\TBG\ket{\lambda_{\mu}} = \lambda_{\mu}\ket{\lambda_{\mu}}$ such that $\lambda_1\le \lambda_2 \le \cdots \le \lambda_{d_{k-1}}$, where $d_{k-1} = \binom{n}{k}$.  The operator $\TBG$ is positive semi-definite and thus
\begin{equation}
    \bra{\psi} e^{-\tilde{B^2_G} t} \ket{\psi} = \cos^2(\theta) +\sin^2(\theta) \bra{\psi_b} e^{-\Delta_{k-1}^G t}\ket{\psi_b} \le \cos^2(\theta) + \sin^2(\theta) e^{-\gamma_{\text{min}}t}.
\end{equation}
If we then pick 
\begin{equation}
    t\ge  \log(1/\epsilon)/\gamma_{\min}, \label{eq:tbd}
\end{equation} 
{where $\gamma_{\min}$ is the smallest non-zero eigenvalue of $\TBG$, the expectation value will be at most $\cos^2(\theta) + \mathcal{O}(\epsilon)$}.

If $\ket{\psi}$ is chosen such that it is a column of a Haar random unitary over the constrained parity subspace $\mathcal{H}_{k}^G$, the expectation value will be
\begin{equation}
    \mathbb{E}_{\rm Haar}(\cos^2(\theta))= \frac{\beta_{k-1}}{d_{k-1}}.
\end{equation}
Thus performing imaginary time evolution and a Haar expectation value will give the required normalized Betti number.

The remaining question centers around whether the imaginary time evolution can be performed on a classical computer in polynomial time. First, let us consider a decomposition of the Hamiltonian of the form 
\begin{equation}
    \TBG = \sum_{p =1}^D c_p H_p
\end{equation} where each $H_p$ is one-sparse, {Hermitian}, and unitary. Hence the eigenvalues of each are $\lambda_{p_i,\nu_i} = \pm c_p$, where $\nu_i$ is an index of the eigenvalue and $p_i$ is the index of the Hamiltonian \cite{berry2007efficient,BerrySTOC14}. The Jordan-Wigner decomposition on $B$ provides such a decomposition and the projector $P$ can always be written as a sum of a reflection over computational basis states and an identity gate, which provides an efficient decomposition into one-sparse Hermitian and unitary terms.

With this decomposition in hand, we focus on using a path-integral Monte-Carlo simulation of $\exp(-\TBG t)$. The path integral expansion works by first breaking up $e^{-\TBG t}$ into $r$ timeslices, Trotterizing over the matrices in our one-sparse {Hermitian} decomposition of $\TBG$ {(which is efficient to determine \cite{berry2007efficient,BerrySTOC14})}, and then expanding each one-sparse Hermitian matrix in its eigenbasis. Since one-sparse {Hermitian} matrices can be efficiently diagonalized, this process is classically efficient. {(Hermitian one-sparse matrices can be decomposed as a direct sum of 1 and 2-dimensional matrices, which are trivial to diagonalize.)
}

Let $\Gamma$ denote a particular path of eigenvectors in the path integral representation, $W(\Gamma)$ be the product of overlaps between the eigenvectors and $\lambda_{p_i,\Gamma_i}$ be the eigenvalue corresponding to the eigenvector that appears in the $i^{\rm th}$ step in the path $\Gamma$.  Finally, let ${\rm Pr}(\Gamma)$ be a probability distribution from which the paths are drawn that can be chosen to reduce the variance (as is standard in importance sampling).
We show in \append{PIMCdequant} that taking the Haar-expectation of the result leads to
\begin{equation}
    \mathbb{E}_{\rm Haar}(\cos^2(\theta))=\frac{1}{d_{k-1}} \mathbb{E}_\Gamma \left(\frac{\exp{\left(-\lambda_{p_1,\Gamma_1}t/r-\sum_{i=2}^{2rD-1}\lambda_{p_i,\Gamma_i}t/2r\right)} W(\Gamma) \delta_{k_1,k_{2rD}}}{\text{Pr}(\Gamma)} \right) .
\end{equation}

We then average over a finite ensemble of these random paths to estimate the expectation value drawn from an appropriate probability distribution. We propose, for general purposes, a Metropolis-Hastings based algorithm for selecting appropriate paths in the decomposition that are unlikely to have zero values of $W$. 
More specifically, the algorithm works by drawing an initial eigenstate of the first term in the one-sparse decomposition of the Dirac operator $B_G$ uniformly. Then a path is drawn by transitioning to one of the two possible connected eigenstates for it randomly. As the terms are Hermitian by assumption, all eigenvalues at each step in the path integral are the same up to a sign. This means that there are only two choices when constructing a path: either we choose to traverse the positive eigenvalue or the negative eigenvalue.  Thus each path can be described using $\mathcal{O}(\log(d_{k-1})rD)$ bits.  A path that has non-zero overlaps between the neighboring eigenstates can then be selected in $\mathcal{O}(\log(d_{k-1})rD)$ time.  This is used as an initial guess that is improved using Metropolis-Hastings, wherein the probability of transitioning between two randomly chosen paths $\Gamma^{(a)}$ and $\Gamma^{(b)}$ is:

\begin{equation}
    P(\Gamma^{(b)}|\Gamma^{(a)})=\frac{\exp{\left(-2\lambda_{p_1,\Gamma_1^{(b)}}t/r-\sum_{i=2}^{2rD-1}\lambda_{p_i,\Gamma_i^{(b)}}t/r\right)}}{\exp{\left(-2\lambda_{p_1,\Gamma_1^{(a)}}t/r-\sum_{i=2}^{2rD-1}\lambda_{p_i,\Gamma_i^{(a)}}t/r\right)}}. \label{eq:deltaPIBd0}
\end{equation}
The equilibrium distribution leads to a thermal distribution over the path $K$ with 
 \begin{equation}
     {\rm Pr}(\Gamma) =\frac{ \exp{\left(-2\lambda_{p_1,\Gamma_1}t/r-\sum_{i=2}^{2rD-1}\lambda_{p_i,\Gamma_i}t/r\right)}\delta_{\Gamma\in S_\Gamma}}{\sum_{\Gamma\in S_\Gamma}\exp{\left(-2\lambda_{p_1,\Gamma_1}t/r-\sum_{i=2}^{2rD-1}\lambda_{p_i,\Gamma_i}t/r\right)}},
 \end{equation}
 where $S_\Gamma$ is the set of all paths.  The number of such updates needed to achieve this distribution (within fixed error) scales as $\mathcal{O}(1/\gamma_M)$ where $\gamma_M$ is the gap of the Markov chain. We show in \append{PIMCdequant} that, provided the gap is large, this distribution can be efficiently sampled from and forms a good choice for the importance distribution for the paths that minimizes the variance over the paths.
 
We ultimately find that the number of arithmetic operations needed to estimate the ratio of the kernel to the size of the set of all $k$-simplices within additive error $\epsilon$ is, assuming that the sample variance in the estimates yielded by Algorithm~\ref{alg:classical} is $\sigma^2$, is in
 \begin{equation}
     \widetilde{\mathcal{O}}\left(\frac{\sigma^2}{\epsilon^2}\left(\frac{|E|d_{k-1}}{|{\rm Cl}_k(G)|}+\frac{D^4 }{\gamma_M}\frac{\kappa^3}{\epsilon} \left( \log(d_{k-1}) D^{-2} +  \frac{\kappa^3}{\epsilon} \right)\right)\right)
 \end{equation}
where $\kappa$ is the ratio of the largest eigenvalue to the smallest non-zero eigenvalue, i.e.\ the condition number of the combinatorial Laplacian, $|E|$ is the size of the edge set in the input graph, and $|\text{Cl}_k(G)|$ is the number of $k$-cliques in the input graph (this is \textit{not} equal to $d_{k-1}$ in general!). This shows that even in cases where the dimension is exponentially large, we can use path integration to estimate {the ratio $\dim \ker \Delta_k/d_{k-1}$} using a number of operations that scales polynomially with the number of vertices $n$ provided that $\sigma$, $D$, $\kappa$, $\gamma_{M}^{-1}$, and $|\mathrm{Cl}_k(G)|^{-1}$ are at most $\mathrm{poly}(n)$. We also show that $\sigma$ can be polynomially large in some cases in \append{PIMCdequant}. 
 
Quantum algorithms for TDA were thought to outperform classical counterparts in the clique dense case \cite{Casper}, but this algorithm serves as a counter-example. Another key point behind this dequantization result is that while an exponentially large dimension is a necessary condition for an exponential speedup for quantum TDA, it is not a sufficient condition. This implies that further work is needed in order to understand when, and even if, quantum algorithms can provide truly exponential advantages relative to all classical randomized algorithms for TDA.
\begin{algorithm}[h!]
\caption{Classical randomized algorithm for Betti number computation.\label{alg:classical}}
\SetAlgoNoLine
\DontPrintSemicolon
\KwData{$k>0, n>0, N_{\rm samp}>0$, $t\ge 0$, $r\ge 0$, a function ${\rm Pr}(\Gamma)$ which assigns a non-zero probability to each vector $\Gamma\in \mathbb{R}^{2rD}$, a function $W(\Gamma)=\braket{\lambda_{p_1,\Gamma_1}}{\lambda_{p_2,\Gamma_2}}\cdots \braket{\lambda_{p_{2rD-1},\Gamma_{2rD-1}}}{\lambda_{p_{1},\Gamma_{1}}}$ where $\ket{\lambda_{p_j,\Gamma_j}}$ is the $\Gamma_j^{\rm th}$ eigenvector of the one sparse matrix $U_{p_j}$.}
\KwResult{Estimate $\bar{E}$ which is an unbiased estimator of $\beta_{k-1}/d_{k-1}$}
\For{$q$ from $1$ to $N_{\rm samp}$}{
\qquad$\Sigma \gets$ a set of $k$ points encoded as an integer\;
\qquad{\textbf{while}} {$\Sigma$ is not a $(k-1)$-simplex }\textbf{do}\;
{\qquad\qquad $\Sigma \gets$ a random set of $k$ points encoded as an integer\;}
\qquad \textbf{end}\;
\qquad Draw a vector $\Gamma= [\Sigma,\Gamma_2,\ldots, \Gamma_{2rD}]$ from the probability distribution ${\rm Pr}(\Gamma)$.\;
\qquad $E_q \gets \frac{1}{d_{k-1}}\left(\frac{\exp{\left(-\lambda_{p_1,\Gamma_1}t/r-\sum_{i=2}^{2rD-1}\lambda_{p_i,\Gamma_i}t/2r\right)} W }{{\rm Pr}(\Gamma)} \right)$
}
$\bar{E} \gets \frac{1}{N_{\rm samp}} \sum_q E_q$ average of $E$.
\end{algorithm}

Apers \textit{et al}.~\cite{Apers22} subsequently developed another classical path-integral Monte Carlo algorithm which can also be efficient in some of the regimes where the quantum algorithm for Betti number estimation works best. It is therefore natural to investigate how it fares compared to our classical and quantum algorithms. They study additive error estimation of the normalized Betti number, so we focus on comparing the runtime of their algorithm for this case.

Apers \textit{et al}. claim a runtime of $n^{\mathcal{O}(\gamma^{-1} \log (1/\epsilon))}$ for general simplicial complexes, where $n$ is the number of vertices, $\epsilon$ the additive error, and $\gamma \leq \lambda_2(\Delta_k)/\hat{\lambda}$ with $\hat{\lambda} \in \Theta(\lambda_{\max}(\Delta_k))$. For clique complexes in particular, their runtime is generically $\mathrm{poly}(n)\cdot (n/\hat{\lambda})^{\mathcal{O}(\gamma^{-1} \log (1/\epsilon))}$, or $2^{\mathcal{O}(\gamma^{-1} \log 1/\epsilon)}$ when $k$ or the maximum up-degree (see the last section of \Cref{append:PIMCdequant} for the definition) of the $k$-simplices are $O(n)$. 

For general simplical complexes, their runtime depends exponentially on $\gamma^{-1}$ of $\Delta_k$, which is an upper bound on the condition number $\kappa$. They thus require constant $\gamma$ and $\epsilon$ for general simplicial complexes. By contrast, our classical algorithm above has a run time depending on fixed polynomials in $\gamma^{-1}$ and $\epsilon^{-1}$, and can thus tolerate $\mathrm{poly}(n)$ scaling for both provided $\sigma^2$ also depends polynomially on $\kappa$ and $1/\epsilon$. For clique complexes, their best-case runtime is polynomial if $\gamma \in \Omega(1)$ and $\epsilon = 1/\mathrm{poly}(n)$, or $\gamma \in \Omega(1/\log n)$ and $\epsilon \in \Omega(1)$. Analogous conclusions hold for our classical algorithm in these regimes if we are once again given suitable promises on the polynomial dependence of $\sigma^2$ on those parameters. Thus given the difficulty in analyzing the dependence of $\sigma^2$ on $\kappa$ and $\epsilon^{-1}$ in general, our algorithms cannot be easily compared. In the worst case bound for $\sigma^2$ given in \eqref{eq:varworst}, we would require $D\kappa$ to be poly-logarithmic in $n$ and $\kappa/\epsilon$ a constant at worst and thus also cannot tolerate inverse polynomial error scaling or spectral gap.

The algorithm of Apers \textit{et al}.\ can efficiently compute a \emph{constant} additive error estimate of the normalized Betti number but cannot efficiently compute an \emph{inverse polynomial} additive error estimate for the graphs $K(n/k, k)$ discussed in this work. This yields an immediate exponential separation with the quantum algorithm since the latter is able to efficiently compute such an inverse polynomial additive error estimate. Moreover, even stronger separations are possible. As noted above, the runtime of the algorithm of Apers \textit{et al}.\ depends exponentially on the inverse of the (normalized) spectral gap of the combinatorial Laplacian, unlike in the case of the quantum algorithm whose runtime only depends polynomially on this parameter. 

For the family $K(n/k, k)$, this (normalized) spectral gap is large and suitable for the algorithm of Apers \textit{et al}. There are nevertheless many graphs that do have a (normalized) spectral gap which incurs an additional exponential scaling for the case of the classical algorithm but not for the quantum one. 
For example, by adding a single edge to each cluster in $K(n/k, k)$ the (normalized) spectral gap of the combinatorial Laplacian becomes much smaller, which causes the algorithm of Apers \textit{et al}.\ to no longer be able to efficiently compute even a constant additive precision estimate of the normalized Betti number.
On the other hand, the quantum algorithm can still efficiently compute an inverse polynomial additive error estimate as can our classical algorithm given similar guarantees on the polynomial scaling of $\sigma^2$ with $1/\gamma$. In Appendix \ref{app:new_graph_examples}, we detail some of these modified graph examples with smaller spectral gaps.

\section{Conclusion}
\label{sec:conc}
In order to provide applications where quantum computers can practically outperform classical computers on hardware anticipated in the near-future, it is necessary to develop algorithms where there is a greater than square-root speedup in the complexity \cite{QuadSpeedup}.
This is because the large overheads involved in implementing quantum gates in an error-corrected code mean that there is a huge slowdown in the gate frequency as compared to classical computers.
When the Betti number is of order 1, the complexity of the quantum algorithm for estimating Betti numbers is only a square root speedup over classical approaches.
This is as compared to classical approaches that scale approximately linearly in $\binom{n}{k}$.
On the other hand, when the Betti number is large, the {quantum} complexity of estimating the Betti number to given \emph{relative} error (error as a ratio to the Betti number) will be small.

There exist classes of graphs with very large Betti numbers.
{We introduce graph classes which exhibit Betti numbers in the regime where} the speedup is superpolynomial; specifically it is approximately a $2\ln n$ root.
The magnitude of the speedup is limited by the need to enter the data in the quantum algorithm, which introduces a $|E|$ factor to the complexity.
These are very specially constructed graphs for large Betti number, but we show there exist far more general classes of graphs {whose parameters give a \emph{quartic} speedup over naive classical algorithms}, showing that speedups beyond quadratic are possible far more generally.

We have also provided a host of new techniques for quantum Betti number estimation that reduce the complexity.
These include Kaiser-window amplitude estimation, improved Dicke state preparation, and improved eigenstate filtering.
These improvements greatly improve the complexity in many ways, though the main scaling of the complexity as $\sqrt{\binom{n}{k}}$ remains.
{In particular, we have major improvements in the complexity arising from improved Dicke state preparation as well as improved amplitude estimation, which together give about 3 orders of magnitude improvement.
We have further improvements arising from our clique checking procedure, separation of amplitude estimation and amplification in initial state preparation, and use of filtering instead of phase estimation, which may give a further order of magnitude of improvement depending on how prior work is interpreted.
Moreover, our}
methods enable accurate estimation of the complexity of the quantum algorithm, including all constant factors.

Based on that, we estimate that tens of billions of Toffolis would be sufficient to estimate a Betti number that should be classically intractable.
This number of Toffoli gates is reasonable for early generations of fully fault-tolerant quantum computers. While the exact threshold for quantum advantage depends on constant factors in the classical algorithms, it seems likely that this application will fall somewhere in between quantum chemistry applications and Shor's algorithm in terms of the resources required for quantum advantage.
The standard classical approaches would be expected to be intractable because they would require an extremely large storage.
We have also presented an alternative approach for classical estimation that could be more efficient, because it does not require large storage and instead requires Monte-Carlo sampling.
That classical approach may be tractable, but it is difficult to evaluate its complexity because it depends on the gap of a Markov chain which is unknown.

There is scope for further improvement of the quantum algorithm for Betti numbers by implementing a more efficient method of clique finding.
We have currently applied just amplitude amplification for clique finding, but there are more efficient classical methods for clique finding that could potentially be adapted for the quantum algorithm.
That is nontrivial because these methods often require large storage, which would not be practical in the quantum algorithm where we need to minimize the number of ancilla qubits.

\section*{Acknowledgements}
	
	The authors acknowledge helpful conversations with David Gamarnik, Robin Kothari, Seth Lloyd, Alexander Schmidhuber, Nikhil Srivastava and Adam Zalcman. DWB worked on this project under a sponsored research agreement with Google Quantum AI. DWB is also supported by Australian Research Council Discovery Projects DP190102633 and DP210101367. NW and VD were funded by a grant from Google Quantum AI. NW was also funded by grants from the US Department of Energy, Office of Science, National Quantum Information Science Research Centers, Co-Design Center for Quantum Advantage under contract number DE-SC0012704. VD and CG were supported by the Dutch Research Council (NWO/ OCW), as part of the Quantum Software Consortium programme (project number 024.003.037). Some of the discussions and collaboration for this project occurred while using facilities at the Kavli Institute for Theoretical Physics, supported in part by the National Science Foundation under Grant No.~NSF PHY-1748958.

\bibliography{ref,ryan}

\appendix

\section{Detailed background on topological data analysis}
\label{append:background}

We present some background material from singular homology needed in topological data analysis, broadly following the treatments in \cite{lee2010,McGuirl2017,ghrist2008barcodes}. 

Let $v_0, \ldots, v_k$ be $k+1$ distinct points in $\mathbb{R}^n$. The set $\{v_0,\ldots,v_k\}$ is said to be \textbf{affinely independent} if the set $\{v_1 - v_0, \ldots, v_k - v_0\}$ is linearly independent. In other words, we consider the given set of points to be affinely independent if when we take one of the points to be the ``origin'' (say $v_0$ WLOG) and draw vectors from this point to the others, the collection of the resulting vectors is linearly independent. 

If $\{v_0,\ldots,v_k\}$ is affinely independent, the $\mathbf{k}$\textbf{-simplex} spanned by them is the set 
\begin{equation}
    [v_0,\ldots,v_k] \coloneqq \bigg \{ \sum_{i=0}^k t_i v_i \colon t_i \geq 0 \text{ and } \sum_{i=0}^k t_i = 1 \bigg \} \, .
\end{equation}
Equivalently, a simplex is just the convex hull of its affinely independent set of vertices. The points $v_i$ are the \textbf{vertices} of the simplex and the integer $k$ is the \textbf{dimension} of the simplex. \autoref{fig:simplices} shows some examples of simplices. 
Note that it follows from the definitions given that $k \leq n$ since any set of $n+2$ points or more cannot be affinely independent. This is because no collection of $n+1$ vectors or more in an $n$-dimensional vector space can be linearly independent. Such a collection of vectors therefore cannot determine any simplicies of dimension $n+1$ or higher.

Let $\sigma$ be a $k$-simplex. A simplex spanned by a non-empty subset of the vertices of $\sigma$ is a \textbf{face} of $\sigma$. For example, the $0$-dimensional faces of $\sigma$ are its vertices and its $1$-dimensional faces are the edges, which are spanned by two vertices. Faces of $\sigma$ that are not equal to $\sigma$ are called \textbf{proper faces}. The $(k-1)$-dimensional faces of $\sigma$ are called its \textbf{boundary faces} and their union is its \textbf{boundary}. 

\begin{figure}[ht]
\includegraphics[width=5cm]{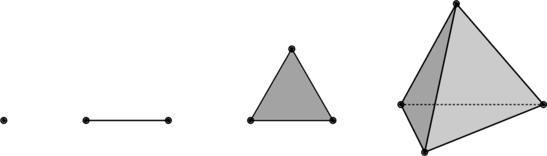}
\centering 
\caption{A 0-simplex (point), 1-simplex (edge), 2-simplex (triangle), and 3-simplex (tetrahedron) from left to right.}
\label{fig:simplices}
\end{figure}

\begin{definition}[\textbf{Simplicial Complex}]
\label{defn:simplicial}
    A simplicial complex $S$ is a finite collection of simplices satisfying the following conditions: 
    \begin{enumerate}
        \item If $\sigma \in S$, every face of $\sigma$ is in $S$
        \item If $\sigma_1, \sigma_2 \in S$, then $\sigma_1 \cap \sigma_2 = \varnothing$ or $\sigma_1 \cap \sigma_2$ is a face of $\sigma_1$ and $\sigma_2$
    \end{enumerate}
\end{definition}

The first condition says that a simplicial complex should also contain all the faces of a given simplex in the complex. The second condition says that any two simplices in a simplicial complex either do not intersect or intersect at a common face of both. We define the \textbf{dimension of a simplicial complex} to be the maximum of the dimensions of all simplices in the complex. 

\begin{figure}[ht]
\centering
  \includegraphics[width=.2\linewidth]{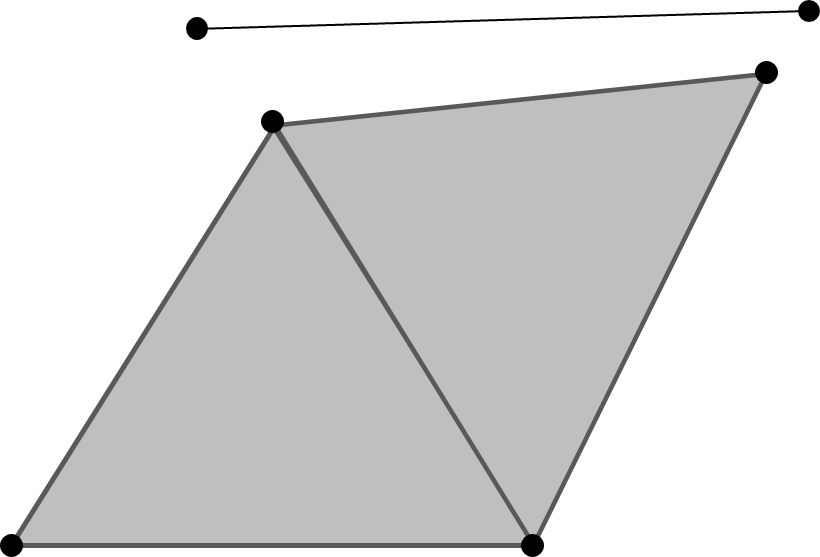}
\hspace{5cm}
  \includegraphics[width=.2\linewidth]{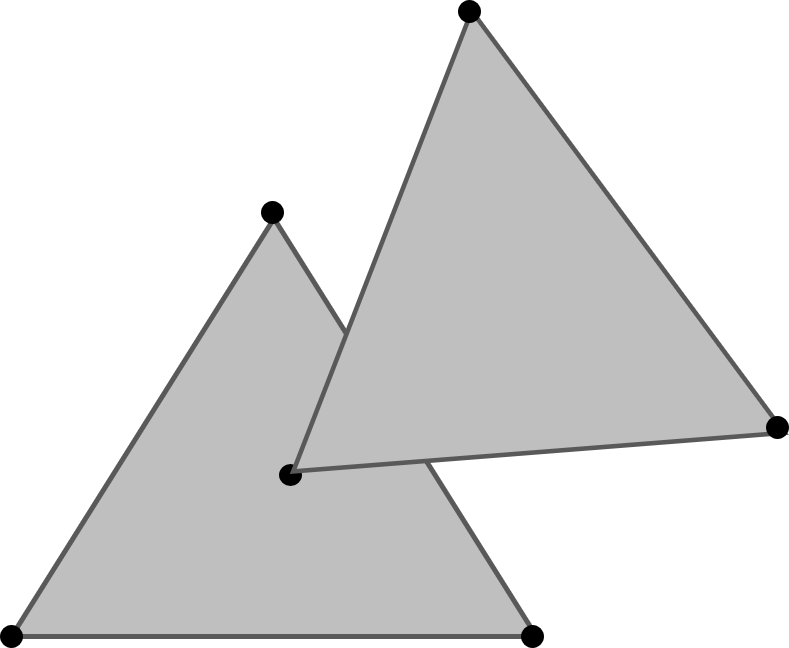}
\caption{(Left) A 2D simplicial complex in $\mathbb{R}^2$. (Right) A set that is not a simplicial complex in $\mathbb{R}^2$. It violates condition 2 of \defn{simplicial}.}
\label{fig:exampsimpcomplexes}
\end{figure}

It can be shown that a simplicial complex is completely determined by its vertices and information about which sets of vertices span which simplices. This provides the motivation for the following definition.

\begin{definition}[\textbf{Abstract Simplicial Complex}]
An abstract simplicial complex is a collection $C$ of non-empty finite sets such that if $s \in C$, then every non-empty subset of $s$ is also in $C$. 
\label{defn:abssimpcomp}
\end{definition}

This general notion of a simplicial complex is particularly useful when we wish to construct one ``abstractly", i.e.\ without reference to a particular embedding into Euclidean space. 

We will mainly be concerned with computing topological invariants of certain kind of simplicial complexes constructed from graphs, called clique complexes, throughout this paper. The graphs serving as inputs into the quantum algorithms in this paper are not necessarily induced by any finite-dimensional data and are instead abstract graphs, i.e.\ abstract simpicial complexes in the sense of \defn{abssimpcomp} in the preceding section.

\begin{definition}[\textbf{Graphs}]
{A graph $G$ is a pair of objects $G = (V,E)$, where $V$ is a set of elements referred to as the ``vertices" of $G$ and $E$ is a set consisting of pairs of vertices thought of as ``edges" connecting the pairs of vertices.}
\end{definition}

{Given an undirected graph (i.e.\ one in which the edges are not assumed to have direction), a \textbf{clique} $C$ of a graph is a subset of $V$ such that every pair of distinct vertices in $C$ is connected by an edge. $C$ is called a $k$-clique if $|C| = k$.} 

{We now define the notion of a clique complex:}

\begin{definition}[\textbf{Clique Complex}]
    {The clique complex of a graph $G$ is the abstract simplical complex formed by associating a $k$-simplex to every $k+1$-clique in $G$.}
\end{definition}

The invariants we are most interested in are the Betti numbers of a clique complex associated to a graph, which give the number of holes of a given dimension in that clique complex. We now show how to determine the Betti numbers of an arbitrary simplicial complex via simplicial homology.

{Let $K$ be a simplicial complex and consider the set of $k$-simplices in $K$. In what follows, we would like to make sense of taking ``linear combinations" of the $k$-simplices in $K$ with coefficients in some field $R$ (we will only need $R = \mathbb{R}$ or $R = \mathbb{C}$ for our purposes).} A \textbf{k-chain} is a formal sum of $k$-simplices $\sum_i c_i \sigma_i$ where $\sigma_i \in K$, $c_i \in R$. 
The set of all $k$-chains is denoted by $C_k(K)$ and is a vector space over $R$. The $k$-simplices form a basis for $C_k(K)$, so the dimension of $C_k(K)$ equals the number of $k$-simplices in $K$. 

\begin{definition}[\textbf{Boundary Map}]
Let $\sigma = [v_0,\ldots,v_k]$ be a $k$-simplex. The boundary map on $k$-simplices is a map 
    \begin{equation}
    \partial_k \colon C_k(X) \rightarrow C_{k-1}(X)
    \end{equation}
    that acts as
    \begin{equation}
    \partial_k \sigma = \sum_{i=0}^k (-1)^i [u_0,u_1,\ldots,\hat{u}_i,\ldots,u_k]
    \end{equation}
    where $\hat{u}_i$ denotes that the vertex $i$ has been removed.
\end{definition}

The boundary map acts on $k$-simplices $\sigma \in C_i(K)$ and gives a $(k-1)$-simplex $\partial_k \sigma$ that can be interpreted as the boundary of $\sigma$. 

An $\mathbf{k}$\textbf{-cycle} is a $k$-chain $c \in C_k(K)$ such that $\partial_k c = 0$. Therefore, $k$-cycles are precisely the kernel of the boundary map and are a subspace of $C_k(K)$ denoted by $Z_k = \ker \partial_k$. An $k$-chain $c$ is an $\mathbf{k}$\textbf{-boundary} if there exists an $(k+1)$-chain $\sigma \in C_{k+1}(K)$ such that $c = \partial_{k+1}(\sigma)$. Equivalently, $k$-boundaries are precisely the image of the boundary map and form a subspace denoted by $B_k(K) = \text{Im} \ \partial_{k+1}$. \autoref{fig:1boundaryand1cycle} shows an example of a 1-boundary and 1-cycle. 

\begin{figure}[ht]
\includegraphics[width=6cm]{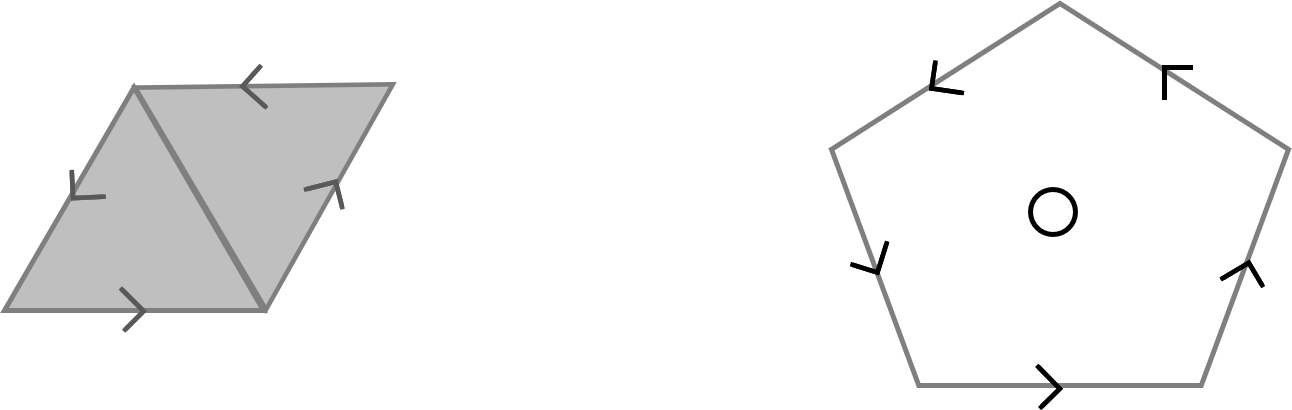}
\centering 
\caption{(Left) The 1-boundary of a 2-chain as indicated by the arrows. (Right) A 1-cycle that is not the boundary of any 2-chain.}
\label{fig:1boundaryand1cycle}
\end{figure}

It turns out that there is a relationship between the two subspaces as implied by the following fundamental result in homology theory. 

\begin{proposition}
    $\partial_k \circ \partial_{k+1}(\sigma) = 0$ for all $k+1$ chains $\sigma$ and all $0 \leq k \leq \dim K$. 
\end{proposition}

The above proposition essentially says that the boundary of a boundary is 0. It implies that the image of $\partial_{k+1}$ is contained in the kernel of $\partial_k$. This allows us to define the homology groups as follows:

\begin{definition}[\textbf{Homology Groups}]
    The $k$-th singular homology group $H_k$ of a simplicial complex is the quotient vector space \begin{equation}
    H_k(K) = Z_k(K)/B_k(K) = \mathrm{Ker} \ \partial_k / \mathrm{Im} \ \partial_{k+1} \, .
    \end{equation}
\end{definition}

\begin{definition}[\textbf{Betti Numbers}]
    The $i$-th Betti number $\beta_i$ is the dimension of the $i$-th homology group. 
\end{definition}

The $i$-th homology group is generated by cycles that are not the boundaries of any simplex. In other words, these are simplices that enclose a ``void'' or ``hole'' (see the right image of \autoref{fig:1boundaryand1cycle} above). It is worth noting that $\beta_0$, the $0^{\text{th}}$ Betti number, represents the number of connected components the simplicial complex has. 

The problem of computing the Betti numbers of a simplicial complex therefore reduces to the problem of computing the rank of the boundary map. A common and simple classical approach to doing this at the computational level is as follows.

Let $K$ be a simplical complex and assume for simplicity we are working over the field $\mathbb{Z}_2$. Label the $p$-simplices in $C_p(K)$ by $x_1,\ldots,x_{n_p}$ and the $(p-1)$-simplices in $C_{p-1}(K)$ by $y_1,\ldots,y_{n_{p-1}}$. These simplices form bases for $C_p(K)$ and $C_{p-1}(K)$ as mentioned previously. We can then represent the action of the boundary map $\partial_p$ on $C_p(K)$ as follows

\begin{equation}
    \partial_p(x_j) = \sum_{i=1}^{n_{p-1}} a^i_j y_i \text{ where } a^i_j =
\begin{cases}
    1& \text{if } y_i \text{ is a face of } x_j\\
    0              & \text{otherwise}
\end{cases}    
\end{equation}

Then for any $p$-chain $c = \sum_{j=1}^{n_p} a_j x_j$, we can write the above in matrix form 

\begin{equation}
\partial_p c = \begin{bmatrix} 
    a_1^1 & a_1^2 & \dots & a_1^{n_p} \\
    a_2^1 & a_2^2 & \dots & a_2^{n_p} \\
    \vdots & \vdots & \ddots & \vdots \\
    a_{n_{p-1}}^1 & a_{n_{p-1}}^2 & \dots & a_{n_{p-1}}^{n_p} 
    \end{bmatrix}
    \begin{bmatrix}
           a_1 \\
           a_2 \\
           \vdots \\
           a_{n_p} 
         \end{bmatrix} \, .
\end{equation}

Thus the boundary map on $C_p(K)$ can be represented as an $n_{p-1} \times n_p$ sparse matrix with entries in $\mathbb{Z}_2$. The columns of this matrix span $\text{Im } \partial_p = B_{p-1}$, so $\text{rank } \partial_p = \dim B_{p-1} = b_{p-1}$. The boundary matrix can brought into the Smith Normal Form via a generalization of Gaussian elimination which applies to any principal ideal domain (this includes fields). This results in a matrix with a number of 1's on the diagonal equal to the rank of the boundary matrix. Gaussian elimination for $\partial_p$ takes $\mathcal{O}(n_{p-1}n_p \min(n_p,n_{p-1}))$ time and requires $\mathcal{O}((n_{p-1} + n_p)^2)$ memory \cite{zomorodian2005,edelsbrunnerTop}.  {When working over arbitrary fields, the same procedure holds with the exception that the matrix elements of the boundary operator can be $\pm 1$ or $0$.} 

{In the worst-case scenario however, the number of $p$ simplices grows exponentially with the number of points in the complex, so this procedure becomes intractable for large data sets. This motivates the problem of finding efficient classical and quantum algorithms for extracting Betti numbers of simplicial complexes.} 

\section{Schemes for Dicke state preparation}
\label{append:Dicke}
\subsection{First scheme}\label{append:Dicke1}
{In our first scheme we prepare $n$ registers with $n_{\rm seed}:=\lceil\log cn\rceil$ qubits in equal superposition.
Then the principle is to find a threshold such that $k$ registers are less than or equal to this threshold, and the inequality test is used to set the ones in the Dicke state.
In the limit of large $n_{\rm seed}$, there is very low amplitude for the registers to be equal to each other, which can cause a failure.
The problem is that making $n_{\rm seed}$ larger increases the complexity.
Reducing $n_{\rm seed}$ will reduce the complexity but increase the amplitude for failure.
We will adjust the value of $c$ in order to make both the complexity and probability of failure reasonably small.

The steps of the} procedure are as follows.
\begin{enumerate}
    \item Prepare $n$ registers in an equal superposition of {$2^{n_{\rm seed}}$ values.}
    This may be done with Hadamards (no non-Clifford gates).
    \item Sum the most significant bits of these registers. This may be done with no more than $n$ Toffolis.
    This gives the number of registers at least as large as $10000\ldots$.
    \item Perform an inequality test with $k$ (with $\log n$ Toffolis) and place the result in an ancilla.
    That is, we are testing if the number is $\ge k$, and we call the bit representing the result $b_1$.
    \item For $j=2$ to $n_{\rm seed}$, we perform the following.
    \begin{enumerate}
        \item Perform an inequality test on each register checking if the first $j$ bits are $\ge b_1\ldots b_{j-1}1$.
        The cost is $(j-1)n$ Toffolis.
        \item Sum the results of those inequality tests, with cost $n$.
        \item Perform an inequality test if that number is $\ge k$, setting the result as $b_j$, with cost $\log n$.
    \end{enumerate}
    \item Perform an inequality test if all $n_{\rm seed}$ bits are $\ge b_1\ldots b_{n_{\rm seed}}$.
    The cost is $n_{\rm seed} n$ Toffolis.
    \item Sum the results of those inequality tests, with cost $n$.
    \item Check if the result is equal to $k$, with cost $\log n$.
\end{enumerate}

The logic of this procedure is that whenever {there are more than $k$ registers greater than or equal to the threshold,} we increase the threshold which is given by the bits $b_1b_2\ldots$, otherwise we reduce it.
{This test corresponds to the inequality test $\ge k$ in part 4(c), and the threshold is adjusted by choosing $b_j$.}
At the end, provided we have success where there sum of the ones was equal to $k$, then we have $k$ ones in a superposition of permutations corresponding to a Dicke state.
These are obtained by the inequality tests in step 5, and we check there are exactly $k$ ones in steps 6 and 7.
{The prepared state} is entangled with the ancilla registers, but this is suitable for our application.
Summing all the costs in this procedure gives a total Toffoli cost
\begin{equation}
    (n_{\rm seed}+1)\left[\frac n2  (n_{\rm seed}+2)+ \lceil \log n\rceil\right].
\end{equation}

{
To derive the probability of success, note that the failure occurs where it is not possible to provide a threshold where there are $k$ numbers greater than or equal to it.
That will be true if the $k$th smallest and $k+1$th smallest numbers are equal.
In other words, if we were to sort our numbers, and compare the values in register $k$ and register $k+1$, there would be failure if these numbers were equal.
(Note that we do not perform this sort in practice.)

The success probability for a given $k$ may be given in the form of a sum as
\begin{equation}\label{eq:exactprob}
    \frac 1{(cn)^n} \binom{n}{k} \sum_{\ell=1}^{cn} [\ell^k-(\ell-1)^k] (cn-\ell)^{n-k} \, .
\end{equation}
This expression is obtained by considering values up to $\ell$ in registers $1$ to $k$, and greater than $\ell$ for registers $k+1$ to $n$.
The number of combinations of $k$ registers with a maximum of $\ell$ is $\ell^k-(\ell-1)^k$, and the number of combinations of $n-k$ registers with values above $\ell$ is $(cn-\ell)^{n-k}$.
The factor of $\binom{n}{k}$ is to account for the number of positions to choose the $k$ registers with maximum value $\ell$.
Summing over $\ell$ then gives the number of combinations of values where there is a success when sorting the numbers, and we divide by the total number $(cn)^n$ to give the probability.

This sum can be approximated by an integral to give
\begin{align}\label{eq:intaprx}
& \binom{n}{k} \sum_{\ell=0}^{cn} \left(\frac {\ell}{cn}\right)^k \left(1-\frac {\ell}{cn}\right)^{n-k}
- \binom{n}{k} \sum_{\ell=1}^{cn} \left(\frac {\ell}{cn}-\frac {1}{cn}\right)^k \left(1-\frac {\ell}{cn}\right)^{n-k} \nn
    &\approx cn \binom{n}{k} \int_{0}^{1} dx \, x^k (1-x)^{n-k} - cn \binom{n}{k} \int_{1/cn}^{1} dx \, \left(x-\frac {1}{cn}\right)^k (1-x)^{n-k} \nn
    &= \frac{cn}{n+1} - \frac{(cn-1)^{n+1}}{(cn)^{n}(n+1)} \nn
    &= \frac{cn}{n+1} \left[ 1- \left(1-\frac 1{cn}\right)^{n+1} \right] \nn
    &\ge 1 - \frac {1}{2c} \, .
\end{align}
The integrals are evaluated by repeated integration by parts which cancels the factor of $\binom{n}{k}$.
This shows that the exact expression is well approximated by an upper bound of $1-1/2c$.
See Appendix \ref{sec:probs} for analysis of the error in this approximation, showing that it is of order $\mathcal{O}(c^{-(\min(k,n-k)+1)})$.
Moreover, we show that $1-1/2c$ is an exact lower bound to the success probability.
}

In practice, we find that moderate constant values of $c$ are suitable for minimising the complexity, and we will take $c=8$ in examples below.
That only increases the cost about $3\%$, while only needing another 3 qubits to store the numbers.
The result of taking $c$ constant is that the Toffoli complexity is approximately $(n/2)\log^2 n$.
In comparison, in the sorting approach from \cite{BerryNPJ18}, the complexity is approximately $2 n_{\rm seed} n \log n$, where $n\log n$ is the number of steps in the sorting network.
Moreover, that approach fails when \emph{any} {of the numbers in the sorted registers} have equal values, which requires taking {$2^{n_{\rm seed}}\gtrsim n^2$} to provide a reasonable probability of success.
That results in a preparation complexity of approximately {$4n\log^2 n$, or 8} times what we have provided here.

As a simple example where our threshold procedure is performed, consider $n=4$ and $c=4$, so we have 4 registers of 4 qubits.
An example of a basis state is
\begin{equation}
    \ket{0110}  \ket{1110}  \ket{0111}  \ket{0010}.
\end{equation}
Now say we aim for $k=2$.
The steps are as follows.
\begin{itemize}
    \item Sum the first (most significant bits) giving 1.
    This is equivalent to checking how many registers are at least $1000$.
    The sum is 1, which is less than $k$, so $b_1=0$.
    \item Perform an inequality test on the first two bits of each register with $01$.
    We get a total of 3, which is greater than $k$, so $b_2=1$.
    \item Perform an inequality test of the first three bits with $011$.
    We get 3 again, so $b_3=1$.
    \item Perform an inequality test with $0111$.
    We get a total of 2, which is equal to $k$, so we get $b_4=1$.
    \item Perform an inequality test with $0111$ again.
    Now test equality with $k$, which succeeds, so we get overall success.
\end{itemize}
As an alternative to the above example, consider the case that the basis state is
\begin{equation}
    \ket{0110}  \ket{1110}  \ket{0110}  \ket{0010}.
\end{equation}
Then in the second-last step, the total would be 1, so we would have $b_4=0$, and in the last step we would test inequality with $0110$.
That would give three greater than or equal to $0110$, so we would fail the test that the number is equal to $k$.
In this case, if you sort the numbers you have $0010$, $0110$, $0110$, $1110$, and you can see that the second and third numbers are equal.
This is the condition for failure discussed above.

\subsection{Second scheme}\label{append:Dicke2}
Here we provide an alternative scheme for Dicke state preparation that {can be more efficient,} but provides a lower amplitude for success.
It is useful in the case of large $n$ but small $k$ (smaller than $\sqrt n$), where $\binom{n}{k}\sim n^k/k!$.
The steps of this scheme are as follows.
\begin{enumerate}
    \item First prepare $k$ registers in equal superposition states over $n$ values.
    The complexity of preparing an equal superposition over $n$ values (with high probability of success) is $4\lceil \log n\rceil+1$ Toffolis \cite{SandersPRQ20}, so this would give complexity $k(4\lceil \log n\rceil+1)$.
    In the case where $n$ is a power of 2, then this preparation can be performed just with Hadamards.
    \item For each of the $k$ registers, apply unary iteration as in \cite{BabbushPRX18} with the $n$ qubits for the Dicke state as the target.
    This is used to flip the corresponding qubit.
    The complexity of each unary iteration is $n-2$, giving total complexity $k(n-2)$.
    There will be approximately $\log n$ temporary ancillas used in the unary iteration that are reset to zero.
    \item Now we sum the bits in the string of $n$ qubits.
    A method of summing bits is given in \cite{SandersPRQ20}, where multiple groups of bits are summed, and their sums are summed.
The overall complexity is no more than $2n$ Toffolis, and only a logarithmic number of ancillas is used.
The ancillas used that need to be kept for later stages of the calculation can be given as $\lceil \log n\rceil$, and there are $2\lceil \log n\rceil$ temporary ancillas used.
\item The sum of the bits is compared to $k$.
This has complexity $\lceil \log k\rceil$ because we are guaranteed that the number of ones is at most $k$.
\end{enumerate}

As before, this is a scheme which prepares the Dicke state in an entangled state with an ancilla.
The overall complexity is then
\begin{equation}
    (k+2)n + k(4\lceil \log n\rceil-1) + \lceil \log k\rceil
\end{equation}
Toffolis with the number of ancillas used being
\begin{equation}
    (k+1)\lceil \log n\rceil
\end{equation}
with a logarithmic number of working ancillas as well.
In comparison, the scheme used above uses a number of qubits scaling as $n\log n$, so is much larger when $n\gg k$.
In the case where $n$ is a power of 2, the initial preparation of the equal superposition can be just performed with Hadamards, and so the Toffoli complexity is
\begin{equation}
    (k+2)n - 2k + \lceil \log k\rceil.
\end{equation}

There is a probability of failure for the preparation, because it is possible for the ones to overlap, which will be detected in the final stage where the sum of bits is compared to $k$.
This is an example of the birthday problem, and the probability of success will be at least $1/2$ when $k$ is not larger than approximately $\sqrt n$.
The probability of success is more specifically given by
\begin{equation}
    \frac{k!}{n^k} \binom{n}{k}.
\end{equation}
It is possible to perform amplitude amplification on this step, but it is simpler to combine this step with the clique finding.
The net result on the complexity is that wherever we have $\binom{n}{k}$ in the original costs, it is replaced with $n^k/k!$.

\section{Accuracy of Dicke preparation success probability}
\label{sec:probs}
{ 
Here we analyse the accuracy of the approximation of the sums by integrals in Eq.~\eqref{eq:intaprx}.
The sums correspond to the integrals discretised in steps of $1/cn$, so the approximation can be expected to converge for large $c$.
Even for moderate $c$ the approximation of the sums by integrals is highly accurate, because the integrands are zero at the bounds of the integrals, and the first non-zero derivatives at the bounds are of order $\min(k,n-k)$.
That means the error term in the Euler–Maclaurin formula will correspond to the derivative of order $\min(k,n-k)$.

In particular, since the summands have nonzero derivatives only up to order $n$, the first sum can be given by
\begin{align}
    \sum_{\ell=0}^{cn} \left(\frac {\ell}{cn}\right)^k \left(1-\frac {\ell}{cn}\right)^{n-k} &= \int_{0}^{cn} d\ell \left(\frac {\ell}{cn}\right)^k \left(1-\frac {\ell}{cn}\right)^{n-k} + \frac{f(0)+f(cn)}{2} \nn
    & \quad + \sum_{j=1}^{\lceil\tfrac n2 \rceil} \frac{B_{2j}}{(2j)!} \left[ f^{2j-1}(cn) - f^{2j-1}(0) \right] \,
\end{align}
where $f(\ell)$ is the summand, $B_{2j}$ are Bernoulli numbers, and the upper bound on the sum over $j$ is chosen so that only derivatives up to order $n$ are included.
Now the summand is zero at the bounds, $f^{2j-1}(0)$ is nonzero only for $2j-1\ge k$, and  $f^{2j-1}(cn)$ is nonzero only for $2j-1\ge n-k$.
When $2j-1\ge k$, the only nonzero term in the derivative at $\ell=0$ is for $k$ derivatives of $(\ell/cn)^k$ and $2j-1-k$ derivatives of $(1-\ell/cn)^{n-k}$.
This is because we need exactly $k$ derivatives of $(\ell/cn)^k$ for it to be nonzero at $\ell=0$.
The general Leibniz rule therefore gives
\begin{align}
    \left. \frac{d^{2j-1}}{d\ell^{2j-1}} \left(\frac {\ell}{cn}\right)^k\left(1-\frac {\ell}{cn}\right)^{n-k}\right|_{\ell=0}
    &= \left. \binom{2j-1}{k}
    \left\{\frac{d^{k}}{d\ell^{k}}\left(\frac {\ell}{cn}\right)^k \right\}
    \left\{\frac{d^{2j-1-k}}{d\ell^{2j-1-k}}\left(1-\frac {\ell}{cn}\right)^{n-k}\right\}\right|_{\ell=0} \nn
    &= (-1)^{2j-1-k} \frac{(2j-1)!}{(cn)^{2j-1}k!(2j-1-k)!} k! \frac{(n-k)!}{[n-k-(2j-1-k)]!} \nn
    &= -(-1)^k \frac{(2j-1)!(n-k)!}{(cn)^{2j-1}(2j-1-k)!(n-2j+1)!} \, .
\end{align}
Similarly, the only nonzero term in the derivative at $\ell=cn$ is for $n-k$ derivatives of $(1-\ell/cn)^{n-k}$, and $2j-1-n+k$ derivatives of $(\ell/cn)^k$.
In that case the general Leibniz rule gives
\begin{align}
    \left. \frac{d^{2j-1}}{d\ell^{2j-1}} \left(\frac {\ell}{cn}\right)^k\left(1-\frac {\ell}{cn}\right)^{n-k}\right|_{\ell=cn}
    &= \left. \binom{2j-1}{n-k}
    \left\{\frac{d^{2j-1-n+k}}{d\ell^{2j-1-n+k}}\left(\frac {\ell}{cn}\right)^k \right\}
    \left\{\frac{d^{n-k}}{d\ell^{n-k}}\left(1-\frac {\ell}{cn}\right)^{n-k}\right\}\right|_{\ell=cn} \nn
    &= (-1)^{n-k} \frac{(2j-1)!}{(cn)^{2j-1}(n-k)!(2j-1-n+k)!} \frac{k!}{k-(2j-1-n+k)]!} (n-k)! \nn
    &= (-1)^{n-k} \frac{(2j-1)!k!}{(cn)^{2j-1}(2j-1-n+k)!(n-2j+1)!} \, .
\end{align}
These values for the derivatives then give us
\begin{align}
    &\sum_{\ell=0}^{cn} \left(\frac {\ell}{cn}\right)^k \left(1-\frac {\ell}{cn}\right)^{n-k} \nn
    &= \int_{0}^{cn} d\ell \left(\frac {\ell}{cn}\right)^k \left(1-\frac {\ell}{cn}\right)^{n-k}  + \sum_{j=\lceil\tfrac {k+1}2 \rceil}^{\lceil\tfrac n2 \rceil} 
    \frac{B_{2j}}{(2j)!} (-1)^k \frac{(2j-1)!(n-k)!}{(cn)^{2j-1}(2j-1-k)!(n-2j+1)!}\nn
    & \quad 
    +\sum_{j=\lceil\tfrac {n-k+1}2 \rceil}^{\lceil\tfrac n2 \rceil} 
    \frac{B_{2j}}{(2j)!} (-1)^{n-k} \frac{(2j-1)!k!}{(cn)^{2j-1}(2j-1-n+k)!(n-2j+1)!}
    \nn
&= \int_{0}^{cn} d\ell \left(\frac {\ell}{cn}\right)^k \left(1-\frac {\ell}{cn}\right)^{n-k}  + \sum_{j=\lceil\tfrac {k+1}2 \rceil}^{\lceil\tfrac n2 \rceil}
\frac{B_{2j}}{(2j)}(-1)^k \frac{(n-k)!}{(cn)^{2j-1}(2j-1-k)!(n-2j+1)!}  \nn
    & \quad +\sum_{j=\lceil\tfrac {n-k+1}2 \rceil}^{\lceil\tfrac n2 \rceil}
\frac{B_{2j}}{(2j)} (-1)^{n-k} \frac{k!}{(cn)^{2j-1}(2j-1-n+k)!(n-2j+1)!}
\end{align}
Next, for the evaluation of the second integral, we have exactly the same reasoning, except the values of the derivatives at the bounds are slightly different.
For the value at $\ell=1$ no (instead of $\ell=0$), we now have a factor of
\begin{equation}
    \left. \left( 1- \frac{\ell}{cn} \right)^{n-k-(2j-1-k)}\right|_{\ell=1} = \left( 1- \frac{1}{cn} \right)^{n-2j+1} .
\end{equation}
Similarly, for the values of the derivatives at $\ell=cn$, we obtain a factor of
\begin{equation}
    \left. \left( \frac{\ell}{cn} - \frac{1}{cn} \right)^{k-(2j-1-n+k)}\right|_{\ell=cn} = \left( 1- \frac{1}{cn} \right)^{n-2j+1} ,
\end{equation}
which is the same.

Using these factors and combining the expressions for the two sums then gives
\begin{align}
& \binom{n}{k} \sum_{\ell=0}^{cn} \left(\frac {\ell}{cn}\right)^k \left(1-\frac {\ell}{cn}\right)^{n-k}
- \binom{n}{k} \sum_{\ell=1}^{cn} \left(\frac {\ell}{cn}-\frac {1}{cn}\right)^k \left(1-\frac {\ell}{cn}\right)^{n-k} 
\nn
&= \frac{cn}{n+1} \left[ 1- \left(1-\frac 1{cn}\right)^{n+1} \right] 
     + \sum_{j=\lceil\tfrac {k+1}2 \rceil}^{\lceil\tfrac n2 \rceil}
    D(n,k,c,j) +\sum_{j=\lceil\tfrac {n-k+1}2 \rceil}^{\lceil\tfrac n2 \rceil}
    D(n,n-k,c,j) \, ,
\end{align}
where
\begin{equation}
    D(n,k,c,j) := (-1)^k\frac{n!}{k!} \frac{B_{2j}}{2j(cn)^{2j-1}(2 j - 1 - k)!(n - 2 j + 1)!}\left[ 1- \left(1-\frac 1{cn}\right)^{n-2j+1} \right] \, .
\end{equation}
We will show that $|D(n,k,c,j)|$ decreases with $j$, so the dominant terms come from the smallest $j$.
Moreover, we will show that the largest values come from taking $j=1$ with $k=1$ or $n-k=1$, in which case it is of order $1/c^2$.
That means the largest value the correction can take is of higher order than the failure probability.

To show that $|D(n,k,c,j)|$ decreases with $j$, first note that
\begin{equation}
    B_{2j} = \frac{(-1)^{j+1}2(2j)!}{(2\pi)^{2j}} \zeta(2j) \, ,
\end{equation}
so the monotonic decreasing property of the Riemann zeta function $\zeta(2j)$ gives
\begin{equation}
    \left|\frac{B_{2j+2}}{B_{2j}}\right| < \frac{2(j+1)(2j)}{(2\pi)^{2}}  \, .
\end{equation}
We can therefore upper bound the ratio of $|D(n,k,c,j+1)|$ to $|D(n,k,c,j)|$ as
\begin{align}
    \frac{|D(n,k,c,j+1)|}{|D(n,k,c,j)|} < \frac{(2j)(2j+1)(n - 2 j)(n - 2 j + 1)}{(2\pi cn)^2(2 j - k)(2 j + 1 - k)} < \frac 1{(2\pi c)^2} \, .
\end{align}
Here we have used the fact that the factor in square brackets for $D(n,k,c,j)$ is monotonically decreasing in $j$.
Note that in the sums we use both $D(n,k,c,j)$ and $D(n,n-k,c,j)$, and the same reasoning holds for both.
This shows that the summands $D(n,k,c,j)$ and $D(n,n-k,c,j)$ rapidly decrease with $j$, so the value for the smallest $j$ is dominant.

Now let us consider the size of $D(n,k,c,j)$ for the smallest $j$ in the sum.
For $k=1$, the sum starts from $j=1$, which gives
\begin{align}\label{eq:largest}
    &\frac{n!}{k!}\frac{B_{2j}}{2j(cn)^{2j-1}(2 j - 1 - k)!(n - 2 j + 1)!}\left[ 1- \left(1-\frac 1{cn}\right)^{n-2j+1} \right]  \nn
    &= \frac{1}{12c}\left[ 1- \left(1-\frac 1{cn}\right)^{n-1} \right] \nn
    &\le \frac{n-1}{12nc^2} \, ,
\end{align}
where we have used $B_2=1/6$.
Then for $k=2$, the sum starts from $j=2$, which gives
\begin{align}
    |D(n,2,c,2)|
    &= \frac{(n-1)(n-2)}{240n^2 c^3}\left[ 1- \left(1-\frac 1{cn}\right)^{n-3} \right] \nn
    &\le \frac{(n-1)(n-2)(n-3)}{240n^3 c^4} \, ,
\end{align}
where we have used $B_4=-1/30$.

Thus we find that the first term in the sum is $\mathcal{O}(c^{-(k+1)})$ for $k=1$ or $2$.
Next we will consider the value as we increase $k$ in steps of 2.
Let us put $j_k := \lceil\tfrac {k+1}2 \rceil$ for the first $j$ in the sum.
Regardless of whether $k$ is even or odd we find $(2j_k-1-k)!=1$.
For the case of odd $k$ we have $2j_k=k+1$, so 
\begin{align}
    \frac{|D(n,k+2,c,j_{k+2})|}{|D(n,k,c,j_k)|} < \frac{(2j_k)(2j_k+1)(n - 2 j)(n - 2 j + 1)}{(2\pi cn)^2(k+1)(k+2)} < \frac {1}{(2\pi c)^2} \, .
\end{align}
This shows that the sum is decreasing when we consider steps of 2 in $k$.
For $k$ even we have $2j_k=k+2$, so
\begin{align}
    \frac{|D(n,k+2,c,j_{k+2})|}{|D(n,k,c,j_k)|} < \frac {k+3}{(2\pi c)^2(k+1)} \le \frac {5}{3(2\pi c)^2} \, ,
\end{align}
which shows the sum is still decreasing.

In either case, increasing the value of $k$ by 2 corresponds to reducing the value by at least $c^2$, so we find that this first term in the sum is $\mathcal{O}(c^{-(k+1)})$ in general.
Moreover, because we have shown that successive terms in the sum are smaller by factors of $1/(2\pi c)^2$, the entire sum is $\mathcal{O}(c^{-(k+2)})$.
We have two sums, one with $k$ and one with $n-k$, so the total of the two sums over $j$ can be given as
\begin{equation}
    \mathcal{O} \left( \frac 1{c^{\min(k,n-k)+1}}\right) \, .
\end{equation}

These results can be used to show a lower bound of $1-1/2c$ for the success probability.
We will take $n\ge 3$, since in the trivial case $n=2$ we find the probability is exactly $1-1/2c$.
First, we use
\begin{align}\label{eq:combnd}
0&\le (n - 2)n - (n - 2)(n - 3)/3 - n^2/3 \nn
n^2/3 &\le (n - 2)n - (n - 2)(n - 3)/3 \nn
\frac 1{3}    &\le   \frac{(n-2)}{n}-\frac{(n-2)(n-3)}{3n^2} \nn
\frac 1{2\pi^2}    &<   \frac{(n-2)}{n}-\frac{(n-2)(n-3)}{3n^2} \nn
\frac 1{2\pi^2}    &<   \frac{c(n-2)}{n}-\frac{(n-2)(n-3)}{3n^2} \nn
\frac{n-1}{nc} \frac 1{(2\pi c)^2}    &<   \frac{(n-1)(n-2)}{2(nc)^2}-\frac{(n-1)(n-2)(n-3)}{6(nc)^3} \nn
\frac{n-1}{nc} \left( 1-\frac 1{(2\pi c)^2} \right)     &>  \frac{n-1}{nc} - \frac{(n-1)(n-2)}{2(nc)^2}+\frac{(n-1)(n-2)(n-3)}{6(nc)^3} \nn
\frac{n-1}{nc} \left( 1-\frac 1{(2\pi c)^2} \right)     &>  \left[ 1- \left(1-\frac 1{cn}\right)^{n-1} \right] \nn
    \frac{n-1}{6nc^2 }      &>  \frac{1}{6c}\left[ 1- \left(1-\frac 1{cn}\right)^{n-1} \right] \frac{1}{1-\frac 1{(2\pi c)^2}} \, .
\end{align}
In the fifth line we have used $c\ge 1$.
Next, we lower bound the probability of success by upper bounding the sums over $j$ by the cases with $k=1$ or $n-k=1$, and using a factor of $(1-1/(2\pi c)^2)^{-1}$ to account for the sum where each successive term is less than the previsou by at least a factor of $1/(2\pi c)^2$.
We find that
\begin{align}
& \binom{n}{k} \sum_{\ell=0}^{cn} \left(\frac {\ell}{cn}\right)^k \left(1-\frac {\ell}{cn}\right)^{n-k}
- \binom{n}{k} \sum_{\ell=1}^{cn} \left(\frac {\ell}{cn}-\frac {1}{cn}\right)^k \left(1-\frac {\ell}{cn}\right)^{n-k} 
\nn
&\ge \frac{cn}{n+1} \left[ 1- \left(1-\frac 1{cn}\right)^{n+1} \right] 
     -2 \frac{1}{12c}\left[ 1- \left(1-\frac 1{cn}\right)^{n-1} \right] \frac{1}{1-\frac 1{(2\pi c)^2}} \nn
    &\ge 1 -  \frac {1}{2c} + \frac{n-1}{6nc^2 }      -  \frac{1}{6c}\left[ 1- \left(1-\frac 1{cn}\right)^{n-1} \right] \frac{1}{1-\frac 1{(2\pi c)^2}} \nn
    &> 1 -  \frac {1}{2c} \, .
\end{align}
In the last line we have used the bound we derived in Eq.~\eqref{eq:combnd}.
Including the case $n-2$ as well, this shows that the probability of success is lower bounded by $1-1/2c$ in general.

\begin{figure}[bh]
\centering
\label{fig:interr}
  \includegraphics[width=0.6\linewidth]{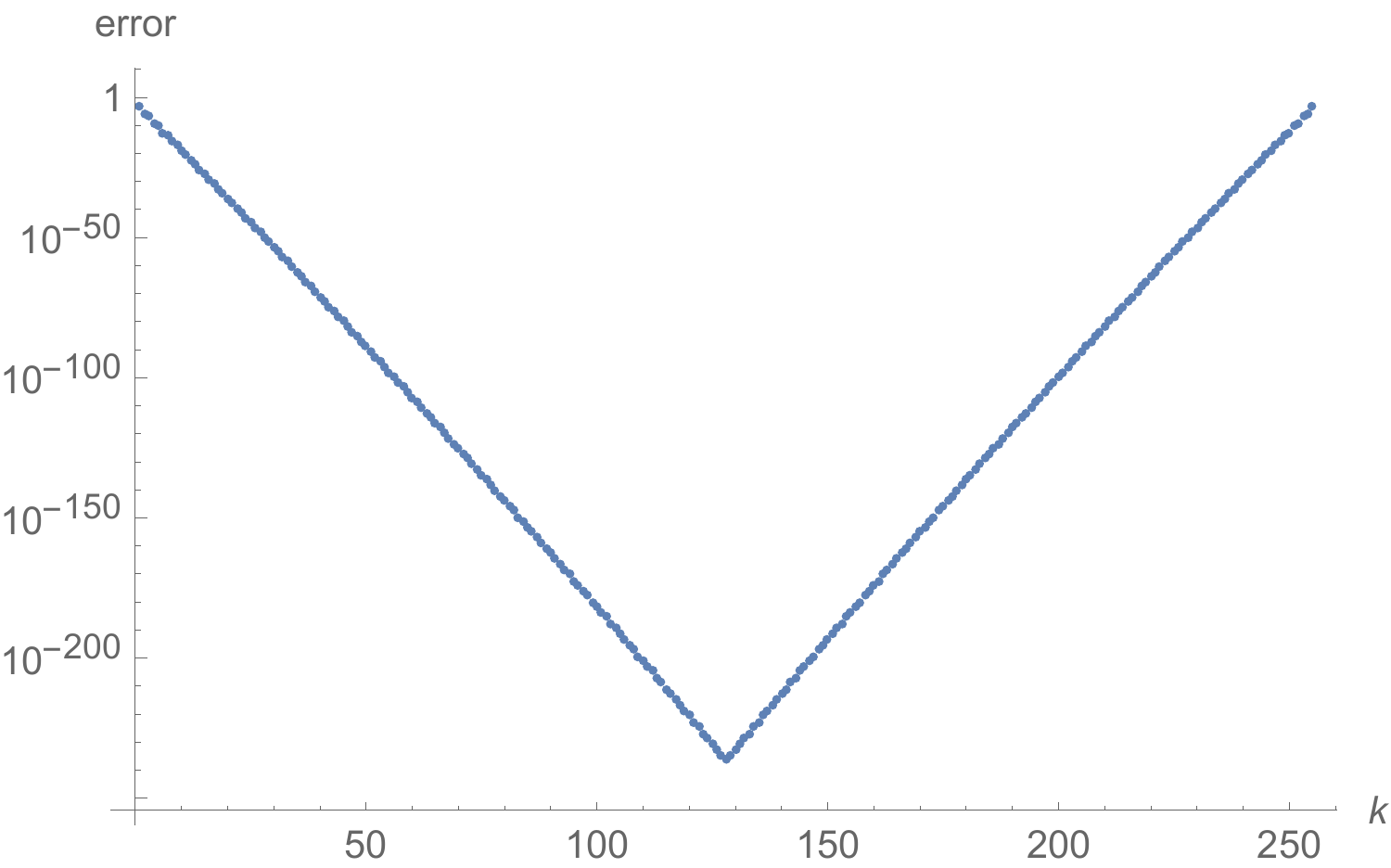}
  \caption{The error in the integral approximation in Eq.~\eqref{eq:intaprx} as a function of $k$ for $n=256$, $c=8$.}
\end{figure}

For the example where $n=256$ and $c=8$, the integral approximation gives estimated success probability $0.94001551223575$.
We show the error in the integral approximation in Fig.~\ref{fig:interr} for this example.
As can be expected from our analysis of the Euler-Maclaurin formula, it is found that the integral approximation is more accurate as $\min(k,n-k)$ is increased.
For $k$ near $n/2$ the approximation is accurate to well over $200$ decimal places.
For $k=16$, the integral approximation is accurate to about $30$ decimal places.
For $k=2$ the error in the approximation is still less than $10^{-6}$, though it increases to $0.0012$ for $k=1$ or $n-1$.

We can also more simply prove that $1-1/2c$ lower bounds the probability of success when averaging over $k$.
Summing Eq.~\eqref{eq:exactprob} over $k=1$ to $n-1$ gives
\begin{align}
   & \frac 1{(cn)^n} \sum_{k=1}^{n-1} \binom{n}{k} \sum_{\ell=1}^{cn} [\ell^k-(\ell-1)^k] (cn-\ell)^{n-k} \nn
   &= \frac 1{(cn)^n} \sum_{k=0}^{n} \binom{n}{k} \sum_{\ell=1}^{cn} [\ell^k-(\ell-1)^k] (cn-\ell)^{n-k} - \frac 1{(cn)^n} \sum_{\ell=1}^{cn} [\ell^n-(\ell-1)^n] \nn
   &= \frac 1{(cn)^n} \sum_{\ell=1}^{cn} [(cn)^n-(cn-1)^n]  - \frac 1{(cn)^n} (cn)^n \nn
   &= cn-1 - \frac{(cn-1)^n}{(cn)^{n-1}} \nn
   &\ge cn-1 - \frac{(cn)^n-n(cn)^{n-1}+n(n-1)(cn)^{n-2}}{(cn)^{n-1}} \nn
   &= (n-1) \left( 1 - \frac 1{2c} \right) \, .
\end{align}
Dividing by $n-1$ for the number of values of $k$ then gives a lower bound of $1-1/2c$ for the average success probability.}

\section{Proof of complexity of amplitude estimation}
\label{append:proofest}

For the complexity of amplitude estimation, the standard approach is to use phase estimation on the Grover iterate of amplitude amplification.
If there is an initial amplitude of $a$, then the phase of each step of amplitude amplification is $2\arcsin{a}$.
The original proposal was to use control registers in the phase estimation in an equal superposition, but of course for phase estimation that is a poor choice.
Here we would like a small probability of error beyond a given confidence interval, and for that case it is better to use a Kaiser window.

When applying phase measurement, we would start with a control state of the form (omitting normalisation)
\begin{equation}
    \sum_{m=-N}^N \frac 1{2N} \frac{I_0\left( \pi\alpha\sqrt{1-(m/N)^2}\right)}{I_0(\pi\alpha)} |m\rangle ,
\end{equation}
{where $I_0$ is a zeroth-order modified Bessel function of the first kind.}
If we call the operator combining $U$ and the reflection on the flag qubit $W$, then
we would then control between applications of $W$ and $W^\dagger$ with eigenvalue $e^{i\theta}$ to give
\begin{equation}
    \sum_{m=-N}^N \frac 1{2N} \frac{I_0\left( \pi\alpha\sqrt{1-(m/N)^2}\right)}{I_0(\pi\alpha)} e^{im\theta} |m\rangle.
\end{equation}
The inverse quantum Fourier transform then corresponds to an inner product with a phase state
\begin{equation}
    \frac 1{\sqrt{2N+1}}\sum_{m=-N}^N e^{im\hat \theta} |m\rangle .
\end{equation}
The inner product then gives the Fourier transform, so is proportional to
\begin{equation}
    \frac{\sin\left( \sqrt{N^2(\hat\theta-\theta)^2-(\pi\alpha)^2} \right)}{I_0(\pi\alpha) \sqrt{N^2(\hat\theta-\theta)^2-(\pi\alpha)^2}}.
\end{equation}
This needs to be squared to give the probability distribution for the error in the phase measurement.

The distribution has its first zero for $\theta=(\pi/N)\sqrt{1+\alpha^2}$, so to estimate the probability in the wings of the distribution we should integrate past that point.
We also have the difficulty that we are not given the exact normalisation of the probability distribution.
To approximate the normalisation, we can approximate the centre of the distribution by a Gaussian.
The approximation can be found by taking the Taylor series of the log of the distribution about zero, and gives
\begin{equation}
    \frac{\sinh^2(\pi\alpha)}{\pi^2\alpha^2 I_0^2(\pi\alpha)} e^{-N^2(\pi\alpha\coth(\pi\alpha)-1)\Delta\theta^2/(\pi^2\alpha^2)},
\end{equation}
where we have replaced $\hat\theta-\theta=\Delta\theta$.
Taking the integral over $\Delta\theta$ then gives
\begin{equation}
    \frac{\sinh^2(\pi\alpha)}{N I_0^2(\pi\alpha)\sqrt\pi\alpha\sqrt{\pi\alpha\coth(\pi\alpha)-1}}.
\end{equation}
It is found that this expression is asymptotically
\begin{equation}
    \frac {\pi}{2N\sqrt{\alpha}} + \mathcal{O}(\alpha^{-3/2}).
\end{equation}
This can be found using the asymptotic properties of Bessel functions,
\begin{equation}
    \frac 1{I_0^2(\pi\alpha)} \approx \frac{2\pi^2\alpha}{e^{2\pi\alpha}}
\end{equation}
and $\sinh^2(\pi\alpha)\approx e^{2\pi\alpha}/4$ and $\coth(\pi\alpha)\approx 1$, so
\begin{equation}
    \frac{\sinh^2(\pi\alpha)}{N I_0^2(\pi\alpha)\sqrt\pi\alpha\sqrt{\pi\alpha\coth(\pi\alpha)-1}}
    \approx \frac{2\pi^2\alpha}{e^{2\pi\alpha}} \frac{e^{2\pi\alpha}}4 \frac 1{N\sqrt\pi \alpha \sqrt{\pi\alpha-1}}.
\end{equation}
That gives the asymptotic expression claimed.

Now, for the integral over the tails we can upper bound the probability by that where we replace the sin with 1, so we have an upper bound
\begin{equation}
    2\int_{(\pi/N)\sqrt{1+\alpha^2}}^\infty \frac{1}{I_0^2(\pi\alpha) [N^2\delta\theta^2-(\pi\alpha)^2]} d\Delta\theta =  \frac 1{I_0^2(\pi\alpha)} \frac{2\,{\rm arcsinh}(\alpha)}{\pi N \alpha}.
\end{equation}
Now using ${\rm arcsinh}(\alpha)\approx \ln(2\alpha)$, we have the asymptotic expression
\begin{equation}
    \frac{2\pi^2\alpha}{e^{2\pi\alpha}} \frac{2\ln(2\alpha)}{\pi N \alpha} = \frac{4\pi \ln(2\alpha)}{N e^{2\pi\alpha}}.
\end{equation}
Dividing by the asymptotic expression for the normalisation then gives
\begin{equation}
    \frac{4\pi \ln(2\alpha)}{N e^{2\pi\alpha}}\frac {2N\sqrt{\alpha}}{\pi} = 8\ln(2\alpha) \sqrt{\alpha} \, e^{-2\pi\alpha}.
\end{equation}
This tells us that, if we want probability of error outside the range given by $\delta$, then we should take
\begin{equation}
    \ln (1/\delta) \approx 2\pi\alpha - \ln[8\ln(2\alpha) \sqrt{\alpha}].
\end{equation}
Solving for $\alpha$ then gives
\begin{equation}
    \alpha = (1/2\pi)\ln (1/\delta) + \mathcal{O}(\ln \ln (1/\delta)).
\end{equation}
The size of the confidence interval is $(\pi/N)\sqrt{1+\alpha^2}$, so if that needs to be $\epsilon$, we should take
\begin{equation}
    N = \frac{\pi}{\epsilon}\sqrt{1+\alpha^2} = \frac{1}{2\epsilon} \ln(1/\delta) + \mathcal{O}(\epsilon^{-1} \ln \ln (1/\delta)).
\end{equation}
The higher-order $\ln \ln$ term for $\alpha$ is larger than the correction term for approximating $\sqrt{1+\alpha^2}$ by $\alpha$.
The number of calls to $U$ or $U^\dagger$ is $N$, giving the complexity stated.

\section{Qubitization with projection}
\label{append:proj}

Here we derive the expression for qubitization with a more general projection as given in Eq.~\eqref{eq:Wperp}.
When the block encoding is defined more generally using 
\begin{equation}
    \left( \ket{0}\bra{0} \otimes P \right) V \left( \ket{0}\bra{0} \otimes P \right) = \ket{0}\bra{0} \otimes H/\lambda ,
\end{equation}
then for $\ket{k}$ an eigenstate of $H$ with energy $E_k$ (and $P\ket{k}=\ket{k}$), we have
\begin{align}
    \left( \ket{0}\bra{0} \otimes P \right) V \left( \ket{0}\bra{0} \otimes P \right)\ket{0} \ket{k} &= \frac{E_k}{\lambda} \ket{0} \ket{k} \nn
    \left( \ket{0}\bra{0} \otimes P \right) V \ket{0} \ket{k} &= \frac{E_k}{\lambda} \ket{0} \ket{k} \nn
    V \ket{0} \ket{k} &= \frac{E_k}{\lambda} \ket{0} \ket{k}  + i \sqrt{1-\left| \frac{E_k}{\lambda}\right|^2}\ket{0k^\perp},
\end{align}
where $\ket{0k^\perp}$ is defined as a state such that $\left( \ket{0}\bra{0} \otimes P \right)\ket{0k^\perp} = 0$.
That is, $V$ gives an application of $H$ to the target system, with failure being flagged by states orthogonal to $\ket{0}$ on the ancilla or perpendicular to the projector $P$ on the system.
The phase factor on the orthogonal part can be chosen arbitrarily, and is chosen here as $i$ to simplify later expressions.
Then, one can define the qubiterate as $W := R V$ in the usual way except
with the reflection being
\begin{equation}
    R := i\left( 2\ket{0}\bra{0} \otimes P - I \right).
\end{equation}
This is similar to that in \cite{BerryNPJ18}, except we have included the projection $P$ in the reflection operation.
Essentially the entire chain of reasoning as in \cite{BerryNPJ18} can be used, except replacing $\ket{0}\bra{0} \otimes I$ with $\ket{0}\bra{0} \otimes P$.
Then we obtain
\begin{equation}
    W \ket{0} \ket{k} = i \frac{E_k}{\lambda} \ket{0} \ket{k}  + \sqrt{1-\left| \frac{E_k}{\lambda}\right|^2}\ket{0k^\perp}.
\end{equation}
Then, to show the correct expression for $W \ket{\chi k^\perp}$ we use
\begin{align}
W^\dagger \ket{\chi} \ket{k} &= -VR \ket{\chi} \ket{k}\nn
&=- i V \ket{\chi} \ket{k} \nn
&= i RW \ket{\chi} \ket{k}\nn
&= i R\left(  i \frac {E_k}\lambda\ket{\chi}\ket{k} + \sqrt{1-\left|\frac{E_k}\lambda\right|^2} \ket{\chi k^\perp}\right)\nn
&=- i\frac {E_k}\lambda\ket{\chi}\ket{k} +  \sqrt{1-\left|\frac{E_k}\lambda\right|^2} \ket{\chi k^\perp} \, .
\end{align}
Similarly, for $\ket{\chi k^\perp}$, we have
\begin{align}
W^\dagger \ket{\chi k^\perp} &= -VR \ket{\chi k^\perp} \nn
&= i V \ket{\chi k^\perp} \, .
\end{align}
Now applying $W^\dagger$ to the expression for $W \ket{\chi}\ket{k}$ gives
\begin{align}
\ket{\chi}\ket{k} &= i \frac {E_k}\lambda W^\dagger \ket{\chi}\ket{k} + \sqrt{1-\left|\frac{E_k}\lambda\right|^2} W^\dagger\ket{\chi k^\perp}\nn
&=  i \frac {E_k}\lambda \left(- i \frac {E_k}\lambda\ket{\chi}\ket{k} + \sqrt{1-\left|\frac{E_k}\lambda\right|^2} \ket{\chi k^\perp}\right) +i \sqrt{1-\left|\frac{E_k}\lambda\right|^2} V \ket{\chi k^\perp}\nn
&= \left| \frac {E_k}\lambda \right|^2 \ket{\chi}\ket{k}
+i \frac {E_k}\lambda
\sqrt{1-\left|\frac{E_k}\lambda\right|^2} \ket{\chi k^\perp}
+i \sqrt{1-\left|\frac{E_k}\lambda\right|^2} V \ket{\chi k^\perp}\nn
\left( 1- \left| \frac {E_k}\lambda \right|^2 \right) \ket{\chi}\ket{k} &= i \frac {E_k}\lambda
\sqrt{1-\left|\frac{E_k}\lambda\right|^2} \ket{\chi k^\perp}
+i \sqrt{1-\left|\frac{E_k}\lambda\right|^2} V \ket{\chi k^\perp}\nn
-i \sqrt{ 1- \left| \frac {E_k}\lambda \right|^2 } \ket{\chi}\ket{k} &= \frac {E_k}\lambda
 \ket{\chi k^\perp}
+ V \ket{\chi k^\perp}\nn
V \ket{\chi k^\perp} &= -\frac {E_k}\lambda
 \ket{\chi k^\perp} - i\sqrt{ 1- \left| \frac {E_k}\lambda \right|^2 } \ket{\chi}\ket{k}\nn
RV \ket{\chi k^\perp} &= i\frac {E_k}\lambda
 \ket{\chi k^\perp} + \sqrt{ 1- \left| \frac {E_k}\lambda \right|^2 } \ket{\chi}\ket{k} \, .
\end{align}
Hence we obtain Eq.~\eqref{eq:Wperp} as required.
We have corrected the extra factors of $i$ included in \cite{BerryNPJ18}.
Note that the rest of the reasoning in \cite{BerryNPJ18} is correct.

\section{Betti number and spectral gap calculations}
\label{append:proof_betti_density}

The purpose of this section is to prove Propositions \ref{prop:Betti_number}, \ref{prop:spectral_gap}, \ref{prop:Betti_number_Rips} and \ref{prop:spectral_gap_Rips}.
In this section, we will work with \emph{reduced} homology. This is identical to regular homology, except that we have an extra 1-dimensional space $C_{-1}$ and an extra boundary map $\partial_0 : C_0 \rightarrow C_{-1}$ which maps every vertex (0-simplex) to the unique basis vector of $C_{-1}$. This has the effect that the reduced homology $H_0$ is equal to the number of connected components \emph{minus one}, rather than simply the number of connected components. The rest of the homology groups $H_k$ for $k>0$ are unchanged.

\begin{definition}
Given two simplicial complexes $X$ and $Y$, define their \emph{join} $X \ast Y$ to be the simplicial complex consisting of faces $\sigma \otimes \tau := \sigma \cup \tau$ for all $\sigma \in X$, $\tau \in Y$.
\end{definition}

We observe that $K(m,k) = K(m,k-1) \ast K(m,1)$.
Moreover, the homology of the join is given by the well-known Kunneth formula.

\begin{lemma} \emph{(Kunneth formula)}
\begin{align}
\widetilde{H}_k(X \ast Y) &= \bigoplus_{i+j = k-1} \widetilde{H}_i(X) \otimes \widetilde{H}_j(Y) \\
\implies \widetilde{\beta}_k(X \ast Y) &= \sum_{i+j = k-1} \widetilde{\beta}_i(X) \widetilde{\beta}_j(Y)
\end{align}
\end{lemma}

We would also like to relate the Laplacian of $X \ast Y$ to the Laplacians of $X$ and $Y$.

\begin{lemma} \label{Laplacian_lemma}
Let $\sigma \in X$ be an $i$-simplex and $\tau \in Y$ a $j$-simplex with $i+j=k-1$. Then
\begin{equation}
\Delta^{X \ast Y}_k (\sigma \otimes \tau) = (\Delta^{X}_i \sigma) \otimes \tau + \sigma \otimes (\Delta^{Y}_j \tau)
\end{equation}
\end{lemma}
\begin{proof}
Let's work in the \emph{graded algebra} $C_{-1} \oplus C_0 \oplus C_1 \oplus \dots$. We have
\begin{align*}
\Delta &= \partial^\dag \partial + \partial \partial^\dag \\
\partial(\sigma \otimes \tau) &= (\partial \sigma) \otimes \tau + (-1)^{|\sigma|} \sigma \otimes (\partial \tau) \\
\partial^\dag(\sigma \otimes \tau) &= (\partial^\dag \sigma) \otimes \tau + (-1)^{|\sigma|} \sigma \otimes (\partial^\dag \tau) \\
\implies \Delta(\sigma \otimes \tau) &= (\Delta\sigma) \otimes \tau + \sigma \otimes (\Delta\tau)
\end{align*}
\end{proof}

\begin{corollary} \label{Laplacian_corollary}
Let $\spec{\Delta}$ denote the set of eigenvalues of $\Delta$.
\begin{equation}
\spec{\Delta^{X \ast Y}_k} = \bigcup_{i+j = k-1} \spec{\Delta^X_i} + \spec{\Delta^Y_j}
\end{equation}
Here the plus notation for sets means $A + B = \{a+b : a\in A, \ b\in B\}$.
\end{corollary}
\begin{proof}
Use Lemma \ref{Laplacian_lemma} and let $\sigma \in C^X_i$ and $\tau \in C^Y_j$ be eigenchains of $\Delta^X_i$ and $\Delta^Y_j$ respectively.
\end{proof}

\begin{proposition}
\emph{(Restatement of Proposition \ref{prop:Betti_number}.)}

The $(k-1)^{\text{th}}$ Betti number of (the clique complex of) $K(m,k)$ is
\begin{equation}
\beta_{k-1} = (m-1)^k
\end{equation}
\end{proposition}
\begin{proof}
$K(m,k) = K(m,k-1) \ast K(m,1)$ and the Betti numbers of $K(m,1)$ are $(m-1,0,0,\dots)$. Thus by induction using the Kunneth formula, we have $\beta_{k-1} = (m-1)^k$.
\end{proof}

\begin{proposition}
\emph{(Restatement of Proposition \ref{prop:spectral_gap}.)}

The combinatorial Laplacian $\Delta_{k-1} = \partial_{k-1}^{\dagger} \partial_{k-1} + \partial_{k}\partial_{k}^{\dagger}$ of (the clique complex of) $K(m,k)$ has spectral gap
\begin{equation}
\lambda_{\min} = m
\end{equation}
\end{proposition}
\begin{proof}
Again $K(m,k) = K(m,k-1) \ast K(m,1)$. The spectrum of the $\Delta^{K(m,1)}_0$ is 0 with multiplicity $m-1$, and $m$ with multiplicity 1. Thus by induction using Corollary \ref{Laplacian_corollary}, the spectrum of $\Delta^{K(m,k)}_{k-1}$ is (ignoring multiplicities) $\{0,m,2m,\dots,km\}$. This gives $\lambda_{\min} = m$.
\end{proof}

\begin{proposition}
\emph{(Restatement of Proposition \ref{prop:Betti_number_Rips}.)}

The Rips complex described in \sec{Rips_complex_speedup} has $(2k-1)^{\text{th}}$ Betti number given by
\begin{equation}
\beta_{2k-1}(S) = (m-1)^k = \left(\frac{n}{2k} - 1\right)^k
\end{equation}
\end{proposition}
\begin{proof}
The Betti numbers of $\mathcal{R}_1(S_i)$ are $(m-1,0,\dots)$. Furthermore, $R_1(S) = R_1(S_0) \ast \dots \ast R_1(S_{k-1})$. Thus by the Kunneth formula $\beta_{2k-1}(S) = (m-1)^k$.
\end{proof}

\begin{proposition}
\emph{(Restatement of Proposition \ref{prop:spectral_gap_Rips}.)}

The Rips complex described in \sec{Rips_complex_speedup} has a combinatorial Laplacian $\Delta_k = \partial_k^{\dagger} \partial_k + \partial_{k+1}\partial_{k+1}^{\dagger}$ with a constant spectral gap $\lambda_{\min}$.
\end{proposition}
\begin{proof}
Consider first the Laplacian $\Delta^{S_i}_k$ of $\mathcal{R}_1(S_i)$. Say it has smallest eigenvalue $c$. Then by induction using Corollary \ref{Laplacian_corollary}, the smallest eigenvalue of $\Delta^{S}_k$ is also $c$.
\end{proof}

\subsection{Perturbations of $K(m,k)$ and smaller spectral gaps} 
\label{app:new_graph_examples}

In Section~\ref{subsec:classical}, we introduced a novel classical algorithm that, in conjunction with Apers et al.'s algorithm~\cite{apers2022simple}, significantly enhances classical runtimes in the domain where quantum algorithms achieve polynomial run-times. 
Moreover, these classical algorithms can potentially operate in polynomial time, contingent on the required precision for estimating the Betti number, provided the spectral gap of the combinatorial Laplacian is sufficiently large.
In this section, we study perturbations of the graph $K(m,k)$ (defined in the previous section) that lead to a smaller spectral gap. 
These perturbations render the classical algorithms inefficient, regardless of the precision needed for the Betti number estimation, though they do not impair the efficiency of the quantum algorithm.
Additionally, these perturbations also show that the speed-ups are robust, in the sense of a much larger class of graphs having a guaranteed speed-up

We will denote our perturbed family of graphs $K'(m,k)$. 
Recall that $K(m,k)$ consisted of $k$ clusters, each with $m$ vertices.
There were no edges within the clusters, and all edges between vertices in different clusters were included.
Let $G'$ be the graph consisting of $m$ vertices with a single edge. 
Thus $G'$ consists of a single edge, and $m-2$ additional completely disconnected vertices. 
Let $K'(m,k)$ consist of $k$ clusters, where each cluster is a copy of $G'$. 
Likewise, we include all edges between clusters.
In other words, we obtain an instance of $K'$ by adding a single edge to each of the clusters which were originally independent sets.
\begin{lemma} 
\label{lemma:Betti_number_prime}
The $(k-1)^{\text{th}}$ Betti number of the clique complex of $K'(m,k)$ is
\[
\beta_{k-1} = \left(m-2\right)^k .
\]
\end{lemma}

\begin{lemma} 
\label{lemma:eigenvals_prime}
The $(k-1)^{\text{th}}$ combinatorial Laplacian $\Delta_{k-1}$ of the clique complex of $K'(m,k)$ has smallest nonzero eigenvalue
\begin{equation*}
\lambda_{\min} = 2,
\end{equation*}
and largest eigenvalue
\begin{equation*}
\lambda_{\max} = 2mk = 2n.
\end{equation*}
\end{lemma}

\begin{proof}[Proof sketch]
First, note that the clique complex of $K(m, k)$ is the $k$-fold join of the clique complex of the graph $G'$ consisting of $m$ vertices with a single edge between two of the vertices.
Next, use the Kunneth formula and Corollary~\ref{Laplacian_corollary}.
Finally, note that the zeroth combinatorial Laplacian of $\mathrm{Cl}(G')$ is
\begin{equation*}
\Delta_0 =
\begin{pmatrix}
2 & 0 & 1 & 1 & \dots \\
0 & 2 & 1 & 1 & \dots \\
1 & 1 & 1 & 1 & \dots \\
1 & 1 & 1 & 1 & \dots \\
. & . & . & . & \dots
\end{pmatrix}
\end{equation*}
The spectrum of $\Delta_0$ is thus 0 with multiplicity $m-2$, 2 with multiplicity $m$, and $2m$ with multiplicity 1. 
\end{proof}

\begin{corollary}
\label{cor:gap_prime}
The $(k-1)^{\text{th}}$ combinatorial Laplacian $\Delta_{k-1}$ of the clique complex of $K'(n/k,k)$ has normalized spectral gap
\[
\gamma = \frac{\lambda_{\min}}{\lambda_{\max}} = \frac{1}{n}.
\]
\end{corollary}

In short, when $k =c\cdot\log(n)$ the graph $K'(m,k)$ has an inverse polynomial spectral gap, as opposed to an inverse logarithmic spectral gap as in the case for $K(m,k)$. 
This key distinction amplifies the disparities in runtime between the quantum and classical algorithms. 
Notably, the quantum algorithm (still) attains a superpolynomial speedup over its classical counterparts, even in scenarios where only a constant level of additive precision for the normalized Betti number is required.

Next, we demonstrate the potential for even broader generalization of these perturbations.
Consider the $m$-vertex graph $G''$ that is a disjoint union of $m/2$ edges (i.e., there are no isolated vertices), and let $K''(m,k)$ be the $k$-fold join of $G''$.

\begin{lemma}
\label{lemma:spec_primeprime}
The spectrum of the $(k-1)^{\text{th}}$ combinatorial Laplacian of the clique complex of $K''(m,k)$ is
\[
\spec \Delta_{k-1} = \Big\{\lambda_{i,j,\ell} \text{ }\big|\text{ } i + j + \ell = k\Big\},
\]
where
\[
\lambda_{i, j, \ell} = 0\cdot i + 2\cdot j + m\cdot l
\]
and they have multiplicities
\[
\mathrm{multiplicity}[i, j, \ell] = \binom{n}{i, j, \ell} \cdot \left(\frac{m}{2} - 1\right)^i \cdot \left(\frac{m}{2}\right)^j \cdot 1^\ell.
\]
\end{lemma}
\begin{proof}
The spectrum of the $0^{\text{th}}$ combinatorial Laplacian of $G''$ is
\begin{equation*}
\Delta_0(\mathrm{Cl}(G'')) =
\begin{pmatrix}
2 & 0 & 1 & 1 & \dots \\
0 & 2 & 1 & 1 & \dots \\
1 & 1 & 2 & 0 & \dots \\
1 & 1 & 0 & 2 & \dots \\
. & . & . & . & \dots
\end{pmatrix}
\end{equation*}
A quick computation reveals that $\Delta_0$ has a spectrum
\[
    \spec\left(\Delta_0(\mathrm{Cl}(G''))\right) = \left\{0 \text{ (with multiplicity }m/2 - 1), 2 \text{ (with multiplicity } m/2), m \text{ (with multiplicity }1)\right\}.
\]
Following this observation, the lemma follows from the application of Corollary~\ref{Laplacian_corollary}.
\end{proof}

\begin{corollary}
\label{cor:gap_primeprime}
The $(k-1)^{\text{th}}$ combinatorial Laplacian $\Delta_{k-1}$ of the clique complex of $K''(n/k,k)$ has normalized spectral gap
\[
\gamma = \frac{\lambda_{\min}}{\lambda_{\max}} = \frac{2}{n}.
\]
\end{corollary}

Again, the graph $K''(m,k)$ has an inverse polynomial spectral gap, as opposed to an inverse logarithmic spectral gap. 
In particular, the quantum algorithm (still) attains a superpolynomial speedup over its classical counterparts, even in scenarios where only a constant level of additive precision for the normalized Betti number is required.
In conclusion, adding a single edge or $m/2$ edges to each cluster within $K(m,k)$ does not change the superpolynomial quantum speedup that it exhibits.

\section{Dequantization using path integral Monte Carlo}
\label{append:PIMCdequant}

Previously, we argued that the cases where TDA can potentially have a super-polynomial advantage relative to classical approaches is in cases where the clique density is high.  In such cases sampling in our quantum algorithm is efficient and eigen-decomposition is inefficient.  However, we will see here that this is not necessarily the case and that there are cases where the clique density is high wherein randomized classical algorithms can achieve scaling that is polynomially equivalent to quantum algorithms. This will show that the conditions for a substantial improvement using quantum TDA are potentially even more subtle than previous work suggests.

As discussed in the main body, our dequantization looks at imaginary time simulations of  the Hermitian operator $\TBG= B_G^2 + (1-P)\gamma_{\min}$, where $B_G$ is the square of the constrained Dirac operator, i.e.\ the combinatorial Laplacian, and the projector selects all input states that are valid simplices. Given an upper bound on the sparsity $s$ of the combinatorial Laplacian, we can obtain the decomposition $\TBG = \sum_{p=1}^D c_{p} H_p$, where $H_p$ is one-sparse, {Hermitian and unitary}, and $D=\mathcal{O}(s^2)$, in polynomial time using distributed graph coloring algorithms. These algorithms also let us compute the position of the non-zero matrix element in row $x$ of $U_{\alpha}$ using a number of queries to $B_G^2$ that scales as $\mathcal{O}(\log^*(D))$~\cite{berry2007efficient}. We will provide explicit bounds on the sparsity of the combinatorial Laplacian in a subsequent section.

In order to set up the relevant path integrals, we must first employ a  Trotter-decomposition.  This allows us to represent the exponential in terms of exponentials of the one-sparse matrices, which can then be simulated through randomization.  This leads us to the conclusion that
\begin{equation}
    e^{-\TBG t} = (e^{-\TBG t/r})^r = \left(\prod_{p=1}^D e^{-c_{p} H_{p} t/2r}\prod_{p=D}^1 e^{-c_{p} H_{p} t/2r} + O\left(\frac{\big(\sum_{p} |c_{p}|\big)^3 t^3}{r^3}\right)\right)^r.
    \label{eq:trotterboundarydecomp}
\end{equation}
As $H_{p}$ is Hermitian, it has a complete set of eigenvectors. Since $H_{p}$ is also one-sparse and as such matrices can be written as the direct sum of irreducible one and two-dimensional matrices, we can parameterize the eigenvectors to respect the structure of the two dimensional space via
 \begin{equation}
H_p \ket{\lambda_{p,\nu}} = \lambda_{p,\nu} \ket{\lambda_{p,\nu}} \, .
 \end{equation}
 Note that each eigenvector $\ket{\lambda_{p,\nu}}$ is such that $\braket{p}{\lambda_{p,\nu}}$ is non-zero for only two different computational basis vectors. 
 
Note that the first term on the RHS of~\eqref{eq:trotterboundarydecomp} contains $2rD$ terms. If we introduce a vector of indices $p = \{1,\ldots,r,1,\ldots,r,\ldots,r\}$ with $2rD$ entries denoted by $p_i$, we can express the expectation of the exponential of the boundary operator as ($H$ stands for the Haar average)
\begin{equation}
\mathbb{E}^{H}_{\ket{\psi}}\left( \prod_{i=1}^{2rD}e^{-c_{p_i} H_{p_i} t/2r} \right) \coloneqq \mathbb{E}_{H} \bra{\psi} \left( \prod_{i=1}^{2rD}e^{-c_{p_i} H_{p_i} t/2r} \right) \ket{\psi}.
\end{equation}
 
Next we set up our path integrals by selecting sets of $2rD$ indices that correspond to the eigenstates that we transition to in the path integral.  We denote such a path via the vector $\Gamma$ where $\Gamma_j$ corresponds to the index of the $j^{\rm th}$ eigenstate in the path.
Using this notation we can insert resolutions of the identity of the form $\sum_{\Gamma_j} \ket{\lambda_{p_j,\Gamma_j}}\bra{\lambda_{p_j,\Gamma_j}}$ consisting of the eigenvectors $\ket{\lambda_{p_j,\Gamma_j}}$ of each $U_{p_j}$ in between each of the $2rD$ terms and defining $\rho = \ket{\psi}\bra{\psi}$ gives
\begin{align}
     &\mathbb{E}^{H}_{\ket{\psi}}\left( \prod_{i=1}^{2rD}e^{-c_{p_i} H_{p_i} t/2r} \right) = \mathbb{E}_{H} \ \text{Tr}\left(\rho  \prod_{i=1}^{2rD}e^{-c_{p_i} H_{p_i} t/2r} \right) \nonumber \\
     &= \mathbb{E}_{H} \text{Tr}\left(\rho \sum_{\Gamma_1,\ldots,\Gamma_{2rD}} \exp{\left(-\sum_{i=1}\lambda_{p_i,\Gamma_i}t/2r\right)} \ket{\lambda_{p_1,\Gamma_1}}\bra{\lambda_{p_1,\Gamma_1}}\cdots\ket{\lambda_{p_{2rD},\Gamma_{2rD}}}\bra{\lambda_{p_{2rD},\Gamma_{2rD}}} \right) \nonumber \\
     &= \mathbb{E}_{H} \text{Tr}\left(\rho \sum_{\Gamma_1,\ldots,\Gamma_{2rD}} \exp{\left(-\sum_{i=1}\lambda_{p_i,\Gamma_i}t/2r\right)} W(\Gamma) \ket{\lambda_{p_1,\Gamma_1}}\bra{\lambda_{p_{2rD},\Gamma_{2rD}}} \right) \nonumber \\
     &= \mathbb{E}_{H} \text{Tr}\left(\rho \ \mathbb{E}_\Gamma \frac{\exp{\left(-\sum_{i=1}\lambda_{p_i,\Gamma_i}t/2r\right)} W(\Gamma) \ket{\lambda_{p_1,k_1}}\bra{\lambda_{p_{2rD},k_{2rD}}}}{\text{Pr}(\Gamma)} \right) \nonumber \\
     &= \mathbb{E}_{H} \mathbb{E}_\Gamma \text{Tr}\left(\rho \frac{\exp{\left(-\sum_{i=1}\lambda_{p_i,\Gamma_i}t/2r\right)} W(\Gamma) \ket{\lambda_{p_1,\Gamma_1}}\bra{\lambda_{p_{2rD},\Gamma_{2rD}}}}{\text{Pr}(\Gamma)} \right)\nonumber\\
     &= \frac{1}{d_{k-1}} \mathbb{E}_\Gamma \text{Tr}\left(\frac{\exp{\left(-\sum_{i=1}\lambda_{p_i,\Gamma_i}t/2r\right)} W(\Gamma) \ket{\lambda_{p_1,\Gamma_1}}\bra{\lambda_{p_{2rD},\Gamma_{2rD}}}}{\text{Pr}(\Gamma)} \right)\nonumber\\
     &= \frac{1}{d_{k-1}} \mathbb{E}_\Gamma \left(\frac{\exp{\left(-\lambda_{p_1,\Gamma_1}t/r-\sum_{i=2}^{2rD-1}\lambda_{p_i,\Gamma_i}t/2r\right)} W(\Gamma) \delta_{\Gamma_1,\Gamma_{2rD}}}{\text{Pr}(\Gamma)} \right) \, , \label{eq:haaravgtrot}
\end{align}
  where we have defined for convenience the quantity 
\begin{equation}
     W(\Gamma) = \braket{\lambda_{p_1,\Gamma_1}}{\lambda_{p_2,\Gamma_2}}\cdots \braket{\lambda_{p_{2rD-1},\Gamma_{2rD-1}}}{\lambda_{p_{2rD},\Gamma_{2rD}}} \, , \label{eq:WGamma}
\end{equation}
 with $\Gamma= [\Gamma_1,\ldots,\Gamma_{2rD}]$. In the fourth line, we divided and multiplied by a probability $\text{Pr}(\Gamma)$ to express the sum as an average, which allows to use importance sampling to minimize the variance via a judicious choice of $\text{Pr}(\Gamma)$. In the last line, we used the fact that $p_{2rD} = p_1$ for the symmetric Trotter formula. 
 
If we wish to estimate this value by sampling, the primary driver of the complexity will be the estimation of the expectation value through the sample mean which corresponds to the optimal unbiased estimator of the population mean. The number of samples scales with the variance of the set that one averages over and the variance over $\Gamma$ of the above Haar expectation is then simply
 \begin{align}
     &\mathbb{V}_\Gamma\left(\frac{1}{d_{k-1}} \text{Tr}\left(\frac{\exp{\left(-\sum_{i=1}\lambda_{p_i,\Gamma_i}t/2r\right)} W(\Gamma) \ket{\lambda_{p_1,\Gamma_1}}\bra{\lambda_{p_{2rD},\Gamma_{2rD}}}}{\text{Pr}(\Gamma)} \right) \right)\nonumber\\
     &=\frac{1}{d_{k-1}^2} \sum_{\Gamma_1,\ldots,\Gamma_{2rD-1}}  \frac{\exp{\left(-2\lambda_{p_1,\Gamma_1}t/r-\sum_{i=2}^{2rD-1}\lambda_{p_i,\Gamma_i}t/r\right)} |W(\Gamma)|^2 \delta_{\Gamma_1,\Gamma_{2rD}}}{\text{Pr}(\Gamma)}\nonumber\\
     &\quad~- \left(\frac{1}{d_{k-1}} \mathbb{E}_\Gamma \left(\frac{\exp{\left(-\lambda_{p_1,\Gamma_1}t/r-\sum_{i=2}^{2rD-1}\lambda_{p_i,\Gamma_i}t/2r\right)} W(\Gamma) \delta_{\Gamma_1,\Gamma_{2rD}}}{\text{Pr}(\Gamma)} \right) \right)^2\nonumber\\
     &\le \frac{1}{d_{k-1}^2} \sum_{\Gamma_1,\ldots,\Gamma_{2rD-1}}  \frac{\exp{\left(-2\lambda_{p_1,\Gamma_1}t/r-\sum_{i=2}^{2rD-1}\lambda_{p_i,\Gamma_i}t/r\right)} |W(\Gamma)|^2 \delta_{\Gamma_1,\Gamma_{2rD}}}{\text{Pr}(\Gamma)} \, . \label{eq:varSamp}
 \end{align}
 
There are many probability distributions that we could choose to sample from to minimize the variance in~\eqref{eq:varSamp}.  The most straight forward distribution to choose, and the appropriate one to pick in the limit of short $t$, is a uniform distribution.  However, in practice the eigenvalues in the sum may have wildly varying sizes and so the importance of each of the different paths can swing substantially.  A more natural choice to make for the probability of drawing each path is
 \begin{equation}
     {\rm Pr}(\Gamma) =\frac{ \exp{\left(-2\lambda_{p_1,\Gamma_1}t/r-\sum_{i=2}^{2rD-1}\lambda_{p_i,\Gamma_i}t/r\right)}\delta_{\Gamma\in S_\Gamma}}{\sum_{\Gamma\in S_\Gamma}\exp{\left(-2\lambda_{p_1,\Gamma_1}t/r-\sum_{i=2}^{2rD-1}\lambda_{p_i,\Gamma_i}t/r\right)}},
 \end{equation}
 where $S_\Gamma$ is the set of valid paths with $2rD$ vertices such that each edge corresponds to a path of connected eigenvectors for the one-sparse matrices used in the decomposition of $\TBG$.
 
 The central challenge in employing this formula is to estimate the value of the sums over the values of $\Gamma_i$. {As the $H_j$ in used in the derivation are Hermitian and unitary, we have that}  
 \begin{equation}
    \lambda_{p_i,\Gamma_i} = \pm  c_{p_i}. \label{eq:PIUnitaryAssump}
 \end{equation}
Furthermore, the one-sparse decomposition is chosen in such a way that the matrix elements are all off diagonal with the exception of any diagonal matrix that appears in the decomposition.  This can be seen explicitly using the discussion of the Jordan-Wigner representation of the Dirac operator.  This means that each eigenvector couples to at most two eigenvectors. At most one term is diagonal in the standard Trotter decomposition of the Dirac operator~\cite{berry2007efficient}. A simple combinatorial argument therefore leads to the conclusion that the total number of valid paths is at most $d_{k-1} 2^{2r(D-1)-1}$.

Given this choice, the normalization constant (which is analogous to a partition function) can be expressed (assuming that the diagonal element is always $p_i=D$) as
 \begin{align}
     &\sum_{\Gamma\in G} \exp{\left(-2\lambda_{p_1,\Gamma_1}t/r-\sum_{i=2}^{2rD-1}\lambda_{p_i,\Gamma_i}t/r\right)}\nonumber\\
     &\quad= d_{k-1}2^{2rD-1-r}\cosh(2c_{p_1}t/r) \prod_{i=2}^{2rD-1} (\cosh(c_{p_i} t/r)\delta_{p_i \ne D} + \delta_{p_i = D}e^{-\lambda_{p_i,\Gamma_i} t/r}/2).
 \end{align}
 This can be computed using $\mathcal{O}({\rm poly}(rD))$ arithmetic operations and so does not ruin the efficiency of the algorithm.  Note that were the sum over the $W(\Gamma)$ terms considered instead, then the result would be computationally difficult to compute as these terms generate correlations that would prevent us from performing an independent sum for each of the factors.
 
The variance $\sigma^2$ over the values of $k$ chosen in the path integrals for the expression for the Haar average is then
 \begin{align}
     \sigma^2 &= \frac{1}{d_{k-1}^2}\left(\sum_{\Gamma\in S_\Gamma} { |{W(\Gamma)}|^2 }\right)\left({\sum_{\Gamma\in S_\Gamma}\exp{\left(-2\lambda_{p_1,\Gamma_1}t/r-\sum_{i=2}^{2rD-1}\lambda_{p_i,\Gamma_i}t/r\right)}}\right) \nonumber\\
     &- \frac{1}{d_{k-1}^2}\left|{\sum_{\Gamma\in S_\Gamma}}\exp{\left(-\lambda_{p_1,\Gamma_1}t/r-\sum_{i=2}^{2rD-1}\lambda_{p_i,\Gamma_i}t/2r\right)} {W(\Gamma)} \right|^2\nonumber\\
     &\le \frac{2^{2rD} e^{2Dt\max_{i} \lambda_{p_i,\Gamma_i}}}{d_{k-1}}\nonumber\\
     &\le \frac{2^{2rD} (1/\delta)^{2D\max_{i} |\lambda_{p_i,\Gamma_i}|/\gamma_{\min}}}{d_{k-1}}\label{eq:sigma2} \, .
 \end{align}
 This shows that the variance after making this substitution is precisely equal to the gap in the Cauchy-Schwarz inequalty in the latter sum.  This suggests that in certain cases where the Cauchy-Schwarz inequality is tight, the variance may be extremely small given that we have the ability to sample from the distribution ${\rm Pr}(\Gamma)$.

\subsection{Metropolis Hastings algorithm}

Outside of specific cases such as graphs, it is difficult in general to sample from the probability distribution ${\rm Pr}(\Gamma)$ to employ the above variance reduction strategy. It is therefore necessary to provide a general method to obtain these samples if we wish to understand how we could address the problem more generally. One way to address the issue of how to sample from the distribution ${\rm Pr}(\Gamma)$ is to use the Metropolis Hastings algorithm.  The idea behind the algorithm is to design a Markov chain whose stationary distribution equals our choice of $\text{Pr}(\Gamma)$. 

We first start with a connected, undirected graph on the set of all possible states $\Gamma_1,\Gamma_2,\ldots,\Gamma_{2rD-1}$ which represent the ``paths" involved in our Trotter decomposition. Each vertex of the graph then represents one possible collection of values for $\Gamma_1,\ldots,\Gamma_{2rD-1}$. 

At each vertex $a$, we therefore select a neighbor with probability $1/(2rD-1)$. Since the degree may be less than $2rD-1$ at a given vertex, the walk may remain at that vertex as there is a non-zero probability of no edge being selected. To account for such situations, we have the following rules: if a neighboring vertex $b$ is selected and the probability of transitioning
\begin{equation}
p_b \coloneqq \text{Pr}(\Gamma)_b = \exp{\left(-2\lambda_{p_1,\Gamma_1^{(b)}}t/r-\sum_{i=2}^{2rD-1}\lambda_{p_i,\Gamma_i^{(b)}}t/r\right)}
\end{equation}
is at least as great as the probability of remaining $p_a \coloneqq \text{Pr}(\Gamma)_a$, we transition to $b$. If $p_b < p_a$, then we transition to $b$ with probability $p_b/p_a$, where $b\in \mathcal{N}(a)$ with $\mathcal{N}(a)$ referring to the neighbors of $a$ and
\begin{equation}
    \frac{p_b}{p_a}=\frac{\exp{\left(-2\lambda_{p_1,\Gamma_1^{(b)}}t/r-\sum_{i=2}^{2rD-1}\lambda_{p_i,\Gamma_i^{(b)}}t/r\right)}}{\exp{\left(-2\lambda_{p_1,\Gamma_1^{(a)}}t/r-\sum_{i=2}^{2rD-1}\lambda_{p_i,\Gamma_i^{(a)}}t/r\right)}}. \label{eq:deltaPIBd}
\end{equation}
Otherwise we remain at $a$ with probability $1 - p_b/p_a$. Defining 
\begin{equation}
p_{ab} \coloneqq \frac{1}{R}\min \left (1,\frac{p_b}{p_a} \right)
\end{equation}
and
\begin{equation}
p_{aa} \coloneqq 1 - \sum_{b \neq a} p_{ab} \, ,
\end{equation}
we can easily verify $p_a p_{ab} = p_b p_{ba}$. By the Fundamental Theorem of Markov Chains~\cite{haggstrom2002finite}, it follows that the stationary probabilities are $p_a$ as needed. 

The cost of computing the ratio $p_b/p_a$ is $\mathcal{O}(rD)$ arithmetic operations, which coincides with the cost of performing an update. The number of such updates needed to reach $\delta$ error from the stationary distribution, where $\delta$ is the total variational distance (TVD) from the stationary distribution desired, is
\begin{equation}
    T^* \in O\left(\frac{\log(1/\delta)}{\gamma_M} \right),
\end{equation}
where $\gamma_M:=1-\lambda_2$ is the eigenvalue gap of for the transition matrix $p$. This implies that the number of arithmetic operations needed to  sample from a distribution that is $\delta-$close to the stationary distribution is in
\begin{equation}
    O\left(\frac{rD{\log}(1/\delta)}{\gamma_M} \right). \label{eq:NopsMarkov}
\end{equation}
If one samples from a distribution, $P'$, that is $\delta-$close to the intended distribution $P$ then the expectation value of any function $f$ is $|\sum_j(P(j)f(j)) -\sum_j(P'(j) f(j))|\le \delta \max |f(j)|$.  Similarly, the variance obeys $|\sum_j P'(j) f(j)^2 - (\sum_j P'(j) f(j))^2|\le \mathbb{V}(f) + \mathcal{O}(\delta \max|f(j)|^2)$. 
Thus if we want the error in the mean to be less than some error $\epsilon_M$, we require \begin{equation}
\delta = \epsilon_M/\max |f(j)|.
\end{equation}

Note that in our situation 
\begin{equation}
|f| = \left| \frac{\exp{\left(-\lambda_{p_1,\Gamma_1}t/r-\sum_{i=2}^{2rD-1}\lambda_{p_i,\Gamma_i}t/2r\right)} W(\Gamma)}{\text{Pr}(\Gamma)}\right| \, ,
\end{equation}
where $r$ is the number time-steps needed in the Trotterization procedure to attain a desired Trotter error $\epsilon_T$. 
Further, we have the following bound
\begin{align}
    & \sum_{\Gamma_1,\ldots,\Gamma_{2rD-1}}\left| W(\Gamma)\exp{\left(-2\lambda_{p_1,\Gamma_1}t/r-\sum_{i=2}^{2rD-1}\lambda_{p_i,\Gamma_i}t/r\right)}\right|\nonumber\\
    &\le \sqrt{\sum_{\Gamma_1,\ldots,\Gamma_{2rD-1}} |W(\Gamma)|^2} \sqrt{\sum_{\Gamma\in S_\Gamma} \exp{\left(-4\lambda_{p_1,\Gamma_1}t/r-2\sum_{i=2}^{2rD-1}\lambda_{p_i,\Gamma_i}t/r\right)}} \, .
    \label{eq:fmaxbd}
\end{align}

Note that when we sum over a specific $p_i,\Gamma_i$ in $|W(\Gamma)|^2$, we get $1$ since the eigenvectors are orthonormal. Performing all $2rD-2$ sums therefore gives $1$ from all the inner products and the last sum involving the $(2rD-1)$-th index gives a factor of $d_{k-1}$. Additionally, each $\lambda_{p_i,\Gamma_i}$ can either be positive or negative but is upper-bounded by $\|\TBG\|_{\max} \leq \|\TBG\|_{\infty} = \gamma_{\max}$, i.e.\ the largest eigenvalue of $\TBG$. We can then bound $f_{\max}$ as follows
\begin{align}
    |f|\le f_{\max} & \le d_{k-1} e^{2\gamma_{\max}t D} 2^{r(D-1/2)}. \label{eq:fmaxbound}
\end{align}

\subsection{Trotter error in path integration}

From Corollary 12 of~\cite{Childs2019c}, the multiplicative Trotter error $m$ for a $p$-th order Trotter formula is asymptotically bounded by
\begin{equation}
    \mathcal{O} \left (\alpha \left(\frac{t}{r}\right)^{p+1} \exp{\left(\frac{2t}{r} \Upsilon \sum_{\ell=1}^{\Gamma} \|H_{\ell}\| \right )} \right )
\end{equation}
where the operator $H$ is decomposed into a sum of $\Gamma$ terms,
\begin{equation}
\alpha = \sum_{\ell_1,\ell_2,\ldots,\ell_{p+1}=1} \|[H_{\ell_{p+1}},\ldots,[H_{\ell_2},H_{\ell_1}]\cdots]\| \, ,
\end{equation}
and $\Upsilon$ is the number of ``stages" of the formula. For the symmetric Trotter-Suzuki formula, $\Upsilon = 2(5)^{q-1}$ for a TS formula of order $2q$, where $q=1$ and $p=2$ for our case. $\Gamma = 2rD$ and $\alpha$ can be upper bounded by 
\begin{equation}
    \alpha \leq (2rD) 4\max_p (c_{p})^3 \leq 8rD \gamma_{\max}^3 \, . \label{eq:alphanaivebound}
\end{equation} 
Since this is the short time multiplicative error bound for simulating an operator for time $t/r$, we want to bound the resulting error when simulating for large $t$. To this end note that if we have an operator $A$ we approximate by an operator $B$ up to some multiplicative error $m$, then $B = A(I + mC)$ where $C$ is an operator such that $\|C\| \leq 1$, and $m$ is a constant. Then 
\begin{align}
    \|B\|^r &\leq (\|A\|(I + m\|C\|)^r \leq \|A\|^r \left(1 + \sum_{q=1}^r (m\|C\|)^q \binom{r}{q} \right) \nonumber \\
    &\leq \|A\|^r \left(1 + \sum_{q=1}^r m^q \binom{r}{q} \right) \leq \|A\|^r \left(1 + \sum_{q=1}^r \left(\frac{mre}{q}\right)^q \right) \nonumber \\
    &\leq \|A\|^r \left(1 + \sum_{q=1}^r (mre)^q \right) \leq \|A\|^r \left(1 + \frac{mre}{1-mre} \right). 
\end{align}

Therefore the long-time multiplicative error is bounded by $mre/(1-mre)$ and we would like this to be less than some desired error $\epsilon_T > 0$. This implies that we must have $mre \leq \epsilon_T/(1+\epsilon_T) \leq \epsilon_T$. In our context, $A = e^{tH/r}$, $B$ is an approximation to $A$ as given by a Trotter formula, and $m$ is the short-time multiplicative Trotter error bound cited above. 

Using the bound on $m$ and substituting in the parameters relevant for our situation, we have \begin{equation}
mre \leq e\alpha \frac{t^3}{r^2} \exp \left(\frac{4t}{r} \sum_{\ell=1}^{2rD} \|H_{\ell}\|\right) \leq \epsilon_T .
\end{equation} 
If $r \geq \frac{4t}{\ln 2}\sum_{\ell} \|H_{\ell}\|$, then
\begin{equation}
mre \leq 2e\alpha \frac{t^3}{r} \leq 2e \alpha \frac{t^3}{r^2} \leq \epsilon_T
\end{equation} 
which implies that 
\begin{equation}
    r = t \max \left\{\left(\frac{4et\alpha}{\epsilon_T}\right)^{1/2},\  \frac{4}{\ln 2}\sum_{\ell} \|H_{\ell}\| \right\} . \label{eq:PIrbd}
\end{equation}
The former term dominates asymptotically, so we will henceforth take $r \in \Theta \left ( t \left(\frac{4et\alpha}{\epsilon_T}\right)^{1/2} \right)$.

The systematic error in the estimate of the expectation value from the Trotter-Suzuki formula and the finite length Markov chain is at most
\begin{equation}
    \|e^{-\TBG t}\|\epsilon_T + \epsilon_M:=\epsilon_{TM}  .
\end{equation}
The total number of operations needed to draw a single sample from the distribution with bias at most $\epsilon_{TM}$ is from~\eqref{eq:NopsMarkov} in
\begin{equation}
 \mathcal{O}\left(\frac{rD{\log}(1/\delta)}{\gamma_M} \right) = O\left(\frac{\sqrt{\alpha} t^{3/2}D{\log}(f_{\max}/\epsilon_M)}{\sqrt{\epsilon_T}\gamma_M} \right)   .
\end{equation}
Next taking $\epsilon_M = \epsilon_{TM}/2$ and similarly for $\|e^{-\TBG t}\|\epsilon_T$, we have
that the number of operations needed to draw a sample with the required bias is in
\begin{equation}
    \mathcal{O}\left(\frac{\sqrt{\|e^{-\TBG t}\|\alpha} t^{3/2}D{\log}(f_{\max}/\epsilon_{TM})}{\sqrt{\epsilon_{TM}}\gamma_M} \right) .
\end{equation}
Finally, if we set the error $\epsilon_{TMH}$ to be the error also including the bias in the mean estimate of $\dim \ker \Delta_k$ from having $t = \log(1/\epsilon)/\gamma$, we have after choosing both sources of error to be equal that the systematic error can be made less than $\epsilon_{TMH}$ using a number of operations in
\begin{equation}
    \widetilde{\mathcal{O}}\left(\frac{\sqrt{\|e^{-\TBG \log(1/\epsilon_{TMH})/\gamma}\|\alpha} D \log(f_{\max}/\epsilon_{TMH})}{\sqrt{\epsilon_{TMH}}\gamma_M \gamma_{\min}^{3/2}} \right)\subseteq \widetilde{\mathcal{O}}\left(\frac{D\sqrt{\alpha} {\log}(f_{\max}/\epsilon_{TMH})}{\sqrt{\epsilon_{TMH}}\gamma_M \gamma_{\min}^{3/2}} \right) .
\end{equation}

\subsection{Sample bounds}

We finally need to consider the sampling error $\epsilon_S$ that arises from taking only a finite number of samples. Standard probabilistic arguments show that the number of samples $N_S$ needed to achieve a given $\epsilon_S$ scales as $\sigma^2/\epsilon_S^2$.  The mean-squared error $\epsilon^2$ is then
\begin{equation}
    \epsilon^2 = \epsilon_S^2 + \epsilon_{TMH}^2 .
\end{equation}
As before, we choose to make the two contributions to the error equal. There is a final source of complexity that needs to be considered though.  Algorithm~\ref{alg:classical} begins by drawing a valid $(k-1)$-simplex to start at to ensure that we are within the space of interest. This means that we need to randomly draw vertices until we find a vertex that is in a $k$-clique.  The probability of drawing such a simplex is $|{\rm Cl}_k(G)|/|\mathcal{H}_k|$, which is the clique density for the graph. Thus with high probability, a number of samples proportional to the reciprocal of this will be needed.  Each such sample requires clique detection, which is argued in~\sec{clique} scales as $\mathcal{O}(|E|)$ up to logarithmic terms in $k$. This leads us to a cost of
\begin{equation}
  N_{\rm op} \in \widetilde{\mathcal{O}}\left(\frac{|E|d_{k-1}\sigma^2}{|{\rm Cl}_k(G)|\epsilon^2} +\frac{D \sqrt{\alpha} \log(f_{\max}/\epsilon)}{\epsilon^{5/2}} \frac{\sigma^2}{\gamma_M \gamma_{\min}^{3/2}}\right) .\label{eq:initialNopbound}
\end{equation}

We now substitute $\alpha, f_{\max}$ for variables related directly to the properties of the combinatorial Laplacian. We substitute $t \geq \log(1/\epsilon)/\gamma_{\min}$ throughout and drop subdominant logarithmic terms. Firstly, from \eqref{eq:alphanaivebound} we get 
\begin{equation}
    \alpha \leq 8rD \gamma_{\max}^3 \in \widetilde{\Theta} \left( \frac{D t^{3/2} \alpha^{1/2} \gamma^3_{\max}}{\epsilon^{1/2}}\right) \implies \sqrt{\alpha} \in \widetilde{\Theta}\left( D \frac{\gamma_{\max}^3}{\gamma_{\min}^{3/2} \epsilon^{1/2}} \right) . \label{eq:sqrtalphabound}
\end{equation}
This bound on $\sqrt{\alpha}$, Eq.~\eqref{eq:PIrbd}, and the fact that $\TBG$ is positive semi-definite imply
\begin{equation}
    r\in \widetilde{\mathcal{O}}\left(\|e^{-\TBG t/2}\|\sqrt{\alpha} t^{3/2} /\sqrt{\epsilon} \right)= \widetilde{\mathcal{O}}\left(D\gamma_{\max}^3 /\epsilon \gamma_{\min}^3 \right) .
\end{equation}
From \eqref{eq:fmaxbound}, the logarithm scales as
\begin{equation}
    \log\left(\frac{f_{\max}}{\epsilon}\right) \in \widetilde{\Theta}\left( \log(d_{k-1}) + \frac{D^2\gamma_{\max}^3}{\gamma_{\min}^3\epsilon}\right) .
\end{equation}
The prior bound on the variance in \eqref{eq:sigma2} evaluates to
\begin{align}\label{eq:varworst}
        \sigma^2 &\leq \frac{2^{2rD} (1/\epsilon)^{2D\max_{i} |\lambda_{p_i,\Gamma_i}|/\gamma_{\min}}}{d_{k-1}}\nonumber\\
        &\in  \frac{2^{\mathcal{O}(D^2 \kappa^3/\epsilon))} (1/\epsilon)^{2D \kappa}}{d_{k-1}},
\end{align}
where we have defined the following quantity, which is analogous to the condition number for the Dirac operator restricted to $\mathcal{H}_k$, $\kappa = \gamma_{\max}/\gamma_{\min}$ (and neglecting the kernel). Substituting these expressions in Eq.~\eqref{eq:initialNopbound} then implies the number of operations for the algorithm obeys 
\begin{align}\label{eq:NopBdgeneral}
    N_{op} &\in \widetilde{\mathcal{O}}\left( \frac{|E|d_{k-1}\sigma^2}{|{\rm Cl}_k(G)|\epsilon^2}+  \frac{D^4\sigma^2\gamma_{\max}^3}{\gamma_{\min}^3 \epsilon^3 \gamma_M}\left(\log(d_{k-1})D^{-2}+ \frac{\gamma_{\max}^3}{\gamma^{3} \epsilon} \right)\right)\nonumber\\
    &\in \widetilde{O}\left(\frac{\sigma^2}{\epsilon^2}\left(\frac{|E|d_{k-1}}{|{\rm Cl}_k(G)|}+\frac{D^4 }{\gamma_M}\frac{\kappa^3}{\epsilon} \left( \log(d_{k-1}) D^{-2} +  \frac{\kappa^3}{\epsilon} \right)\right)\right).
\end{align}

In the event we assume the worst case bound on the variance, the total number of operations is in
\begin{equation}\label{eq:NopBd}
    N_{\rm op}\in \widetilde{\mathcal{O}} \left(\frac{2^{\mathcal{O}(D^2 \kappa^3/\epsilon))} }{d_{k-1}\epsilon^{2+2D \kappa}}\left(\frac{|E|d_{k-1}}{|{\rm Cl}_k(G)|}+\frac{D^4 }{\gamma_M}\frac{\kappa^3}{\epsilon} \left( \log(d_{k-1}) D^{-2} +  \frac{\kappa^3}{\epsilon} \right)\right)\right) .
\end{equation}
This shows that under worst case scenario scaling for the variance, our algorithm is efficient if $D\kappa$ is poly-logarithmic in $n$, $\kappa/\epsilon$ is a constant and $\gamma_M^{-1}$ and the inverse density of cliques are polynomial in $n$.

While the above restrictions on the situations where the classical randomized algorithm is efficient are significant, they do imply that the TDA algorithm can be efficient even in cases where $d_{k-1} = \binom{n}{k}$ is exponentially large provided the graph is clique-dense.  This possibility is not obvious if one only compares to classical algorithms like diagonalization, which scales polynomially with the dimension. 

Finally, the number of operations varies in~\eqref{eq:NopBdgeneral} with the variance of the path integrals which we upper bound with an exponential in $D$.  While this scaling may seem prohibitive in the case where the graph is nearly complete, the variance bound in this case is extremely loose and using a particular bound designed for this scenario yields much better scaling as we will see later.

\subsection{Analysis}

In order to give bounds on the sparsity of the combinatorial Laplacian, we first need to define a few terms. Let $K$ be a simplicial complex with $N$ vertices. The \textbf{up-degree} of a $k$-simplex $\sigma \in K$, denoted by $\text{deg}_U(\sigma)$, is the number of $k+1$ simplices in $K$ that $\sigma$ is in the boundary of. The \textbf{lower-degree} or \textbf{down-degree} of $\sigma$, denoted by $\deg_L(\sigma)$, is the number of $k-1$ simplices in $K$ that are in the boundary of $\sigma$. If two $k$-simplices $\sigma_1,\sigma_2 \in K$ both contain a $k-1$ simplex in their boundary, they are said to be \textbf{lower adjacent}. If $\sigma_1, \sigma_2$ are in the boundary of a $k+1$ simplex, they are said to be \textbf{upper adjacent}. Lemma 3.2.4 of \cite{goldberg2002} shows that if $\sigma_1, \sigma_2$ are distinct and lower adjacent, their common $(k-1)$-simplex is $\sigma_1 \cap \sigma_2$ and is unique if it exists. Similarly, Lemma 3.2.2 of \cite{goldberg2002} shows that if $\sigma_1$ and $\sigma_2$ are upper adjacent, their common $(k+1)$-simplex is unique. 

The down degree of a $k$-simplex $\sigma \in K$ is always $k+1$ for any simplicial complex simply because any $k$-simplex contains $\binom{k+1}{k} = k+1$ simplices of dimension $k-1$ in its boundary. Its up-degree is less trivial to determine, but can be bounded as follows. 

\begin{proposition}
{Let $K$ be a simplicial complex with $N$ vertices and let $\sigma$ be a $k$-simplex in $K$. The up-degree of $\sigma$ is bounded by $\min\{N-k-1, d\}$, where $d$ is maximum (up)-degree of all the vertices in $K$.}    
\end{proposition}

\begin{proof}
The most simple argument is that the up-degree is bounded by the number of possible ways to extend a $k$-simplex to a $k+1$ simplex by adding another point. Since we have $N - k - 1$ other points to choose from, the up-degree is bounded by this quantity. 

A more rigorous argument is to consider the largest eigenvalue of the combinatorial Laplacian $\Delta_k = \partial_k^{\dag} \partial_{k} + \partial_{k+1}\partial^{\dag}_{k+1}$ has the bound $\lambda_{\max}(\Delta_k) \leq N$ \cite{duval2002shifted}. The diagonal matrix elements of $\Delta_k$ are given in Theorem 3.3.4 of \cite{goldberg2002} as $\deg_U(\sigma_i) + k + 1$ for all the $k$-simplices $\sigma_i \in K$ when $k > 0$. Since $\Delta_k$ is a real symmetric matrix, the {Courant-Fischer theorem \cite{horn_johnson_2012}} and the preceding bound on its largest eigenvalue together imply $\deg_U(\sigma) + k + 1 \leq N$, which in turn shows $\deg_U(\sigma) \leq N - k -1$.  

Another trivial upper bound is simply given by the maximal (up) degree $d$ of any vertex $v$ in the simplicial complex $K$. This is because the vertex added to a $k$-simplex to turn it into a $(k+1)$-simplex has to be connected to all vertices in the $k$-simplex. The number of vertices connected to all the vertices in the $k$-simplex is upper bounded by $d$. Hence we may take $\deg_U(\sigma) \leq \min\{N-k-1,\ d\}.$\end{proof}

\begin{proposition}
{$\partial_k$ has row-sparsity equal to $\min\{N-k-1, d\}$ and column sparsity equal to $k+1$. $\Delta_k$ has row and column sparsity bounded by $(k+1)(N-k-1)$.} 
\label{cor:BoundarySparsity}
\end{proposition}

\begin{proof}
With the same notation as above, $\partial_k$ acts on the vector space of $k$-simplices in $K$ with the standard basis vectors corresponding to the $k$-simplices themselves. It has column sparsity equal to the down-degree of any $k$-simplex in $K$, i.e.\ $k+1$, because $\partial_k$ has non-zero elements in a column corresponding to a fixed $k$-simplex only when a $k-1$-simplex is the boundary of that $k$-simplex. 

Its row sparsity however is given by the largest up-degree of all $k$-simplices in $K$ and from the preceding discussion is bounded by $\min\{N-k-1,\ d\}$. As $\Delta_k$ is Hermitian, the column and row sparsity of $\Delta_k$ coincide and equals the maximum of the number of $k$-simplices $c'$ such that $c \cap c'$ is a $(k-1)$-simplex and $c \cup c'$ is \textit{not} a $(k+1)$-simplex for all $k$-simplices $c$ (see Theorem 3.3.4 of \cite{goldberg2002}). A trivial upper bound for this is $(k+1)(N - k -1)$ (note that this is precisely the product of the bounds given for the up and lower degrees earlier). This is because in order to construct for a given $c$ a $c'$ as above, we can remove any of the $k+1$ vertices from $c$ and add any of the $(N-k-1)$ vertices not in $c$ to form $c'$. 
\end{proof}

Note however that tighter bounds can be achieved by also using the largest degree $d$ of a vertex, e.g.\ another upper bound on the sparsity is $\mathcal{O}(kd)$. Thus we can set $D = O(k^2d^2)$ or $D = O(k^2(N-k)^2)$ for the number of terms $D$ in the one-sparse decomposition of the combinatorial Laplacian depending on if $d$ or $N-k-1$ is smaller (see the discussing preceding \eqref{eq:trotterboundarydecomp}).

We now consider a few cases where our algorithm can run efficiently. Our first example will be the extreme case where the input to the algorithm is the completely disconnected graph on $N$ points. The only non-zero Betti number of the associated clique complex in this case is $\beta_0 = N$. Since $\dim \ker (\Delta_0) = \beta_0 = N$, $\Delta_0$ is the $N \times N$ zero matrix and $\kappa$ is undefined. We therefore cannot directly use the preceding asymptotic expressions for the number of samples required. The analysis is simple if we refer to the general formula for $\sigma^2$ in \eqref{eq:sigma2} however. 

\begin{proposition}
{Let $K$ be the completely disconnected graph on $N$ vertices. The variance of the path-integral Monte Carlo sampling procedure for the Betti numbers of $K$ is 0.} 
\end{proposition}

\begin{proof}
    In this case, $\lambda = 0$ for all the eigenvalues of $\Delta_0$. We can choose the standard basis vectors for the $N$-dimensional vector space of $0$-simplices, where each basis vector corresponds to a $0$-simplex, as the eigenvectors of $\Delta_0$. Then $W(\Gamma)$ is merely a product of Kronecker delta functions. There are precisely $N$ valid paths (loops) for the Markov chain corresponding to each of the $N$ vertices, so $\sum_{\Gamma \in S_{\Gamma}} W(\Gamma) = \sum_{\Gamma \in S_{\Gamma}} |W(\Gamma)|^2 = \sum_{\Gamma\in S_\Gamma} = N$. The variance then reduces to $\sigma^2 = \frac{N^2}{d_0^2} - \frac{N^2}{d_0^2} = 0.$
\end{proof}

Even though the variance is precisely 0 for this case, we do not have a priori knowledge of the structure of the input graph from the standpoint of this algorithm, aside from an upper bound on the degree of the vertices. We only possess access to it via oracle queries that verify whether a collection of vertices forms a $k$-clique. Therefore, we cannot conclude anything definite about the Betti numbers with only a single sample in this situation as the zero variance result might suggest, unless given an additional promise that the \textit{least} upper bound is a fixed value or 0.

Another simple case is when the input is a complete graph on $N = n+1$ points. 

\begin{proposition}
{Let $K$ be the clique complex of the complete graph on $N = n +1$ vertices. The variance of the path-integral Monte Carlo sampling procedure for the Betti numbers of $K$ is 0.} 
\end{proposition}

\begin{proof}
The clique complex of this graph corresponds to an $n$-simplex. Again from Theorem 3.3.4 of \cite{goldberg2002}, $(\Delta_k)_{ii} = \deg_U(\sigma_i) + k + 1$ and the off-diagonal entries are either 0 or $\pm 1$ if $k > 0$. We argue that in this case, all the off-diagonal entries are 0. From the same theorem, it suffices to show that for a fixed $k$-simplex $\sigma_i$, the sum of the number of other upper adjacent simplices and other non-lower adjacent simplices must equal $\binom{n+1}{k+1} - 1$ (since $\Delta_k$ is a $\binom{n+1}{k+1} \times \binom{n+1}{k+1}$ matrix). 

We first calculate the up-degree of any $k$-simplex in our $n$-simplex. The number of $(k+1)$-simplices in an $n$-simplex is $\binom{n+1}{k+2}$. The number of $k$-simplices in a $(k+1)$-simplex is $k+2$. The number of $k$-simplices in an $n$-simplex is $\binom{n+1}{k+1}$. Intuitively, we can determine the up-degree by finding the number of $k$-simplices contained in the boundary of all $(k+1)$-simplices and dividing by the actual number of $k$-simplices in an $n$-simplex to account for over-counting. This quantity is precisely $$\frac{(k+2)\binom{n+1}{k+2}}{\binom{n+1}{k+1}} = n-k$$ and we thus get that the diagonal entries of $\Delta_k$ are all equal to $n+1$ when $k > 0$ (note that this saturates the bound on the up-degree of a $k$-simplex in any simplicial complex given earlier with $N = n +1$). 

Now fix a particular $k$-simplex $\sigma$ in the $n$-simplex. We want to count the number of \textit{other} $k$-simplices that are \textit{not} lower adjacent to $\sigma$. $\sigma$ has $k+1$ simplices of dimension $k-1$ in its boundary. From the above result, each of these $(k-1)$-simplices has up-degree $n-(k-1) = n-k+1$. Then the number of potential lower adjacent $k$-simplices is $(k+1)(n-k+1)$. But by the uniqueness of common upper simplices, we have counted $\sigma$ $k+1$ times, one for each of the $(k-1)$-simplices in the boundary of $\sigma$. Thus the number of \textit{other} lower adjacent simplices is $(k+1)(n-k+1) - (k+1) = (k+1)(n-k)$. Thus the number of \textit{other} $k$-simplices \textit{not} lower adjacent to $\sigma$ is $\binom{n+1}{k+1} - (k+1)(n-k) - 1$. On the other hand, the number of \textit{other} $k$-simplices that are upper adjacent to $\sigma$ is given by multiplying the up-degree $n-k$ of $\sigma$, which gives the number of $k+1$ simplices that are upper-adjacent to $\sigma$, by the number of $k$-simplices in each $k+1$-simplex which is $k+2$. But this overcounts $\sigma$ by $n-k$, so $(n-k)(k+2) - (n-k) = (n-k)(k+1)$. The sum of this and $\binom{n+1}{k+1} - (k+1)(n-k) - 1$ is clearly $\binom{n+1}{k+1} - 1$, so the only non-zero entries in $\Delta_k$ are the diagonal ones when $k > 0$. 

When $k = 0$, Theorem 3.3.4 in \cite{goldberg2002} shows the diagonal entries of $\Delta_0$ are the degree of the vertices in our $n$-simplex, which is precisely $n$. It also implies all the off-diagonal entries are $0$ since every vertex is upper-adjacent in an $n$-simplex. Thus in either the $k = 0$ or $k > 0$ case, $\Delta_k$ is proportional to the identity. We can then set $D = r = 1$ in \eqref{eq:sigma2} and the number of valid paths summed over is precisely $d_k$. By an analogous reasoning to the completely disconnected case, $\sum_{\Gamma \in S_{\Gamma}} W(\Gamma) = \sum_{\Gamma \in S_{\Gamma}} |W(\Gamma)|^2 = \sum_{\Gamma\in S_\Gamma} = d_k$ and $$\sigma^2 = e^{-2t\lambda}\left(\frac{d_k^2}{d_k^2} - \frac{d_k^2}{d_k^2}\right) = 0,$$ where $\lambda = n+1$ when $k > 0$ and $\lambda = n$ when $k = 0$. \end{proof}

Analogous arguments on needing a promise on the least upper bound on the degree of the vertices hold as in the case of the completely disconnected graph above. 

More can be said of $\kappa$ in restricted cases. For instance, it is known that orientable simplicial complexes of dimension $d \leq 2$ have $\kappa \in \mathcal{O}(n_k^{2})$, where $n_k$ denotes the number of $k$-simplices and $0 \leq k \leq 2$, and it is further conjectured that $\kappa \in \mathcal{O}(n_k^{2/d})$ in most cases~\cite{friedman:betti}. This latter bound is quite favourable for high-dimensional simplices, i.e.\ when $d \sim n$ and with the approximation that $n_k \sim 2^n$. Unfortunately, it is difficult to give tight bounds on $\gamma_M$ and $\sigma^2$ in \eqref{eq:NopBdgeneral} for arbitrary simplicial complexes. {Though the Perron-Frobenius theorem applied to stochastic matrices implies $0 < \gamma_M \leq 1$ \cite{horn_johnson_2012}, its scaling with the other parameters of relevance are unknown}. 

While we do not ultimately expect this classical TDA method to be efficient generically, the above discussions show there exist sufficient conditions under which it can be used to extract Betti numbers in polynomial time without suffering from a generic exponential dependence on the number of data points characteristic of other classical TDA algorithms. Furthermore, it creates uncertainty about the necessary and sufficient conditions for an exponential advantage for quantum TDA as the cases where this algorithm has an exponential advantage relative to the deterministic classical algorithm are analogous to those for the quantum algorithm. This means that an identification of clear cases where exponential advantage is likely remains an important open problem within the domain of quantum algorithms for TDA.

\end{document}